\documentclass[twoside,11pt]{article}

\usepackage{blindtext}

\usepackage{appendix}
\usepackage{amsmath}
\usepackage{amsthm}
\usepackage{hyperref}
\usepackage{cleveref}
\usepackage{makecell} %
\usepackage[nohyperref]{jmlr2e}

\RenewDocumentEnvironment{proof}{o}{%
    \IfNoValueTF{#1}{\par\noindent{\bf Proof\ }}{\par\noindent{\bf Proof #1}}}{%
    \hfill\BlackBox\\[2mm]}

\newtheorem{assumption}{Assumption}
\crefname{assumption}{assumption}{assumptions}
\crefname{appendix}{supplement}{supplements}

\usepackage{tikz}
\usepackage{pgfplots}
\usepackage{bibunits}
\usepgfplotslibrary{fillbetween} %
\usetikzlibrary{matrix,arrows.meta}
\usetikzlibrary{external}
\newcommand{\SMref}[1]{\Cref{#1}}

\graphicspath{{./figures/}}

\usepackage{macros}

\usepackage{lastpage}
\ShortHeadings{Bayesian Level Set Clustering}{David Buch, Miheer Dewaskar, and David Dunson}
\firstpageno{1}

\begin{document}

\title{Bayesian Level Set Clustering}

\author{\name David Buch\thanks{Joint first authors; MD is the corresponding author.} \email davidbuch42@gmail.com \\
       \addr Two Sigma\\ New York, New York 10013, USA
       \AND
       \name Miheer Dewaskar\footnotemark[1] \email mdewaskar@unm.edu \\
       \addr Department of Mathematics and Statistics\\ University of New Mexico\\ Albuquerque,  New Mexico 87106, USA
       \AND
       \name David B. Dunson \email dunson@duke.edu\\
       \addr Department of Statistical Science\\ Duke University\\Durham, North Carolina 27708, USA}

\editor{Debdeep Pati}

\maketitle

\begin{abstract}%
Classically, Bayesian clustering
interprets each component of a mixture model as a cluster. The inferred clustering posterior is highly sensitive to any inaccuracies in the kernel within each component. %
As this kernel is made more flexible, problems arise in identifying the underlying clusters in the data. To address this pitfall, this article proposes a fundamentally different approach to Bayesian clustering that decouples the problems of clustering and flexible modeling of the data density $f$. Starting with an arbitrary Bayesian model for $f$ and a loss function for defining clusters based on $f$, we develop a Bayesian decision-theoretic framework for density-based clustering. Within this framework, we develop a Bayesian level set clustering method to cluster data into connected components of a level set of $f$. 
We provide  theoretical support, including clustering consistency, and highlight  performance in a variety of simulated examples. An application to astronomical data illustrates improvements over 
 the popular \texttt{DBSCAN} algorithm in terms of accuracy, insensitivity to tuning parameters, and providing uncertainty quantification.
\end{abstract}


\begin{bibunit}[apalike]
\section{Introduction}
\label{s:introduction}
In the Bayesian literature, when clustering is the goal, it is standard practice to model the data as arising from a mixture of unimodal probability distributions \citep{lau2007bayesian, wade2018bayesian}. The observations are then grouped according to their association with a mixture component. 
Bayesian clustering has potential advantages over algorithmic and frequentist approaches, providing natural hierarchical modeling, uncertainty quantification, and the ability to incorporate prior information \citep{wade2023bayesian}. 
However, limitations appear in trying to apply the mixture model framework when clusters cannot be well represented by simple parametric kernels. Even when clusters are \textit{nearly} examples of simple parametric components, mixture model-based clustering can be brittle and result in  \textit{cluster splitting} \citep{miller2018robust, cai2021finite,chaumeny2022bayesian}.  
A potential solution is to use more flexible kernels \citep{malsiner2017identifying}. However, as the components are made more flexible, mixture models become difficult to fit and identify, since the multitude of reasonable models for a dataset tends to explode as the flexibility of the pieces increases \citep{ho2016convergence, ho2019singularity}. 

\begin{figure}%
	\centering
	\subfloat[\centering Bayesian Model-Based Clustering Point Estimate]{{\includegraphics[width=0.45\textwidth]{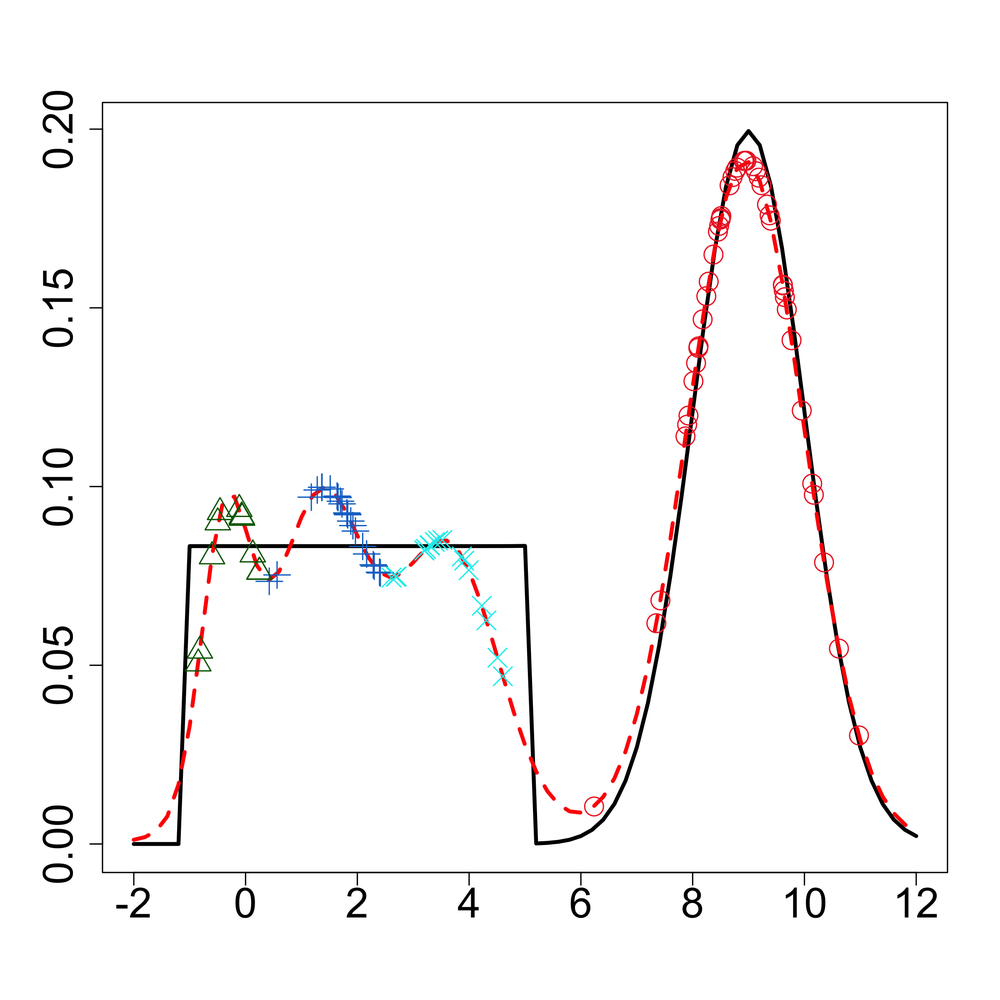} }}%
	\qquad
	\subfloat[\centering Bayesian Level Set Clustering Point Estimate]{{\includegraphics[width=0.45\textwidth]{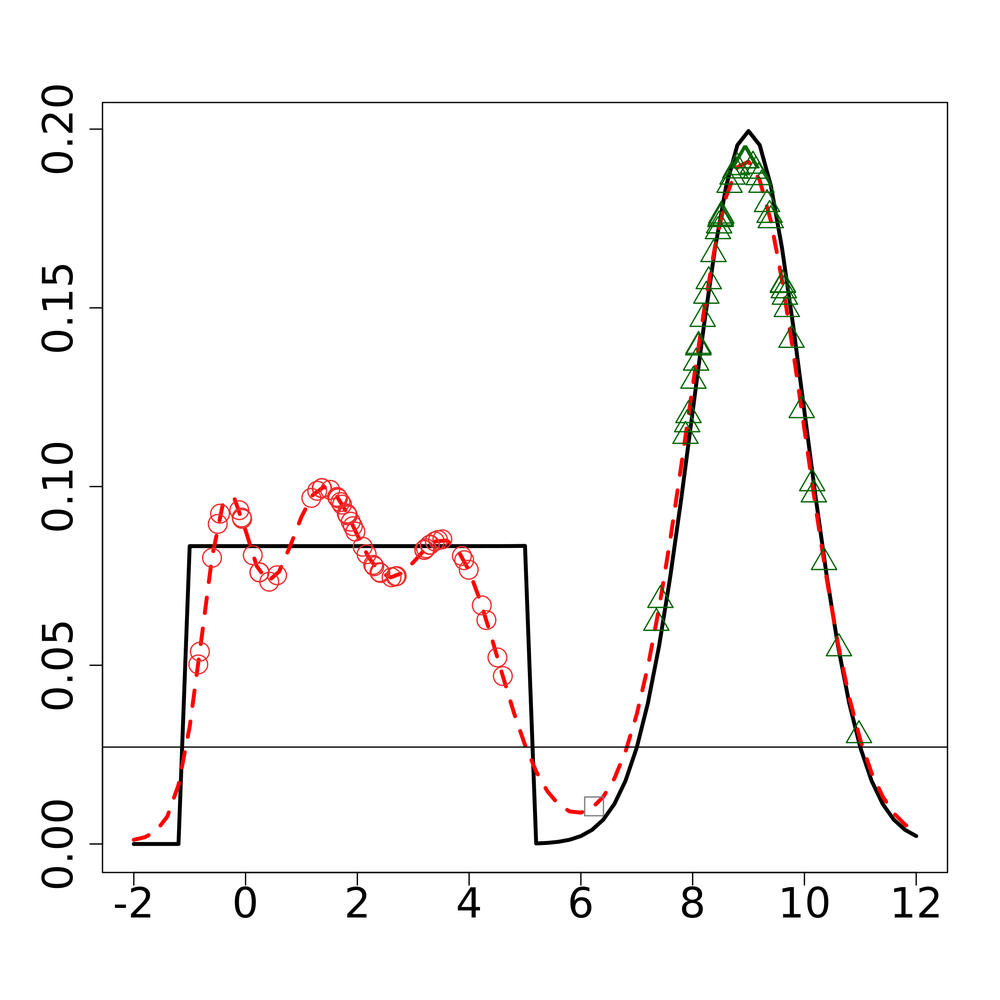} }}%
	\caption{We see the cluster splitting phenomenon among the clusters obtained (left) by fitting a Dirichlet process mixture of Gaussian 
    prior and finding the partition that minimizes expected VI loss under the posterior. Our Bayesian level set clustering (\texttt{BALLET}) point estimate based on the same prior (right) does not suffer from this phenomenon, despite the obvious bias in the posterior expectation of the density caused by the poor choice of prior distribution. We display a random subsample of the data in both plots, with their y-coordinates set to the expectation of the density at their locations, and with cluster assignments reflected by the color and shape of the points. The dashed red line is the expected density under the posterior. The solid line shown in black is the true data-generating density. The density level $\lambda=0.028$ denoted by the horizontal line (right) was  selected using an elbow heuristic in \Cref{s:ballet-tuning-parameter} (see \Cref{fig:illustration-elbow-selection}).} %
	\label{fig:mix-uniform-normal}
\end{figure}

Rather than avoid Bayesian clustering when the mixture approach fails, we propose decoupling the problems of modeling the data density and inferring clusters. Suppose that the data are drawn from the sample space $\cX$, and denote by $\Den$ the density space on $\cX$. Then, letting $\ClustSpace[\cX]$ refer to the space of all possible partitions of $\cX$, we can define functions $\Psi: \Den \to \ClustSpace[\cX]$ that map from densities on $\cX$ to partitions of $\cX$. In the example in Figure \ref{fig:mix-uniform-normal} (b), $\Psi(f)$ was chosen as the partition of $\cX$ induced by the connected components of $\{x \in \cX : f(x) \geq \lambda\}$ at the $\lambda=0.028$ level. Partitions of the sample space determine well-defined clusterings since, for any sample $\cX_n = \{x_1, \ldots, x_n\} \subseteq \cX$, a partition of $\cX$ induces a partition on $\cX_n$. For a particular $\Psi$ and data set $\cX_n$, we denote maps from the densities on $\cX$ to the partition on $\cX_n$ induced by $\Psi$ with the lower case $\psi : \Den \to \ClustSpace[\cX_n]$. Here we have suppressed the dependence on the sample $\cX_n$ to simplify the notation. 
 
Let $D\{\psi(f), \bC\}$ denote the loss for clustering $\bC \in \ClustSpace[\cX_n]$ relative to clustering $\psi(f) \in \ClustSpace[\cX_n]$. If $f_0$ is the true data-generating density, then the target clustering is $\bC_0=\psi(f_0)$. In practice $f_0$ is unknown, so we represent uncertainty in the unknown density using a Bayesian posterior $f \sim P_M(\cdot|\cX_n)$ based on the model $M$. This allows us to define a Bayesian decision-theoretic estimator $\cbayes$, obtained by minimizing the expected posterior loss: $\cbayes = \argmin_{\bC \in \ClustSpace[\cX_n]} E_{f \sim P_M(\cdot|\cX_n)}[D\{\psi(f), \bC\}]$, and to quantify the uncertainty in the clustering.

There is a substantial non-Bayesian literature on clustering based on the data-generating density $f$
\citep{menardi2016modal, campello2020density, bhattacharjee2021survey}. In this article, applying our decision-theoretic Bayesian paradigm for density-based clustering, we propose a new framework for Bayesian level set clustering. Level set clustering
\citep{rinaldo2010generalized, sriperumbudur2012consistency, jiang2017density, jang2019dbscan} is a popular approach that groups data points that fall into the same high-density region, while allowing these regions to have complex shapes. %
Our Bayesian approach has substantial advantages over current algorithmic approaches, such as \texttt{DBSCAN} \citep{ester1996density,schubert2017dbscan}, which we will illustrate in various examples. Advantages include accuracy, less sensitivity to tuning parameters, and uncertainty quantification in clustering. 

Our approach starts with the posterior under any nonparametric Bayesian model for $f$ as the input, defines a loss function appropriate for level set clustering, and develops efficient algorithms for producing Bayes clustering estimates, while also providing a characterization of uncertainty in clustering. We develop supporting theory and demonstrate advantages over model-based and algorithmic level set clustering in various applications. The code for implementing our methodology is available at \url{https://github.com/davidbuch/ballet_article} and can be applied to
data $\cX_n$ and samples $f^{(1)}, \dots, f^{(s)}$ from the posterior distribution of $f$ under any Bayesian model.

As a teaser motivating Bayesian level set clustering over a mixture-based approach, Figure
\ref{fig:mix-uniform-normal} shows clusters produced by (a) a traditional Bayesian clustering approach and (b) our proposed approach. Here, the black line is the true density $f_0$ and both methods rely on fitting the same Dirichlet process mixture of Gaussians to the data to obtain a posterior for $f$. Although the use of Gaussian kernels leads to a noticeable bias in density estimation in the left mode of $f_0$, our inferred level set clusters, which depend on the posterior distribution of the level set $\{x: f(x) \geq \lambda \}$ for our chosen level $\lambda$, are not affected by this. In contrast, an approach that equates clusters to mixture components sub-divides the uniform component into several subclusters. An interesting aspect of level set clustering is no attempt is made to cluster data points falling in low density regions; see \Cref{fig:synthdata-pe} for an example motivated by cosmology.

\subsection{Contributions} 

The closest literature relevant to our work is that of Bayesian estimation of level sets of densities studied by \cite{gayraud2005rates,gayraud2007consistency} and the results in \cite{li2021posterior} on posterior contraction and credible regions for level curves. The frequentist estimation of level curves using bootstrap to characterize uncertainty is studied in \cite{chen2016density}. Compared to the previous work on level set estimation, here we develop a practical method to compute a consistent Bayesian estimator of the induced clustering of the data and describe the associated uncertainty. Obtaining Bayesian clustering approaches that have appealing frequentist asymptotic properties is challenging under the predominant mixture model approach, particularly without making unrealistic assumptions such as correct kernel specification. Consequently, new Bayesian clustering methodologies based on the merging of components of an overfitted mixture of Gaussians \citep{dombowsky2023bayesian,aragam2020identifiability}, the use of repulsive priors in the cluster means \citep{petralia2012repulsive,xie2020bayesian,beraha2022mcmc} and the addition of entropic regularization \citep{franzolini2024entropy}, have been proposed to improve the reliability of Bayesian clustering. With similar motivation, here we propose a Bayesian framework for \emph{density-based clustering}  \citep{menardi2016modal,campello2020density,bhattacharjee2021survey} that is consistent under suitable assumptions (\Cref{thm:consistency}). We show how the standard Bayesian decision-theoretic clustering machinery can be adapted to handle density-based clustering by modifying the loss function \eqref{eq:density-based-estimator}. Focusing on level set clustering, we leverage the current algorithmic and theoretical understanding \citep{schubert2017dbscan,sriperumbudur2012consistency} to implement our Bayesian level set clustering methodology \texttt{BALLET} and establish its consistency (\Cref{cor:ballet-consistency}). Finally, in illustrating the application of \texttt{BALLET} to various datasets, we discuss practical strategies to choose the level $\lambda$ (\Cref{s:ballet-tuning-parameter}) and highlight the advantages offered by describing the clustering uncertainty  associated with \texttt{BALLET}  in  a comprehensive analysis of astronomical sky survey data (\Cref{s:real-data-analysis}).

\section{Bayesian Level Set Clustering Methodology}
\label{s:methods}

\subsection{Level Set Clusters and Sub-partitions}
\label{ss:level-set-clusters}
We start by expanding on the notational conventions of Section \ref{s:introduction}. Suppose that our data $\mathcal{X}_n = \{x_1, \ldots, x_n\}$ are drawn independently from an unknown density $f_0 \in \Den$ on the sample space $\cX$ taken to be $\R^d$ in much of this article, where $\Den$ denotes the space of densities on $\cX$ with respect to the Lebesgue measure. Let $S_{\lambda,f_0} \doteq \{x \in \cX : f_0(x) \geq \lambda\}$ denote the $\lambda$ level set of $f_0$, and temporarily let $W_{1}^{f_0}, \dots, W_{k^*}^{f_0}$ denote the \textit{topologically connected components} of $S_{\lambda, f_0}$. In \Cref{fig:connected-components-fig}, $S_{\lambda,f_0}$ is the colored region on the $x$-axis, with colors corresponding to the different choices of $\lambda$ indicated by the dashed lines. When $d=1$, this region will either be a single interval $S_{\lambda,f_0} = W_{1}^{f_0}$ with $k^*=1$, or more generally, be a union $S_{\lambda,f_0} = W_{1}^{f_0} \cup \cdots \cup  W_{k^*}^{f_0}$ of $k^* \in \{0,1,2,\ldots\}$ disjoint intervals. The \emph{level set clustering} $\bC_{0} =\psi_{\lambda}(f_0)$ of the data points $\cX_n$ associated with $f_0$ is the collection $\bC_{0} = \{C_1^{f_0}, \dots, C_k^{f_0}\}$ of $k \leq k^*$ non-empty sets in $\{W_1^{f_0} \cap \cX_n, \dots, W_{k^*}^{f_0} \cap \cX_n\}$. For instance, the level set clustering corresponding to $\lambda = 0.1$ in \Cref{fig:connected-components-fig} is a grouping of data points $\cX_n$ (not shown) based on whether they fall in a common blue interval or not. Data points that fall outside all of the blue intervals will be called \emph{noise points}.

A level set clustering $\bC = \{C_1, \ldots, C_k\}$ of $\cX_n$ is a \textit{sub-partition}, since $C_i \cap C_j = \emptyset$ for all $i \neq j$ and  $\cup_{i=1}^k C_i \subseteq \cX_n$ but, unlike regular partitions, the presence of noise points not assigned to any cluster can lead to $\cup_{i=1}^k C_i \neq \cX_n$. We call the observations in   
$A = \cup_{i=1}^k C_i$ as \textit{active} or \textit{core points}, while the remaining observations $I = \cX_n \setminus A$ are \textit{inactive} or \textit{noise points}. In Figure \ref{fig:mix-uniform-normal}(b) a noise point is shown in gray. 
Every sub-partition of size $k$ with some noise points can be mapped to a unique partition of size $k + 1$, where the extra set in the partition consists of the noise points.  However, this mapping is not one-to-one because the information on the identity of the noise cluster is lost (see example at the beginning of \Cref{s:ia-loss-function}). Instead, to preserve information about noise points, we explicitly work with the non-standard setup of regarding a clustering as a sub-partition rather than a partition.
To this end, we repurpose the notation $\ClustSpace[\cX_n]$ to denote the space of all \textit{sub-partitions} of $\cX_n$. Note that this is a strict expansion since $\ClustSpace[\cX_n]$ also contains all partitions of $\cX_n$. 

\begin{figure}
    \centering
    \begin{tikzpicture}
		\begin{axis}[
			width=12cm, height=5cm,
			domain=-6:6, samples=2000,
			xmin=-6, xmax=6, ymin=-0.06, ymax=0.35,
			axis x line=middle, axis y line=left,
			xlabel={$x$}, ylabel={$f(x)$},
			clip=false,
			legend style={at={(0.5,1.1)}, anchor=south, legend columns=-1} %
			]
			
			\def\sigma{0.7}
			\def\muA{-3}
			\def\muB{0.7}
			\def\muC{3}
			\def\wA{0.5}
			\def\wB{0.25}
			\def\wC{0.25}
			
			\def\levone{0.15}
			\def\levtwo{0.1}
			\def\levthree{0.05}
			
			\pgfmathdeclarefunction{mixpdf}{1}{%
				\pgfmathparse{%
					(\wA)/(sqrt(2*pi)*\sigma)*exp(-((#1-\muA)^2)/(2*\sigma^2)) +
					(\wB)/(sqrt(2*pi)*\sigma)*exp(-((#1-\muB)^2)/(2*\sigma^2)) +
					(\wC)/(sqrt(2*pi)*\sigma)*exp(-((#1-\muC)^2)/(2*\sigma^2))
				}%
			}
			
			\addplot[name path=mixture, very thick, domain=-6:6] {mixpdf(x)};
			\addplot[gray,dashed,thick, green] coordinates {(-6,\levone) (6,\levone)};
			\addplot[gray,dashed,thick, blue] coordinates {(-6,\levtwo) (6,\levtwo)};
			\addplot[gray,dashdotted,thick, orange] coordinates {(-6,\levthree) (6,\levthree)};
			\addplot[line width=3pt, green, domain=-6:6, unbounded coords=jump]
			({x}, { (mixpdf(x) >= \levone) ? 0.02 : nan });
			
			\addplot[line width=3pt, blue, domain=-6:6, unbounded coords=jump]
			({x}, { (mixpdf(x) >= \levtwo) ? 0 : nan });
			\addplot[line width=3pt, orange, domain=-6:6, unbounded coords=jump]
			({x}, { (mixpdf(x) >= \levthree) ? -0.02 : nan });
		\end{axis}
\end{tikzpicture}
    \caption{Topological connected components of the level set $\{x : f_0(x) \geq \lambda \}$ for a mixture $f_0$ (black curve) of Gaussians based on three colored choices for the level $\lambda$.  Changing $\lambda$ can result in discovery of anywhere from zero to three components (clusters). }
    \label{fig:connected-components-fig}
\end{figure}

\subsection{On the interpretation of level set clusters and the choice of level $\lambda$}

Level set clustering is primarily meant to discover connected regions of high (population) density separated by regions of low density, and the parameter $\lambda$ determines what `high' means here. While a reasonable choice of $\lambda$ may be apparent in certain applications (see \Cref{s:real-data-analysis}), we now discuss strategies from the literature when this is not the case. 

When the clusters are expected to be well-separated from each other (e.g.~\Cref{fig:toy-challenge-dpmm,fig:mix-uniform-normal}), simple strategies to tune $\lambda$ based on elbow plots \citep{ester1996density} and deciding on a small fraction ($\nu = \int f(x) 1_{\{f(x) < \lambda\}} dx$) of noise points in advance \citep{cuevas2001further} are useful and robust.  See \Cref{s:ballet-tuning-parameter} for our implementation.

In general however, as seen in \Cref{fig:connected-components-fig}, care is needed to select the level $\lambda$ and in some cases a single appropriate $\lambda$ does not exist (see \Cref{fig:level-set-pathological} and \cite{menardi2016modal,campello2020density}). In such scenarios, one should study the \emph{cluster-tree} \cite{campello2015hierarchical,wang2019dbscan,steinwart2023adaptive} obtained by running level set clustering across a range of values of $\lambda > 0$. It is common to visualize \citep{clustree} and process \citep{campello2015hierarchical,scrucca2016identifying} this tree to extract clusters that remain stable across a range of values of $\lambda$. This motivates our persistent clustering implementation in \Cref{s:persistent-clustering}.

\subsection{Decision-Theoretic Framework}
\label{ss:decision-theoretic-framework}

We focus on finding the sub-partition of data $\cX_n$ associated with the connected components of $S_\lambda$. We let $\psilam: \Den \mapsto \ClustSpace[\cX_n]$ be the \textit{level-$\lambda$ clustering function}, by which we mean that $\psilam(f)$ returns the sub-partition $\bC$ of $\cX_n$ associated with the level-$\lambda$ connected components of $f$.

We start by choosing a Bayesian model $M$ for the unknown density $f$. Examples of $M$ include not only kernel mixture models but also Bayesian nonparametric approaches that do not involve a latent clustering structure, such as Polya trees \citep{lavine1992some, ma2017adaptive} and logistic Gaussian processes \citep{lenk1991towards, tokdar2007towards}. Under $M$, we obtain a posterior distribution $P_M(f|\cX_n)$ for the unknown density of the data. This also induces a posterior on the $\lambda$ level set of $f$. Based on this posterior, we define $\widehat{\psi}_{\lambda,M}$ as an estimator of $\psi_{\lambda}(f_0)$.

Let $D\{\psilam(f), \bC\}$ denote a loss function measuring the quality of sub-partition $\bC$ relative to the ground truth $\psilam(f)$. The Bayes estimator \cite[e.g.][Section 4.4.1]{berger2013statistical} 
of the sub-partition then corresponds to the value that minimizes the expectation of the loss under the posterior of $f$: 
\begin{equation}
	\widehat{\psi}_{\lambda, M}(\cX_n) = \argmin_{\bC \in \ClustSpace[\cX_n]} \Epost[D\{\psilam(f), \bC\}]. \label{eq:psihat}
\end{equation}
In practice, we use a Monte Carlo approximation based on samples $f^{(1)}, \dots, f^{(S)}$ from $P_M(f|\cX_n)$: $\widehat{\psi}_{\lambda, M}(\cX_n) \approx \argmin_{\bC \in \ClustSpace[\cX_n]} \sum_{s = 1}^S D\{\psilam(f^{(s)}), \bC\}$.

Three major roadblocks stand in the way of calculating this estimator. First, evaluating $\psilam(f^{(s)})$
is problematic, as identifying connected components of level sets of $f^{(s)}$ is extremely costly if the data are in even a moderately high-dimensional space. Instead, we will use a surrogate clustering function $\psitilde_\lambda$, which approximates the true clustering function and is more tractable. We will discuss this in more detail in Section \ref{ss:surrogate-clustering-function}.

The second roadblock is the fact that we must design an appropriate loss function $D$ to use in estimating the level set clustering. Since these objects are sub-partitions, usual loss functions on partitions that are employed in model-based clustering will be inappropriate. We will discuss the issue further and introduce an appropriate loss in Section \ref{s:ia-loss-function}.

Finally, optimizing the risk function over the space of all sub-partitions, as shown in Equation \eqref{eq:psihat}, will be computationally intractable, since the
number of elements in $\ClustSpace[\cX_n]$ is immense. However, leveraging on the current Bayesian clustering literature, we adapt the discrete optimization algorithm of \cite{dahl2022search} to handle our case of sub-partitions.

Having addressed these issues, we refer to the resulting class $\{\widehat{\psi}_{\lambda, M}\}$ as Bayesian level set (\texttt{BALLET}) estimators. In Section \ref{s:theory} we show that, under suitable models $M$ for density $f$, the \texttt{BALLET} estimator $\widehat{\psi}_{\lambda, M}$ consistently estimates the level-$\lambda$ clustering based on $f_0$. 

\subsection{Surrogate Clustering Function}
\label{ss:surrogate-clustering-function}

Computing the clustering function $\psilam(f)$ based on the level set $\Slamf = \{x \in \cX : f(x) \geq \lambda \}$ involves two steps. The first identifies the subset of observations $\Activef = \Slamf \cap \cX_n$, called the active points for $f$, and the second separates the active points according to the topologically connected components of $\Slamf$. The first step is no more difficult than evaluating $f$ at each of the $n$ observations and checking whether $f(x_i) \geq  \lambda$ for $i \in \{1, \dots, n\}$. However,  identifying the connected components of $\Slamf$ can be computationally intractable unless $\cX$ is one-dimensional. This is a familiar challenge in algorithmic level set clustering \citep{campello2020density}. 

A common approach with theoretical support \citep{devroye1980detection, rinaldo2010generalized,sriperumbudur2012consistency} is to approximate the level set $\Slamf$ with a tube of diameter $\delta > 0$ around the active points: $\tube(A) = \cup_{x_i \in A} B(x_i,\delta/2)$, where $B(x, \delta/2)$ is the open ball of radius $\delta/2$ around $x$ and $A = \Activef$ denotes the active points. Calculating the connected components of $\tube(A)$ is straightforward. If we define $G_{\delta}(A)$ as the \emph{$\delta$-neighborhood graph} with vertices $A$ and edges $\{(x,x')\in A\times A \; |\,  \|x-x'\| < \delta \}$, then two points $x, x' \in \cX$ lie in the same connected component of $\tube(A)$ if and only if there exist active points $x_i, x_j \in A$ such that $\|x - x_i\| < \frac{\delta}{2}$, $\|x' - x_j\| < \frac{\delta}{2}$ and $x_i$, $x_j$ are connected by a path in $G_{\delta}(A)$. The problem is further simplified since we only need to focus on the active points: Any $x_i, x_j \in A$ lies in the same connected component of $\tube(A)$ if and only if $x_i, x_j$ are connected by a path in $G_{\delta}(A)$. \Cref{lem:cciscc} in \Cref{ss:prove-level-set-clustering} provides more details.

Hence, we define a computationally-tractable surrogate clustering function 
\begin{equation}
	\label{eq:levelset-clust-function}
	\psiclust(f) = \CC\{G_\delta(\Activef)\}
\end{equation}
where the dependence on the density $f$ and level $\lambda$ enter through the active points $\Activef = \{x  \in \cX_n | f(x) \geq \lambda \}$, and $\CC$ is the function that maps graphs to the \textit{graph-theoretic} connected components of their vertices \citep[][Chapter 3]{sanjoy2008algorithms}.  %

In \Cref{s:comparison-to-dbscan}, we discuss how the \texttt{DBSCAN} clustering algorithm \citep{ester1996density,schubert2017dbscan} essentially corresponds to evaluating $\psiclust(\hat{f})$ for a certain density estimator $\hat{f}$ of $f_0$.  In fact, for a general $f$, the computational complexity of evaluating $\psiclust(f)$ is comparable to that of \texttt{DBSCAN} with the additional cost of evaluating $f$ at the data points $\cX_n$.

Compared to the clustering point estimate $\psiclust(\hat{f})$ obtained by inserting a density estimator $\hat{f}$ based on $\cX_n$, the main motivation behind our Bayesian clustering machinery of \cref{eq:psihat} is to account for the variability of $\psiclust(f)$ in the posterior distribution of $f$. We expect our Bayesian point estimate of \cref{eq:psihat} to be more reliable than $\psiclust(\hat{f})$ in difficult level set clustering problems involving substantial uncertainty in density estimation.

Our clustering $\psiclust(f)$ depends on the choice of the parameter $\delta > 0$. 
 For some $k \in \nat$, $\gamma \in [0,1)$, and an estimate $\hat{f}$ of $f_0$, we suggest the data adaptive value of
 \begin{equation}
 	\label{eq:data-adaptive-delta}
 	 \hat{\delta}=q_{1-\gamma}\{\delta_k(x_i): x_i \in \Activef[\hat{f}] \},
 \end{equation}
 the $1-\gamma$ quantile of the $k$-nearest neighbor distance $\delta_k(x)$ among the estimated active data points $\Activef[\hat{f}]$, with our default choice of $\gamma=0.01$. The intuition here is that the value $\hat{\delta}$ will be smaller than the required distance between disjoint level $\lambda$ clusters of $f_0$ if the $k$-closest data points to most ($> 99\%$) of the active points are known to belong to the same cluster as the initial point. The choice of $k$ here also needs to be large enough to ensure that the level $\lambda$ cluster of $f_0$ is not disconnected by the skeleton graph $G_{\hat{\delta}}(\Activef[\hat{f}])$. Noting that the performance of \texttt{BALLET} clustering was not sensitive to our choice of $k$ (e.g. \Cref{fig:ballet-dbscan-compare-param-sensitivity}), we use the default value of $k=\lceil \log n \rceil$ in our analysis. In \Cref{sss:acurracy-of-level-set-estimator}, we theoretically study the accuracy of approximating $\psi_{\lambda}(f_0)$ by $\psiclust(\hat{f})$.  For suitably large $C > 0$, as long as $k \in [C \log n, n/C]$ and $\gamma < 1$, using $\delta=\hat{\delta}$ from \eqref{eq:data-adaptive-delta} will lead to consistent \texttt{BALLET} clustering with high probability (\Cref{lem:sufficiency-of-data-adaptive-estimator} and \Cref{lem:finite-sample-bound-for-assumption2}).

\subsection{Loss Function for Comparing Sub-partitions}
\label{s:ia-loss-function}

In order for \eqref{eq:psihat} to have the interpretation of a posterior Fréchet mean, $D: \ClustSpace[\cX_n] \times \ClustSpace[\cX_n] \to [0,\infty]$ must be chosen to be a metric on the space of sub-partitions $\ClustSpace[\cX_n]$. While any standard loss function on partitions (see \cite{dahl2022search}) has a natural  extension to sub-partitions, this does not result in a  metric on $\ClustSpace[\cX_n]$. For example, consider the popular Binder's loss $L_{\text{Binder}}$ which is a metric on the space of partitions (\cite{binder1978bayesian,wade2018bayesian}). Given a subset $C \subseteq \cX_n$ and its complement $C' = \cX_n \setminus C$, what should be the resulting loss between the sub-partitions $\bC = \{C\}$ and $\bC' = \{C'\}$? While $\bC$ and $\bC'$ are incredibly different when considered as level set clustering, the induced partitions are the same resulting in the loss $L_{\text{Binder}}(\{C, \cX_n \setminus C\}, \{C', \cX_n \setminus C'\})=0$.

Instead we now propose a modification of the Binder's loss, which will be a metric on the space of sub-partitions $\ClustSpace[\cX_n]$. Our Inactive/Active (IA) Binder's loss takes the form of Binder's loss for data points that are active in both partitions, with a penalty for points active in one partition and inactive in the other. We represent any sub-partition $\bC = \{C_1, \ldots, C_k\} \in \ClustSpace[\cX_n]$ with a length $n$ allocation vector $\bc = (c_1, \ldots, c_n) \in \{0,1\ldots, k\}^n$ such that $c_i = h$  if $x_i \in C_h$ and $c_i = 0$ if $x_i \in \cX_n \setminus \cup_{h=1}^k C_h$. Given two partitions $\bC, \bC'$ with active sets $A, A' \subseteq \cX_n$ and allocation vectors $\bc, \bc'$, the loss between them is defined as
\begin{align}
	\label{eq:ia-binder-loss}
	&\Lia(\bC, \bC') \nonumber\\
	&=  (n - 1)\big(
	m_{ai}\left|A \cap I' \right| + m_{ia} \left|I \cap A' \right|\big) + \sum_{\substack{ 1 \leq i < j \leq n \\ x_i, x_j \in A \cap A'}} a\1{(c_i = c_j; c'_i \neq c'_j)} + b\1{(c_i \neq c_j; c'_i = c'_j)},
\end{align}
where $I =  \cX_n \setminus A$ and $I' = \cX_n \setminus A'$ denote the inactive sets of $\bC$ and $\bC'$.
The loss is a well-defined function of $\bC$ and $\bC'$ since the right-hand side is invariant to any permutation of the active labels in $\bc$ and $\bc'$.
The summation term is the Binder's loss with parameters $a, b > 0$  restricted to points active in both sub-partitions. The first two terms, based on parameters $m_{ai}, m_{ia} > 0$, correspond to a loss of $(n-1)m_{ai}$ and $(n-1)m_{ia}$ incurred by points that are
active in $\bC$ but inactive in $\bC'$ and vice versa. We focus mainly on the setting where $a=b$ and $m_{ai} = m_{ai} = m \geq a/2$ with our default choice of $a=b=1$ and $m=1/2$ used throughout our analysis. Under these conditions \Cref{thm:IABender-is-a-metric} in \Cref{ss:level-set-theory} shows that $\Lia$ is a metric on $\ClustSpace[\cX_n]$. Our starting point is \Cref{rem:sum-form-of-penalty}, which provides an alternate representation of this loss.

Given \textit{any} distribution on $\bC$, we can compute the Bayes risk for an estimate $\bC'$ as the posterior expectation of the IA-Binder's loss:
\begin{align}
	\label{eq:ia-binder-risk}
	R_{\text{IA-Binder}}(\bC') =& E\{L_{\text{IA-Binder}}(\bC, \bC')\} \notag \\
	=& (n - 1)\bigg\{m_{ai}\sum_{i = 1}^n \Pr(x_i \in A) \1{(x_i \in I')} + m_{ia}\sum_{i = 1}^n \Pr(x_i \in I)\1{(x_i \in A')}\bigg\} \quad + \notag\\
	& \quad \quad \sum_{1 \leq i < j \leq n} \1{(x_i \in A', x_j \in A')} \big\{a\Pr(x_i \in A, x_j \in A, c_i = c_j)\1{(c'_i \neq c'_j)} \quad + \notag \\ 
	& \quad \quad \quad \quad \quad \quad \quad b\Pr(x_i \in A, x_j \in A, c_i \neq c_j)\1{(c'_i = c'_j)}\big\}.
\end{align}
The probabilities are computed based on the random clustering 
$\bC = \psiclust(f)$, where $f$ is drawn from the posterior $P_M( \cdot | \cX_n)$. 
Our \texttt{BALLET} estimator for level-$\lambda$ clustering is then
\begin{align}
	\label{eq:ballet-estimator}
	\BALLET &= \argmin_{\bC' \in \ClustSpace[\cX_n]} E_{f \sim P_M(\cdot|\cX_n)}[L_{\text{IA-Binder}}\{\psiclust(f), \bC'\}]  \\
	&\approx \argmin_{\bC' \in \ClustSpace[\cX_n]} \sum_{s = 1}^S L_{\text{IA-Binder}}\{\psiclust(f^{(s)}), \bC'\}, \nonumber
\end{align}
where the dependence of the estimator on the data is mediated by the posterior distribution $P_M(\cdot| \cX_n)$ from which we generate samples $f^{(1)}, \dots, f^{(S)}$. We precompute Monte Carlo estimates of the probabilities appearing in equation \eqref{eq:ia-binder-risk}. Then, estimating $\BALLET$ is based on optimizing the objective function. We rely on a modification of the algorithm of \cite{dahl2022search} described in \Cref{s:ballet-search-algorithm}.

When the posterior uncertainty in $f$ is small, one may use a heuristic \texttt{BALLET} \emph{plugin} estimate $\hat{\bC} = \psiclust(\hat{f})$ that avoids the expensive optimization in \eqref{eq:ballet-estimator} by directly computing the level set clusters of the posterior mean density $\hat{f}(x) \approx \frac{1}{S}\sum_{s = 1}^S f^{(s)}(x)$. While in many instances we found the \texttt{BALLET} plugin estimate to have similar performance to our \texttt{BALLET} estimator \eqref{eq:ballet-estimator} (e.g.~\Cref{tab:synthdata-tabular-results,tab:edsgc-edcci-coverage-full,tab:edsgc-abell-coverage-full}), the two estimates can be different (\Cref{fig:ballet_vs_plugin}). As a general principle, we always recommend the use of a Bayes estimator that directly targets the quantity of interest over a two-stage plugin approach  (see \Cref{ss:plugin-vs-decision-theoretic}).

\section{Credible Bounds}
\label{s:credible-bounds}

In addition to a clustering point estimate, we characterize the uncertainty. One popular strategy in Bayesian clustering is to examine the $n \times n$ posterior similarity matrix, whose $i,j$th entry contains the co-clustering probability $\Pr(c_i = c_j | \cX_n)$. Such summaries are complicated in our case by the fact that the entry $i$ and/or $j$ may be inactive. An appealing alternative is to adapt the method of \cite{wade2018bayesian} to compute credible balls for level set sub-partitions. 

To find a credible ball around the point estimate $\widehat{\bC}$ with credible level $1 - \alpha$ for $\alpha \in [0,1]$, we first find  
\begin{equation}
	\epsilon^* \doteq \argmin_{\epsilon > 0} P_M\{\psiclust(f) \in B_\epsilon(\widehat{\bC}) | \cX_n\} \geq 1 - \alpha,
	\label{eq:credible-ball-radius}
\end{equation}
the smallest radius $\epsilon=\epsilon^*$ such that the ball $B_\epsilon(\widehat{\bC}) = \{\bC' \in \ClustSpace[\cX_n] : \Lia(\widehat{\bC}, \bC') \leq \epsilon\}$ of radius $\epsilon$ around $\widehat{\bC}$ has a posterior coverage probability of at least $1-\alpha$. Then, the posterior distribution will assign a posterior probability close to $1 - \alpha$ to the event that $B_{\epsilon^*}(\widehat{\bC})$ contains $\bC = \psiclust(f)$, the unknown level set sub-partition.

The $1 - \alpha$ coverage credible ball $B_{\epsilon^*}(\widehat{\bC})$ typically contains a large number of possible sub-partitions. To summarize 
credible balls in the space of data partitions, \cite{wade2018bayesian} recommend identifying \textit{vertical} and \textit{horizontal} bounds based on the partial ordering of partitions associated with a Hasse diagram. The vertical upper bounds were defined as the partitions in $B_{\epsilon^*}(\widehat{\bC})$ that contained the smallest number of sets; vertical lower bounds, accordingly, were the partitions in $B_{\epsilon^*}(\widehat{\bC})$ that contained the largest number of sets; horizontal bounds were the partitions in $B_{\epsilon^*}(\widehat{\bC})$ that were the \textit{farthest} from $\widehat{\bC}$ at distance $\Lia$.

In our setting, in addition to similarity of sub-partitions in terms of their clustering structure, we must also compare inclusion or exclusion of observations from the active set. Uncertainty in the clustering structure will be partly attributable to uncertainty in which points are active. Fortunately, the space of sub-partitions is a lattice with an associated Hasse diagram (\SMref{s:subpartition-lattice}). We can move \textit{down} the sub-partition lattice by splitting clusters or removing items from the active set, while we can move \textit{up} the lattice of sub-partitions by merging clusters or absorbing noise points into the active set.

We propose the following computationally efficient algorithm for computing upper and lower bounds for the credible ball. Suppose we know our credible ball radius $\epsilon^*$ from \Cref{eq:credible-ball-radius} needed to achieve the desired coverage. We seek our upper bound by starting at the point estimate and greedily adding to the active set, one at a time, the item from the inactive set that has the greatest posterior probability of being active and reexamining the resulting connected components; this continues until we find a sub-partition that is farther than $\epsilon^*$ from the point estimate. To find a lower bound, we perform the analogous greedy removal process. The resulting bounds from applying this algorithm can be seen in  \Cref{fig:tsne-ballet-bounds,fig:synthdata-ballet-bounds,fig:edsgc-ballet-bounds,fig:tsne-ballet-bounds-lvls}.

\section{Consistency theory}
\label{s:theory}

In \Cref{ss:general-dclust-theory} we develop a general consistency theorem for Bayesian density-based clustering under three intuitive assumptions. Next in \Cref{ss:level-set-theory}, we carefully apply this result to our \texttt{BALLET} estimator $\BALLET$ from \eqref{eq:ballet-estimator} and derive mild conditions under which our method will be consistent. In the process, we provide theoretical guarantees on the accuracy of our surrogate clustering function from  \Cref{ss:surrogate-clustering-function}, indicating the choices of the parameter $\delta$ that lead to consistent estimation of level-$\lambda$ clusters. Indeed, our data adaptive choice of $\hat{\delta}$ in \eqref{eq:data-adaptive-delta} will be seen to satisfy this condition under suitable assumptions.

\subsection{A general consistency result for Bayesian density-based clustering}
\label{ss:general-dclust-theory}

In this section we show asymptotic consistency of a generic Bayesian density-based clustering estimator of the form
\begin{align}
	\label{eq:density-based-estimator}
	\cbayes &= \argmin_{\bC \in \ClustSpace[\cX_n]} E_{f \sim P_M(\cdot|\cX_n)}[D\{\psitilde(f), \bC\}], 
\end{align}
where $D$ is a loss on the space $\ClustSpace[\cX_n]$ of data sub-partitions and $\psitilde: \Den \to \ClustSpace[\cX_n]$ is an easy-to-compute surrogate 
that approximates the target density-based clustering function $\psi: \Den \to \ClustSpace[\cX_n]$. Similar to previous sections, we omit notation for the implicit dependence of $D$, $\psitilde$, and $\psi$ on $\cX_n$ and $n$. 
We will assume that the loss $D$ is a metric that takes values in $[0,1]$. We state our consistency result in terms of convergence in probability. Recall that a sequence of random variables $\{X_n\}_{n \geq 1}$ converges to zero in probability, denoted by $X_n \pconv 0$ as $n \to \infty$, if $\lim_{n \to \infty} \Pr(|X_n| > \epsilon) = 0$ for every fixed $\epsilon > 0$. 

Under some mild assumptions stated later, the following theorem establishes consistency of the estimator \eqref{eq:density-based-estimator}. In particular, when data $\cX_n$ are generated independently from $f_0$, it states that the Bayesian density-based clustering estimator defined in \eqref{eq:density-based-estimator} will be close to the target clustering $\psi(f_0)$ in terms of loss $D$ for large values of $n$. 

\begin{theorem}(Consistency of density-based clustering) Suppose that \Cref{ass:dist-bounded,ass:sup-norm-concentration,ass:continuity-at-f0} stated below hold, and $\cX_n = \{x_1, \ldots, x_n \} \iid f_0$. Then 
	$$
	0 \leq D\{\cbayes, \psi(f_0)\} \leq 2\tau_{1}(\cX_n) + 2\tau_{2}(\cX_n) \pconv 0 \qquad \text{as $n \to \infty$,}
	$$ 
	where $\cbayes$ is the density-based clustering estimate \eqref{eq:density-based-estimator} and the error terms $\tau_{1}$ and $\tau_{2}$ are as defined in  \Cref{ass:sup-norm-concentration,ass:continuity-at-f0}.%
	\label{thm:consistency}
\end{theorem}

We now discuss the assumptions underlying \Cref{thm:consistency}. The proof of \Cref{thm:consistency}, provided in \Cref{ss:proof-of-consistency}, captures the intuition that as long as the posterior distribution of $f$ is concentrated around $f_0$ in terms of a metric $\rho$ on $\Den$ (\Cref{ass:sup-norm-concentration}) that can guarantee that the two clusterings $\psitilde(f)$ and $\psi(f_0)$ are close (\Cref{ass:continuity-at-f0}), then our Bayesian density-based clustering estimator $\cbayes$ will also be close to $\psi(f_0)$ by using the triangle inequality for $D$ (\Cref{ass:dist-bounded}).

\begin{assumption}
	Suppose that $D:  \ClustSpace[\cX_n] \times  \ClustSpace[\cX_n] \to [0,1]$ is a metric. 
	\label{ass:dist-bounded}
\end{assumption}

Next, we assume that the Bayesian model $M$ for the unknown density $f$ is such that its posterior distribution $P_M(\cdot|\cX_n)$, under samples $\cX_n = \{ x_1, \ldots, x_n \}$ drawn independently from $f_0$,  
 contracts at rate $\epsilon_n$ to $f_0$ in some metric $\rho$ on the space of densities $\Den$. 

\begin{assumption}[Posterior contraction] If $\cX_n = \{ x_1, \ldots, x_n \}$ are drawn independently from $f_0$, then there is a metric $\rho$ on $\Den$ and there is a non-negative sequence of numbers $\{\epsilon_n\}_{n \geq 1}$ converging to zero such that 
	$$
	\tau_{1}(\cX_n) \doteq P_M\left(f : \rho(f,f_0) \geq \epsilon_n K_n \big| \cX_n\right) \pconv 0 \text{ as } n \to \infty,
	$$
	 for every non-negative sequence $\{K_n\}_{n \in \nat}$ that diverges to infinity. 
	\label{ass:sup-norm-concentration}
\end{assumption}

\begin{assumption} There is a non-negative sequence $\{K_n\}_{n \in \nat}$ that diverges to infinity such that $\tau_{2}(\cX_n) \doteq \sup_{f \in \Den : \rho(f,f_0) \leq K_n \epsilon_n} D\{\psitilde(f), \psi(f_0)\} \pconv 0$ as $n \to \infty$, where $\rho$ and $\epsilon_n$ are as given in \Cref{ass:sup-norm-concentration}.
	\label{ass:continuity-at-f0}
\end{assumption}	

 \Cref{ass:continuity-at-f0,ass:sup-norm-concentration} are related in that we need a common sequence $\{(\epsilon_n, K_n)\}_{n \geq 1}$ and the same metric $\rho$ on $\Den$ such that both \Cref{ass:continuity-at-f0,ass:sup-norm-concentration} hold. 
Standard posterior contraction results \cite[e.g.][Chapter 9]{ghosal2017fundamentals} can establish the condition in \Cref{ass:sup-norm-concentration} for various models $M$ and suitable rates $\epsilon_n \to 0$ when $\rho$ is the Hellinger or total-variation metric on $\Den$. However, here one may need contraction in a stronger metric $\rho$ on $\Den$ to ensure continuity of the clustering functional $\psi: \Den \to \ClustSpace[\cX_n]$ to guarantee \Cref{ass:continuity-at-f0} even when  $\psitilde=\psi$. For example, for our application to level set clustering we will use the $L^\infty$ metric $\rho(f,g) = \|f-g\|_\infty \doteq \sup_{x \in \cX} |f(x) - g(x)|$ in \Cref{ss:level-set-theory}. Similarly, we expect to use a metric $\rho$ that captures uniform convergence of both the density $f$ and its derivatives to satisfy \Cref{ass:continuity-at-f0} when $\psi$ describes modal clustering \cite[see the introduction of][]{shen2017posterior}. Thus establishing posterior contraction results in stronger metrics $\rho$ than the standard Hellinger distance is a promising active area of research \citep{gine2011rates,  
castillo2014bayesian,castillo2017polya,naulet2022adaptive,shen2017posterior,li2021posterior} that can help establish consistency of Bayesian density-based clustering.

\subsection{Application to level set clustering}
\label{ss:level-set-theory}

Note that \eqref{eq:ballet-estimator} represents a special case of \eqref{eq:density-based-estimator}, when $\psitilde = \psiclust$ is the surrogate clustering function defined in \eqref{eq:levelset-clust-function}, $\psi = \psi_{\lambda}$ is the level-$\lambda$ clustering function defined in \Cref{ss:decision-theoretic-framework}, and $D = \binom{n}{2}^{-1}L_{\text{IA-Binder}}$ is a rescaled version of the Inactive-Active Binder loss \eqref{eq:ia-binder-loss}. We will fix this choice of $\psi, \psitilde$ and $D$ throughout this section. We show that \Cref{ass:continuity-at-f0,ass:dist-bounded,ass:sup-norm-concentration} are satisfied for suitable choices of the parameter $\delta > 0$ and suitable conditions on the density $f_0$, level $\lambda > 0$, and model $M$. 
Following the existing level set clustering theory \citep[e.g.][]{sriperumbudur2012consistency,rinaldo2010generalized,jiang2017density}, in this section we will use the $L^\infty$ metric $\rho(f,g) = \|f-g\|_\infty = \sup_{x \in \cX} |f(x) - g(x)|$ on $\Den$ in \Cref{ass:continuity-at-f0,ass:sup-norm-concentration}. While posterior consistency of the density $f$ in the $L^\infty$ metric is  a strong requirement, \Cref{rem:extend-contraction} briefly discusses how this requirement might be weakened.

\subsubsection{Properties of IA-Binder's loss}

To establish the validity of \Cref{ass:dist-bounded} we study the properties of our Inactive-Active Binder loss \eqref{eq:ia-binder-loss}. The following theorem proved in \Cref{ss:metric-properties} shows that \Cref{ass:dist-bounded} is satisfied for suitable choices of constants in our Inactive-Active Binder loss \eqref{eq:ia-binder-loss}.

\begin{theorem}  Suppose $0 < a = b \leq 1$, $m = m_{ia} = m_{ai} \leq 1$, and $a \leq 2m$. Then $D = \binom{n}{2}^{-1} \Lia$ is a metric on $\ClustSpace[\cX_n]$ that is bounded above by 1.
	\label{thm:IABender-is-a-metric}
\end{theorem}

The following remark, which will be useful to interpret the conclusion of \Cref{thm:consistency}, describes when the distance $D$ between two sub-partitions $\bC_1, \bC_2 \in \ClustSpace[\cX_n]$ will be small.

\begin{remark} We say that a pair of distinct points $x_i, x_j \in \cX_n$ is clustered differently by $\bC_1$ and $\bC_2$ if the activity status of either $x_i$ or $x_j$ is different across  $\bC_1$ and $\bC_2$, or else both $x_i$ and $x_j$ are active in both $\bC_1$ and $\bC_2$ but the two points belong to the same cluster in $\bC_1$ (or $\bC_2$) but to different clusters in $\bC_2$ (or $\bC_1$). Importantly, $\Lia$ can be expressed as a sum of non-negative penalties over distinct pairs of points from $\cX_n$
\begin{equation*}
    \Lia(\bC_1,\bC_2) =  \sum_{1 \leq i < j \leq n}^n \phi_{i,j},
\end{equation*}
where the penalty $\phi_{i,j} \in \{0,a,m,2m\}$ takes a positive value of at least $\min(a,m)$ when the pair $x_i,x_j$ is  clustered differently by $\bC_1$ and $\bC_2$. (See \eqref{eq:sum-representation} in \Cref{ss:metric-properties} for exact details.) Thus for the choice of $a, m \in [1/2,1]$ and any $\epsilon \in (0,1/2)$, if the rescaled loss $D(\bC_1,\bC_2) = \binom{n}{2}^{-1} \Lia(\bC_1,\bC_2)$ is less than $\epsilon$ then at most $2\epsilon$ fraction of all pairs of points from $\cX_n$ will be clustered differently by $\bC_1$ and $\bC_2$. Conversely, if at most $\epsilon$ fraction of all pairs of points from $\cX_n$ are clustered differently by $\bC_1$ and $\bC_2$ then $D(\bC_1,\bC_2) < 2\epsilon$.
\label{rem:sum-form-of-penalty}
\end{remark}

\subsubsection{Accuracy of our level-set clustering surrogate}
\label{sss:acurracy-of-level-set-estimator}

We now examine \Cref{ass:continuity-at-f0} here, while \Cref{ass:sup-norm-concentration} will be examined in \Cref{sss:models-with-linfty-contraction}. 

The following result, proved in \Cref{ss:prove-level-set-clustering}, demonstrates that \Cref{ass:continuity-at-f0}  will be satisfied as long as the density $f_0$ satisfies some mild conditions and $\gamma = K_n \epsilon_n \to 0$.  Generally speaking, we require that $f_0: \R^d \to [0,\infty)$ is continuous and vanishing in the tails (\Cref{ass:uniform-continuity}), is not flat around the level $\lambda$ (\Cref{ass:mass-decay}), and has a level-$\lambda$ clustering that is stable with respect to small perturbations in $\lambda$ (\Cref{ass:robust-cluster}). Under these conditions, with high-probability our surrogate clustering estimator $\psiclust[\delta](f)$ from \Cref{ss:surrogate-clustering-function} will be close to the true clustering $\psi_\lambda(f_0)$ in terms of our distance $D$ as long as $f$ is close to $f_0$ in the $L^\infty$ metric and $\delta$ lies in a suitable range.

\begin{theorem} Suppose $\cX = \R^d$ and the density $f_0$ and the level $\lambda > 0$ satisfy \Cref{ass:uniform-continuity,ass:mass-decay,ass:robust-cluster} in \Cref{ss:prove-level-set-clustering}. Suppose further that $f_0$ is 
$\alpha$-H\"older continuous for some $\alpha \in (0,1]$, the dataset $\cX_n = \{x_1, \ldots, x_n\}$ is drawn independently from $f_0$ with $n \geq 16$, and $D$ is the re-scaled loss in \Cref{thm:IABender-is-a-metric}.  Then, depending on $f_0$,  there are finite constants $C_0, \bar{\delta}, \bar{\gamma} > 0$ such that%
	$$
	\sup_{f: \|f-f_0\|_\infty \leq \gamma } D\{\psiclust(f), \psi_\lambda(f_0)\} \leq C_0 \bigg\{\max(\gamma, \delta^{\alpha}) + \sqrt{\frac{\ln n}{n}} \bigg\}
	$$
	holds uniformly over all $\delta \in [\rnld, \bar{\delta})$ and $\gamma \in (0,\bar{\gamma})$ with probability at least $1-\frac{1+n}{n^2}$. Here $\rnld \doteq \rnldval$ where $v_d$ is the volume of the unit Euclidean ball in $d$ dimensions.
	\label{lem:finite-sample-bound-for-assumption2}
\end{theorem}

The constraint $\delta \geq \rnld$ in \Cref{lem:finite-sample-bound-for-assumption2} ensures that, with high probability, every open ball  $B(x,\delta/2)$ contained in  $S_{\lambda}$ will also contain at least one data point $x_i \in \cX_n \cap B(x,\delta/2)$. This key result is used in \Cref{lem:level-set-containment} to show that the level set estimator $T_{\delta}(A_{f, \lambda})$ from \Cref{ss:surrogate-clustering-function} will be suitably close to $S_\lambda$ when $\|f-f_0\|_\infty$ and $\delta$ are small (and $\delta \geq \rnld$). The following lemma proved in \Cref{ss:analysis-of-data-adaptive-estimator} shows that our data adaptive choice of $\hat{\delta}$ in \eqref{eq:data-adaptive-delta} will satisfy conditions of \Cref{lem:finite-sample-bound-for-assumption2} with high probability if $\log n \ll k \ll n$ as $n \to \infty$.

\begin{lemma}
		\label{lem:sufficiency-of-data-adaptive-estimator}
		Suppose the assumptions of \Cref{lem:finite-sample-bound-for-assumption2} are satisfied and the density estimator $\hat{f}$ satisfies $\|\hat{f}-f\|_\infty \leq \lambda/2$. Then there is a finite constant $L > 0$ depending on $f_0$ and $\lambda$ such that if $k \in [L \ln n,  n/L]$ then $\hat{\delta} \in [\rnld, \bar{\delta})$ with probability at least $1-2e^{-\frac{1}{32}\sqrt{\frac{k}{d \ln n}}}$.
\end{lemma}

\subsubsection{Consistency of level set clustering}
\label{sss:models-with-linfty-contraction}

\Cref{ass:sup-norm-concentration} requires posterior contraction around $f_0$ in the $L^\infty$ norm. While such contraction results can be obtained when the model $M$ is based on a parametric family that contains $f_0$, the search for such results when $M$ is a  non-parametric model is currently an active area of research. For univariate density estimation on $\cX=[0,1]$, such contraction rates have been established for 
kernel mixture models, random histogram priors, P\'olya trees, Gaussian process and wavelet series priors on the log density
\citep{gine2011rates,  
castillo2014bayesian,castillo2017polya,naulet2022adaptive}. 
For multivariate density estimation on $\cX = [0,1]^d$, %
refer to 
\cite{li2021posterior} and references therein. %

Combining all the results in this section leads to the following corollary of \Cref{thm:consistency}.

\begin{corollary} Suppose $\cX=\R^d$, density $f_0 \in \Den$ and level $\lambda > 0$ satisfy \Cref{ass:uniform-continuity,ass:mass-decay,ass:robust-cluster} in \Cref{ss:prove-level-set-clustering}, and data $\cX_n = \{x_1, \ldots, x_n\}$ are drawn independently from $f_0$. Recall the \texttt{BALLET} estimator $\BALLET$ from \eqref{eq:ballet-estimator} based on:
	\begin{enumerate}
		\item the loss $\Lia$ with parameters $0 < a = b \leq 1$, $m = m_{ia} = m_{ai} \leq 1$, and $a \leq 2m$, 
		\item a model $M$ that satisfies \Cref{ass:sup-norm-concentration}, and
		\item a non-random $\delta \in [\rnldval, \bar{\delta})$ or the data adaptive choice of $\delta=\hat{\delta}$ from \eqref{eq:data-adaptive-delta} with $\gamma < 1$ and $\log n \ll k \ll n$ as $n \to \infty$,
	\end{enumerate}
	where $\bar{\delta}$ is a positive constant that depends on $f_0$ and $\lambda$. Then
\begin{equation*}
	\binom{n}{2}^{-1}\Lia\{\BALLET, \psi_\lambda(f_0)\} \pconv 0  \qquad \text{as } n \to \infty.
	\label{eq:re-scaled-binder-consistency}
\end{equation*}
\label{cor:ballet-consistency}
\end{corollary}

By \Cref{rem:sum-form-of-penalty}, the corollary implies that only a vanishingly small fraction of pairs of distinct points from $\cX_n$ will be clustered differently by our \texttt{BALLET} estimator $\BALLET$ and the associated true level set clustering $\psi_\lambda(f_0)$ as $n \to \infty$.

\begin{remark} \label{rem:extend-contraction}
    \Cref{ass:sup-norm-concentration} with $\rho(f,g) = \|f-g\|_{\infty}$ seems stronger than necessary to establish the consistency of our \texttt{BALLET} estimator \eqref{eq:ballet-estimator}, which depends on the model $M$ only through the distribution of the level set  $S_{\lambda, f} = \{x \in \R^d : f(x) \geq \lambda \}$ under the posterior draw $f \sim P_M(\cdot|\cX_n)$. One might thus hope to leverage existing posterior contraction results \citep{gayraud2005rates,gayraud2007consistency,li2021posterior} for level sets that show
    $$
    P_M\left[\tilde{\rho}(S_{\lambda, f},S_{\lambda, f_0}) > \epsilon \big| \cX_n\right] \pconv 0 \text{ as } n \to \infty, \text{ for each } \epsilon > 0,
    $$ 
    where $\tilde{\rho}(A,B) = \textrm{Leb}(A \Delta B)$ is typically the Lebesgue measure of the symmetric difference between (measurable) subsets $A, B \subseteq \cX$. Consistency of \texttt{BALLET} then essentially reduces to establishing a `continuity' result similar to \Cref{lem:finite-sample-bound-for-assumption2} that will bound the distance $D\{\tilde{\psi}_{\delta,\lambda}(f), \psi_{\lambda}(f_0)\}$ between clusterings whenever the distance $\tilde{\rho}(S_{\lambda, f}, S_{\lambda, f_0})$ between the corresponding level sets is small. This approach seems more feasible if $\tilde{\rho}$ can be taken to be a stronger metric like the Hausdorff metric \citep{li2021posterior,chen2016density}.
\end{remark}

\section{Illustrative Challenge Datasets}
\label{s:toy-data-analysis}
To highlight some of the appealing properties of the \texttt{BALLET} estimator, we analyze two illustrative clustering datasets: a simulated example of the classic two moon problem and an RNA sequencing dataset (\url{https://www.reneshbedre.com/blog/tsne.html}).

For each dataset, we model the observations as iid draws from density $f$ and $f$ as a draw from a Dirichlet process mixture of normal distributions with a multivariate normal-inverse Wishart base measure (\texttt{DPMM}). We generate samples $f^{(1)}, \dots, f^{(S)}$ from the posterior $f \,|\, \cX_n$ using the \texttt{dirichletprocess} package, available on \texttt{CRAN}.

We then use these posterior samples to compute \texttt{BALLET}  clustering point estimates. For the two-moon problem, we choose the target density level $\lambda$ at the 10th percentile of the estimated observation densities $\{\hat{f}(x_i):\, x_i \in \cX_n \}$ such that 90\% of the observations are assigned to clusters and 10\% are labeled as noise. For the RNA-seq data, we set $\lambda$ at the 15th percentile. These results are visualized in the right column of Figure \ref{fig:toy-challenge-dpmm}. 

In the center column of Figure \ref{fig:toy-challenge-dpmm} we visualize the clustering estimate obtained from a traditional mixture component allocation approach to Bayesian clustering and summarized using \cite{dahl2022search}. The same \texttt{DPMM} posterior was used for both sets of clustering estimates; the associated density point estimates, $\hat{f}$, are visualized in the left column.

Additional analyses of these and one other simulated data set are collected in \SMref{s:app-toy-challenge}. In particular, we show credible bounds (\Cref{fig:tsne-ballet-bounds}), highlight the robustness of \texttt{BALLET} to alternative models for $f$ (\Cref{fig:toy-challenge-compare}), and present results over a range of values for $\lambda$ (\Cref{fig:toy-challenge-compare-high}, \Cref{fig:tsne-ballet-bounds-lvls}). A discussion of how we chose the level $\lambda$ can be found in \Cref{s:ballet-tuning-parameter}.

\begin{figure}[h]
	\centering
	\includegraphics[width=0.9\textwidth]{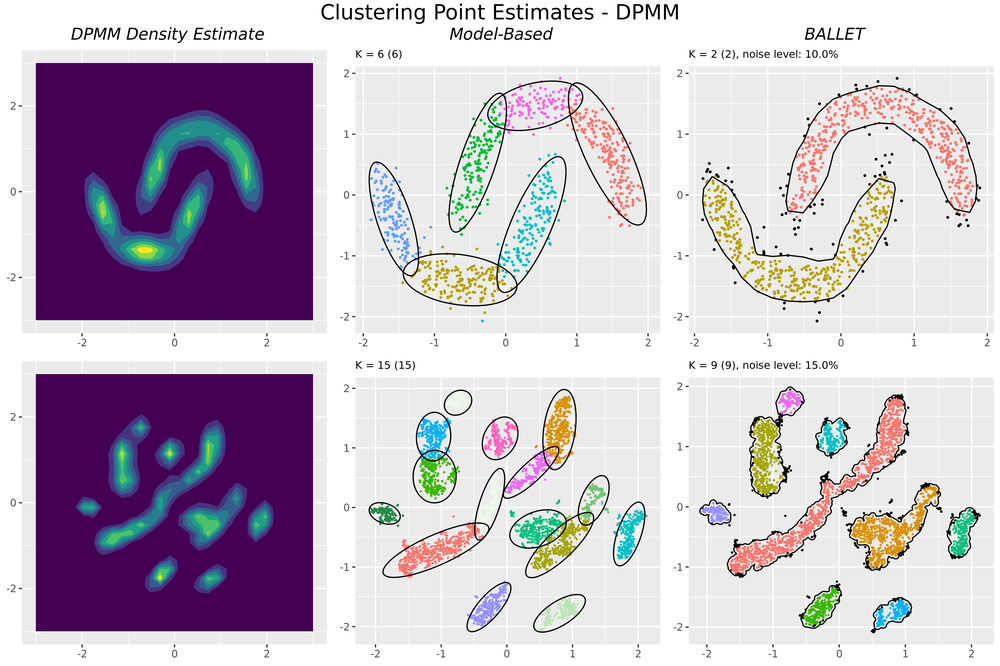}
	\caption{Analysis of the two moons and RNA-seq datasets. The first column shows a heatmap of $E(f|\cX_n)$ for the \texttt{DPMM} model. The center column shows the cluster estimate obtained from the traditional mixture-component allocation approach, and the third shows our \texttt{BALLET} point estimates. The number of clusters identified in each point estimate is shown at the top of each subplot, with the number of non-singleton clusters listed in parentheses. For the \texttt{BALLET} subplots we also note the target level of noise-points used to set $\lambda$.}
	\label{fig:toy-challenge-dpmm}
\end{figure}

\section{Analysis of Astronomical Sky Survey Data}
\label{s:real-data-analysis}
Astronomical sky surveys document the locations and redshifts of galaxies in the cosmos \citep{nichol1992edinburgh}. One aim in collecting the data is to analyze the spatial distribution of galaxies, as the size and distribution of high-density regions can help us estimate certain parameters of cosmological models, 
as described by \cite{jang2006nonparametric} in their non-Bayesian analysis of this level set clustering problem. Here, we perform a parallel analysis using \texttt{BALLET}, which offers us the benefits of more stable Bayesian nonparametric density estimation and Bayesian uncertainty quantification.

The data $\cX_n$ are a cleaned subset of the Edinburgh-Durham Southern Galaxy Catalogue \citep{nichol1992edinburgh} consisting of $n \approx 41K$ observations in a square region $\cX \subseteq \R^2$ and come with two catalogues of \textit{suspected} cluster locations: the Abell catalogue \citep{abell1989catalog} and the Edinburgh/Durham Cluster Catalog I  (EDCCI) \citep{lumsden1992edinburgh}. The former was created by visual inspection of the data by domain experts, while the EDCCI was produced by a custom-built cluster identification algorithm.    \Cref{fig:edsgc-density-estimate} visualizes the locations from these two catalogues overlying our posterior density estimate (\Cref{s:density-model-and-tuning-params}). Here we aim to estimate level set clusters and their uncertainty, and compare the results to locations in the two catalogues, which will serve as our imperfect ground truth.

We first conduct a simulation study, generating one hundred synthetic datasets designed to resemble the Edinburgh-Durham Southern Galaxy Catalogue data, analyzing them by the same \texttt{BALLET} methodology we will use for the real data, and computing sensitivity and specificity in detecting regions with excess density. To accommodate the fact that target clusters are described only by their central point (corresponding to a \emph{simulated} catalogue location) henceforth called a \emph{target point}, we evaluate sensitivity and specificity based on small ellipses enclosing each estimated cluster: sensitivity is measured as the proportion of target points contained in at least one ellipse, while specificity is measured as the proportion of ellipses which contain a target point.  Since sensitivity and specificity will both be equal to one if all the data points are assigned to a single cluster, we also compute a metric called \emph{exact match}, defined as the fraction of ellipses that have exactly one target point. As a competitor, we apply \texttt{DBSCAN} \citep{ester1996density}. 

\subsection{Density Model and Choice of Parameters}
\label{s:density-model-and-tuning-params}

In both the simulation study and real data analysis, we model the density $f$ with a simple mixture of random histograms: $f(x) = \sum_{k = 1}^K \pi_k H_k(x; \mathcal{B}_k, \brho_k)$, where $H_k(x; \mathcal{B}_k, \brho_k) = \sum_{m = 1}^M \1{(x \in B_{km})} \rho_{km}$ is a histogram density with bins $\mathcal{B}_k = (B_{k1}, \dots, B_{kM})$ and weights $\brho_k = (\rho_{k1}, \dots, \rho_{kM})$. We provide more details on our prior along with  a fast approximation to sample from the posterior of $f$ in \SMref{s:random-histogram-model}. 

Cosmological theory \citep[see][]{jang2006nonparametric} suggests the use of  the level $\lambda = (1+c) \bar{f}$,  where the  constant $c$ is  approximately one and $\bar{f} = \frac{\int_{\cX} f(x) dx}{\operatorname{Vol}(\cX)} = 1/\operatorname{Vol}(\cX)$  denotes the average value of $f$.  We chose the value $c=1$ for our analysis of the real data. This corresponded to declaring the fraction $\nu=.927$ of data points as noise. In the simulation study, we fix the fraction of noise points which are not assigned to a cluster at  
$\nu=0.9$ and set $\delta=\hat{\delta}$ from \eqref{eq:data-adaptive-delta}. The analogous parameter settings for \texttt{DBSCAN} are $\MinPts = k$ and $\texttt{Eps}=q_{1-\nu}[\{\delta_k(x_i) : x_i \in \cX_n\}]$ \citep{ester1996density}, where $\delta_k(x)$ is the distance from $x$ to the $k$th nearest point in the dataset $\cX_n$ and $q_{\alpha}$ is the quantile function corresponding to $\alpha \in (0,1)$.  Unlike for \texttt{BALLET}, the performance of \texttt{DBSCAN} in our simulation study was sensitive to the choice of $k$ (see \Cref{fig:ballet-dbscan-compare-param-sensitivity}). We also present results from \texttt{DBSCAN} in \Cref{s:app-edsgc-sky-survey,s:app-synthetic-sky-survey} with $\MinPts=60$ which was chosen via grid-search to optimize performance. The results were comparable to those of \texttt{BALLET} using the default parameter values.

\subsection{Simulation Study}
The simulation data were drawn from a mixture distribution that placed $\nu = 90\%$ of its mass in a uniform distribution over the unit square. %
and divided the remaining 10\% between 42 bivariate isotropic Gaussian components, with relative weights determined by a draw from a uniform distribution over the probability simplex. The component means are sampled uniformly from the unit square, and the variances were drawn from a diffuse inverse gamma distribution. We randomly generated one hundred such mixture distributions and drew $n=40000$ independent and identically distributed observations from each mixture distribution, dropping any observations that fell outside the unit square. We plot a typical synthetic data set in Figure \ref{fig:synthdata-data} and display the associated true and estimated high-density regions in Figure \ref{fig:synthdata-density}.

\begin{figure}%
    \centering
    \subfloat[\centering  DBSCAN clustering]{{\includegraphics[width=0.4\textwidth]{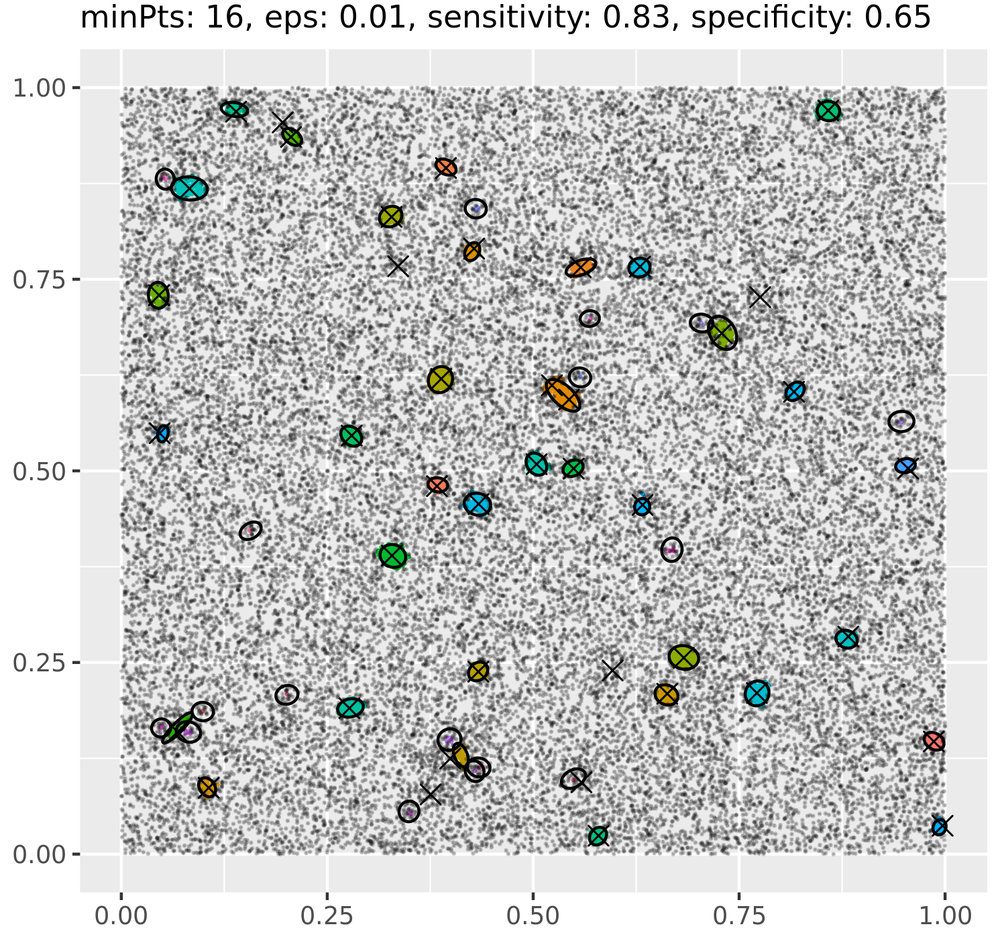} }}%
    \qquad
    \subfloat[\centering  BALLET Clustering]{{\includegraphics[width=0.4\textwidth]{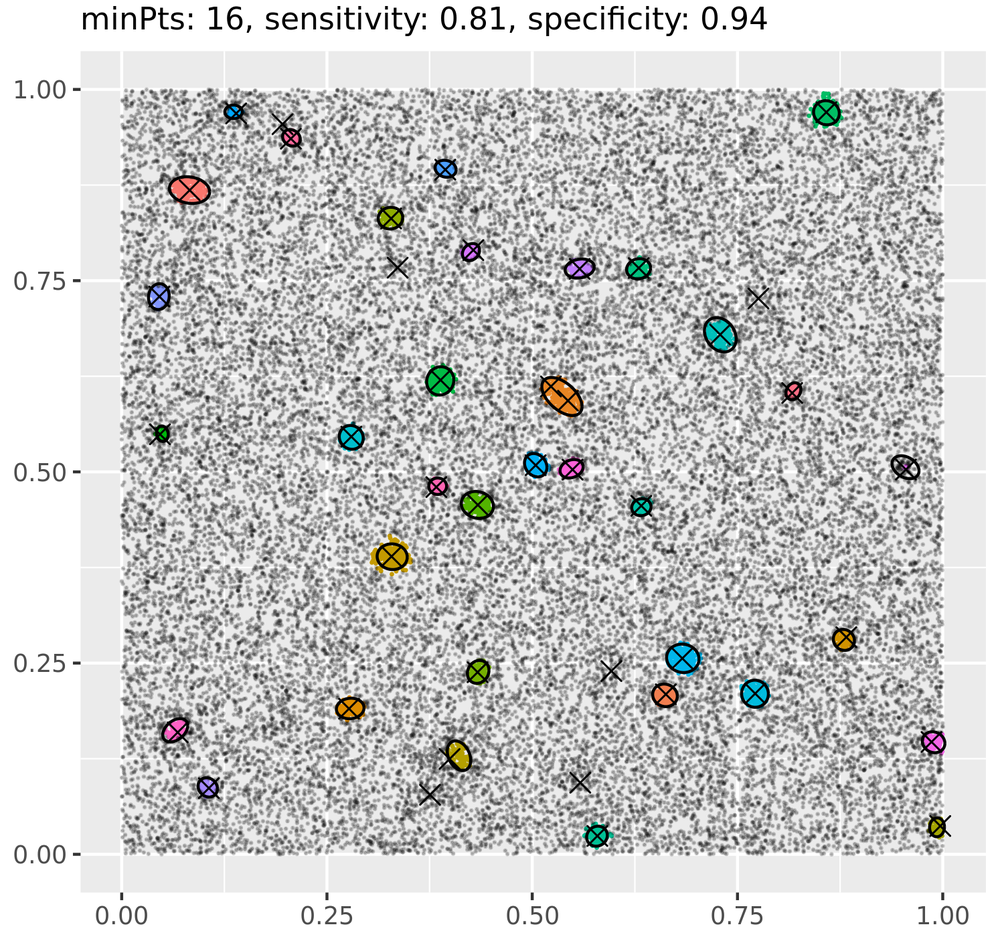} }}%
    \caption{Clusters estimated by \texttt{DBSCAN} and \texttt{BALLET}  for a representative synthetic sky survey dataset from our simulation study.  We see an apparent preference of \texttt{DBSCAN} for detecting a large number of singleton or near-singleton clusters.} %
    \label{fig:synthdata-pe}%
\end{figure}

In \Cref{fig:synthdata-pe}, we show the result of applying \texttt{DBSCAN} and \texttt{BALLET} to the typical synthetic dataset, highlighting \texttt{DBSCAN}'s apparent preference for detecting a large number of singleton or near-singleton clusters given our default choice of $\MinPts = k_0 = \lceil \log_2(n) \rceil = 16$ and the known fraction of noise points $\nu=90\%$. The average performance of \texttt{DBSCAN} and \texttt{BALLET} clustering (point estimate and upper and lower bounds) in all the hundred datasets is shown in \Cref{tab:synthdata-tabular-results}. \texttt{DBSCAN} achieved an average sensitivity of 0.86, but suffered substantial false positives with an average specificity of 0.49 (exact match = 0.45). \texttt{BALLET} achieved an average sensitivity of 0.78 while maintaining a nearly perfect average specificity at 0.99 (exact match =  0.88). The \texttt{BALLET} lower and upper bounds performed more and less conservatively, respectively,  than the point estimate. In particular, on average, the \texttt{BALLET} lower bound had less sensitivity (.62) but more specificity (.99) and exact matches (.9), while the \texttt{BALLET} upper bound had more sensitivity (.89) but less specificity (.96)  and exact matches (.83). 

The performance of \texttt{DBSCAN} improved to match that of \texttt{BALLET} when $\MinPts= k = 60$ was chosen to maximize the sum of the sensitivity and specificity values (\Cref{tab:synthdata-tabular-results}).  The performance of \texttt{BALLET} remained insensitive to the choice of $k$ (\Cref{fig:ballet-dbscan-compare-param-sensitivity}). Thus while carefully tuning hyper-parameters based on the ground truth was necessary for \texttt{DBSCAN} to match the performance of \texttt{BALLET},  the performance of \texttt{BALLET} seems more robust to loss parameters. This may be because \texttt{BALLET} separates careful data modeling from the task of inferring level set clusters.

\subsection{Sky Survey Data Analysis}

We applied \texttt{DBSCAN} and \texttt{BALLET} to the Edinburgh-Durham Southern Galaxy Catalogue data as described above, choosing 
$\MinPts = k_0$ based on our default value of $k_0 = \lceil \log_2(n) \rceil = 16$  %
or $\MinPts=60$, the value optimized in our simulation study.
Clustering results are shown in
\Cref{fig:edsgc-ballet-pe,fig:edsgc-dbscan-heuristic,fig:edsgc-dbscan-tuned}.

\begin{table}[ht]
\centering
\begin{tabular}{|l||l|l||l|l|l||l|}
\hline
& \texttt{DBSCAN} & \texttt{DBSCAN}$^1$ & \thead{\texttt{BALLET}\\ Lower} & \thead{\texttt{BALLET}\\ Est.} & \thead{\texttt{BALLET}\\ Upper} &\thead{\texttt{BALLET}\\ Plugin}\\
\hline
Sensitivity & 0.79& 0.69& 0.29 & 0.67 & 0.86 & 0.67 \\
Specificity & 0.20& 0.65& 0.87 & 0.69 & 0.42 & 0.69 \\
Exact Match & 0.17& 0.46& 0.67 & 0.51 & 0.32 & 0.53 \\
\hline
\end{tabular}

\caption{\texttt{DBSCAN} and \texttt{BALLET} coverage of suspected galaxy clusters in the EDCCI catalogue. Column \texttt{DBSCAN} reports performance with our default tuning parameter choice $\MinPts=16$, while \texttt{DBSCAN}$^1$ 
shows performance with 
$\MinPts=60$ based on our simulation study.}
\label{tab:edsgc-edcci-coverage-full}
\end{table}

\Cref{tab:edsgc-edcci-coverage-full} compares inferred clusters to the EDCCI catalogue of suspected galaxy clusters. While \texttt{DBSCAN} with heuristic parameter choice detected 79 percent of the EDCCI clusters, the method only had a specificity of 20 percent. \texttt{DBSCAN} with parameter optimized in our simulation study found 
69 percent of the EDCCI clusters with a specificity of 65 percent.  \texttt{BALLET} recovered 67 percent of the EDCCI clusters and had a specificity of 69 percent. \texttt{DBSCAN} and \texttt{BALLET} detected only 40 percent of the Abell catalogue clusters (\Cref{tab:edsgc-abell-coverage-full}), 
but performed better at recovering suspected galaxy clusters in the EDCCI, which is considered more reliable
\citep{jang2006nonparametric}. 

\Cref{fig:edsgc-ballet-bounds} visualizes \texttt{BALLET} clustering uncertainty  (\Cref{s:credible-bounds}) via upper and lower bounds for a 95 percent credible ball. The lower bound has fewer and smaller clusters and tends to include locations that the EDCCI and Abell catalogs agree on.  In contrast, the upper bound has larger and more numerous clusters, and tends to include many of the suspected cluster locations from both the catalogs. 
Based on  \Cref{tab:edsgc-edcci-coverage-full,tab:edsgc-abell-coverage-full}, one may suspect that the 14 percent EDCCI locations and 44 percent Abell locations that were not discovered by the \texttt{BALLET} upper bound may be erroneous. On the other hand, we may have high confidence in the 29 percent locations in EDCCI and 21 percent locations in Abell which were discovered by the \texttt{BALLET} lower bound. 

\begin{figure}
    \centering
    \includegraphics[width=0.9\textwidth]{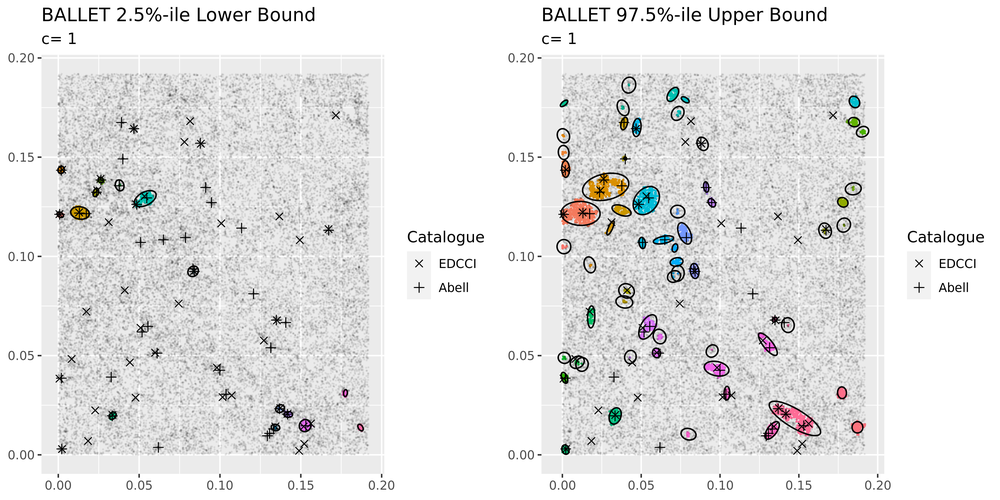}
    \caption{Upper and lower bounds of the 95\% credible ball centered at our \texttt{BALLET} estimate of the galaxy clusters in the Edinburgh-Durham Southern Galaxy Catalogue  data.}
    \label{fig:edsgc-ballet-bounds}
\end{figure}

\section{Discussion}
\label{s:discussion}

In this article, we developed a Bayesian approach to density based clustering, focusing on level set clustering as an important special case. Our key idea is to use Bayesian decision theory \citep{berger2013statistical} to separate the tasks of modeling the data density and inferring clusters. This provides a general new paradigm for inferring clusters, while representing uncertainty in clustering. A decision theoretic decoupling approach has  proved useful in various problem settings like interpretable modeling \citep{gutierrez2005statistical,afrabandpey2020decision,woody2021model}, variable selection in regression \citep{kowal2022bayesian,hahn2015decoupling}, factor analysis \citep{bolfarine2024decoupling}, structured covariance estimation \citep{bashir2019post}, and analysis of functional data \citep{kowal2020bayesian,kowal2022fast}. Our approach is also a case of this posterior decoupling methodology where we establish necessary conditions for consistency (\Cref{thm:consistency}).

A crucial and implicit part of our methodology is the  model $M$ on the space of densities. In any application, the problem of coming up with a good model $M$ is of course an issue that pervades Bayesian statistics. As we note in \Cref{s:theory}, if the posterior $P_{M}(\cdot|\cX_n)$ is consistent, the choice of the density model $M$ will not majorly impact the discovery of the true clustering $\psi_{\lambda}(f_0)$ for large sample sample sizes. \Cref{fig:toy-challenge-compare,fig:toy-challenge-compare-high,fig:toy-challenge-densities} in \Cref{s:app-toy-challenge} demonstrate this effect. 
For smaller sample sizes, a thoughtful choice for $M$ (e.g.~a parametric mixture model with few components) can be used with our methodology to ensure that there is enough signal to detect true clusters. For high dimensional problems, leveraging on  \cite{chandra2023escaping}, one can use \texttt{BALLET} to find the level set clusters for a low-dimensional latent representation of the data. 

\begin{figure}[!ht]
    \centering
    \includegraphics[width=.9\textwidth]{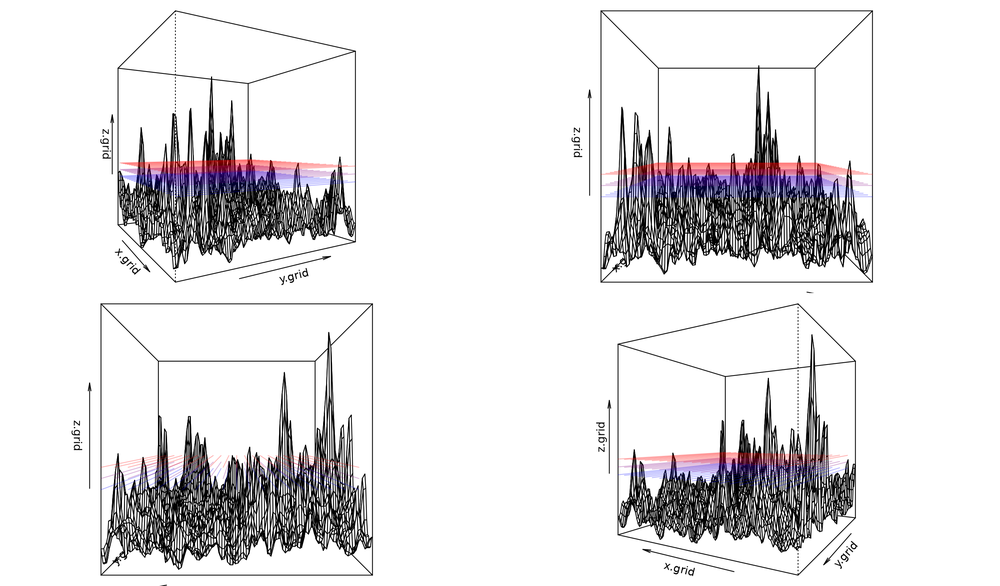}
    \caption{Visualizing our density estimate (plotted on the z-axis) for the Edinburgh-Durham Southern Galaxy Catalogue data. The colored lines mark the choice of different levels corresponding to the values of $c \in \{.8, 1,1.2\}$.}
    \label{fig:edsgc-density-viz}
\end{figure}

While level set clustering is a popular and conceptually appealing framework, a key practical challenge is the choice of the level $\lambda > 0$ \citep{campello2020density}.  %
Indeed, based on visualizing the density estimate for our sky survey data (\Cref{fig:edsgc-density-viz}), we expect our clusters to be sensitive to the exact value of the scientific constant $c$. To reduce sensitivity to $\lambda$, we describe a \emph{persistent} clustering approach in \SMref{s:persistent-clustering} that computes \texttt{BALLET} clusters for values of $c \in [.8, 1.2]$, visualizing these clusters with a cluster tree \citep{clustree}. This tree is then processed to infer clusters that remained active or \emph{persistent} across all the levels in the tree. This approach improved our specificity in detecting the two catalogs without losing sensitivity. 

While we have focused on level set clustering, our Bayesian density-based clustering framework is broad and motivates multiple directions for future work. One possibility is to avoid focusing on a single threshold $\lambda$, but instead estimate a cluster tree obtained by varying the threshold. Loss functions introduced by \cite{fowlkes1983method} provide a relevant starting point.
An alternative is to target a single clustering, but vary the threshold $\lambda$ over the observation space in a data-adaptive manner \citep{campello2015hierarchical}. Varying 
$\lambda$ is important in uncovering distinct cluster structures at varying levels of the density; refer, for example, to the illustrative example in \Cref{fig:level-set-pathological}.

Finally note that for a general non-parametric density $f$, it is hard to find a single notion of clustering that will be universally  appropriate across all applications. However, a natural notion at least when $f$ is sufficiently regular, may be that of modal clustering \citep{chacon2015population,menardi2016modal} that associates clusters with the domain of attraction of the modes of $f$. Interestingly, as recently argued in  \cite{arias2023unifying}, both level set clustering and modal clustering may fundamentally be the same approach.

\acks{This work was partially funded by grants R01-ES028804 and R01-ES035625 from the United States National Institutes of Health and N00014-21-1-2510 from the Office of Naval Research.
The authors would like to thank Dr. Woncheol Jang for kindly providing the data for our case study, and Dr. Cliburn Chan for suggesting applications in cosmology.}

\section*{Supplementary material}
\label{SM}
The accompanying supplementary materials contain additional details, including figures and tables referenced in the article starting with the letter `S'.  Code to reproduce our analysis can be found online at \url{https://github.com/davidbuch/ballet_article}.
\putbib[bibliography]
\end{bibunit}
\newpage

\originalJMLRappendix
\begin{center}
	\textbf{\Huge Supplementary Material for ``Bayesian Level Set Clustering"}
\end{center}

\setcounter{equation}{0}
\setcounter{figure}{0}
\setcounter{table}{0}
\setcounter{page}{1}
\setcounter{theorem}{0}
\setcounter{assumption}{0}

\renewcommand{\theequation}{S\arabic{equation}}
\renewcommand{\thefigure}{S\arabic{figure}}
\renewcommand{\thetable}{S\arabic{table}}
\renewcommand{\thesection}{S\arabic{section}}
\renewcommand{\thepage}{S\arabic{page}}
\renewcommand{\thetheorem}{S\arabic{theorem}}
\renewcommand{\thelemma}{S\arabic{lemma}}
\renewcommand{\theassumption}{S\arabic{assumption}}
\renewcommand{\bibnumfmt}[1]{[S#1]} %
\renewcommand{\citenumfont}[1]{S#1} %

\begin{bibunit}[apalike]

\section{Literature on Bayesian Clustering}
\label{s:related-work}

The last two decades have witnessed a significant maturation of the Bayesian clustering literature \citep{medvedovic2002bayesian, fritsch2009improved, wade2018bayesian, rastelli2018optimal, dahl2022search}. By designing and characterizing loss functions on partitions and developing search algorithms to identify partitions which minimize Bayes risk, these articles and others have established a sound framework for Bayesian decision-theoretic clustering. This literature acknowledges the \textit{cluster-splitting} problem alluded to in our preceding discussion, with \cite{wade2018bayesian} and \cite{dahl2022search} finding that clustering point estimates obtained by minimizing Bayes risk under certain parsimony-encouraging loss functions are less prone to cluster-splitting. 

However, these loss functions cannot completely eliminate the problem. \cite{guha2021posterior} shows that a fundamental cause of cluster splitting is that Bayesian mixture models converge to the mixture that has minimum Kullback-Leibler divergence to the true density. When the components of the mixture are not specified correctly, it may require infinitely many parametric components to recapitulate the true data-generating density. Thus, as data accumulate, it would seem futile to attempt to overcome the cluster-splitting problem merely by encouraging parsimony in the loss function. If the components are at all misspecified as data accumulate, eventually the preponderance of evidence will insist on splitting the clusters to reflect the multiplicity of parametric components. %
Indeed, in our illustrative example in Figure \ref{fig:mix-uniform-normal} (a) we used the parsimony-encouraging Variation of Information (VI) loss to obtain the Gaussian mixture model-based clustering point estimate. 

One response to this problem is the coarsened Bayes methodology of \cite{miller2018robust}, which only assumes the mixture model to be \textit{approximately} correctly specified. Another approach to mitigate the problem is to expand the class of mixture components \citep{fruhwirth2010bayesian, malsiner2017identifying, stephenson2019robust}. As we have claimed above, naive applications of this strategy can lead to loss of practical identifiability and computational challenges, although \cite{dombowsky2023bayesian} have had some success increasing component flexibility \textit{indirectly} by merging nearby less flexible mixture components in a post-processing step. The generalized Bayes paradigm, introduced by \cite{bissiri2016general}, also provides an answer to the cluster splitting problem via a loss-function-based Gibbs posterior for clustering \citep{rigon2020generalized}.

The idea of defining Bayesian clustering as a problem of computing a risk-minimizing summary $\psi$, of the posterior distribution on density $f$ can be viewed as related to the existing literature on decision-theoretic summaries of posterior distributions \citep{woody2021model, afrabandpey2020decision, ribeiro2018anchors}, though this literature has focused largely on extracting interpretable conclusions from posterior distributions on regression surfaces. In contrast, clustering in the manner we have proposed extracts an interpretable summary from a posterior distribution on the data-generating density. In addition, while the authors of that literature focus on the interpretability of summary functions $\psi$, we use the clustering example to emphasize that ideally $\psi$ should also be robust, in the sense that $\psi(f^*)$ will be close to $\psi(f)$ when $f^*$ is close to $f$, since this would suggest that small amounts of prior bias or model misspecification would not lead to large estimation errors.

\section{The lattice of sub-partitions}
\label{s:subpartition-lattice}
The space of sub-partitions $\ClustSpace[\cX]$  forms a lattice under the partial order given by  $\bC \preceq \bC'$ defined by the existence of a map $\phi: \bC \to \bC'$ such that $C \subseteq \phi(C)$ for each $C \in \bC$. One can check that $(\ClustSpace[\cX], \preceq)$ with join $\bC \vee \bC' \doteq \{C \cup C' | C \in \bC, C' \in \bC', C \cap C' = \emptyset\}$ and meet $\bC \wedge \bC' = \{C \cap C' | C \in \bC, C' \in \bC', C \cap C' = \emptyset\}$ is a lattice.

We denote $\bC \prec \bC'$ if $\bC \preceq \bC'$ but it is not the case that $\bC' \preceq \bC$. We can define a Hasse diagram for this lattice based on the relation $\bC \rightarrow \bC'$ if $\bC \prec \bC'$ but there is no $\bC'' \in \ClustSpace[\cX]$ such that $\bC \prec \bC'' \prec \bC'$. One can show that $\bC \rightarrow \bC'$ if and only if one of the following conditions hold:
\begin{itemize}
    \item $\bC'$ is obtained by merging two active clusters in $\bC$. That is, after suitable reordering: $\bC = \{C_1, \ldots, C_k\}$ and $\bC' = \{C_1 \cup C_2\} \cup \{C_r : r \in \{3, \ldots, k\}\}$.
    \item $\bC'$ is obtained by adding a noise point to its own cluster: i.e., $\bC' = \bC \cup \{{n}\}$ for some $n \in \cX$ that is not active in $\bC$.
\end{itemize}

This relation allows us to construct a Hasse diagram: a directed acyclic graph with nodes $\ClustSpace[\cX]$ and edges given by the relation $\rightarrow$. This diagram has the property that $\bC \prec \bC'$ if and only if there is a path from $\bC$ to $\bC'$. The Hasse diagram for the lattice of sub-partitions of $\cX = \{1,2,3\}$ is shown in \Cref{fig:hasse-subpart-lattice}.

\begin{figure}
    \centering
    \includegraphics[width=\textwidth]{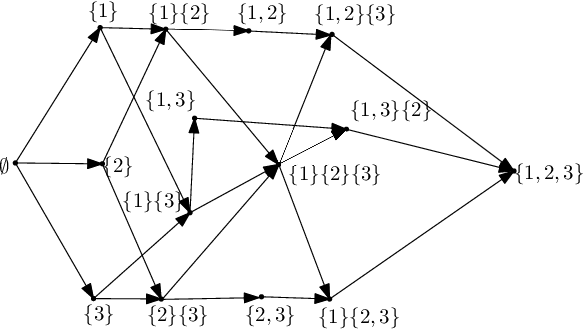}
    \caption{Hasse diagram for the lattice of sub-partitions $\ClustSpace[\cX]$ of the space $\cX = \{1,2,3\}$. This diagram has the property that $\bC \prec \bC'$ if and only if there is a path from $\bC$ to $\bC'$.}
    \label{fig:hasse-subpart-lattice}
\end{figure}

\section{\texttt{DBSCAN} and other level set clustering methods}
\label{s:comparison-to-dbscan}
Starting from works like \cite{hartigan1975clustering}, the topic of level set clustering has been extensively studied from the perspective of algorithms \citep{bhattacharjee2020density,campello2020density}, statistical methodology \citep{cuevas2000estimating,cuevas2001further,stuetzle2010generalized,scrucca2016identifying}, and statistical theory \citep{menardi2016modal,wang2019dbscan,steinwart2023adaptive}. Interestingly, while the popular \texttt{DBSCAN} algorithm \citep{ester1996density,schubert2017dbscan} has been around for a while, tools for its theoretical study are more recent \citep{sriperumbudur2012consistency,jiang2017density,wang2019dbscan}. Here we describe the \texttt{DBSCAN} algorithm and relate it to our surrogate clustering function $\psiclust(f)$, which we described in \Cref{ss:surrogate-clustering-function}  motivated by statistical theory.

The \texttt{DBSCAN} algorithm finds arbitrary shaped clusters of related data points in large spatial databases \citep{ester1996density}. The \texttt{DBSCAN} \emph{cluster model} \cite[Section 2.1]{schubert2017dbscan} is not explicitly described in terms of the data density $f_0$, but rather in terms of a notion of distance $\dist(x_i,x_j)$ measuring relatedness between observations $x_i,x_j \in \cX_n$ and two free parameters $\Eps > 0$ and $\MinPts \in \nat$. A data point $x \in \cX_n$ is called a \emph{core} point if it has at least $\MinPts$ many neighbors $N_{\Eps}(x) \doteq \{y \in \cX_n : \dist(x,y) \leq \Eps\}$ that are within a distance $\Eps$ of it (i.e.~$|N_{\Eps}(x)| \geq \MinPts$). The set of all \emph{core} points $\mathcal{A} = \{x \in \cX_n :|N_{\Eps}(x)| \geq \MinPts \}$ are then   clustered based on the partition induced by the transitive closure of the relation $\{(x,y) \in \mathcal{A} \times \mathcal{A} : \dist(x,y) \leq \Eps\}$. In words, the \texttt{DBSCAN} clustering of $\mathcal{A}$ is the finest partition of $\mathcal{A}$ where each pair of points $x, y  \in \mathcal{A}$ satisfying $\dist(x,y) \leq \Eps$ are clustered together. While the \texttt{DBSCAN} algorithm goes on further to add some of the \emph{non-core} points (called \emph{border} points) that lie within a neighborhood $N_{\Eps}(x)$ of some core point $x \in \mathcal{A}$ to a corresponding cluster, for consistency with level set clustering, this step is avoided by a variant of the algorithm  called \texttt{DBSCAN}$^*$ \cite{campello2015hierarchical}.

When $\cX = \R^{d}$ and $\dist(x,y) = \|x-y\|$ is Euclidean distance, the notion of \emph{core points} from \texttt{DBSCAN} is seen to be related to the notion of \emph{core} or \emph{active} points that we introduced in \Cref{ss:level-set-clusters}. In fact, and as indicated in \cite{sriperumbudur2012consistency,jiang2017density,campello2015hierarchical}, the clustering from \texttt{DBSCAN}$^*$ is the same as our surrogate clustering $\psiclust(\hat{f}_\delta) \in \ClustSpace[\cX_n]$ where $\hat{f}_\delta(x) = n^{-1} \sum_{x_i \in \cX_n} \kappa_\delta(x_i - x)$ is the kernel density estimate based on the uniform kernel $\kappa_\delta(z) = \I{\|z\| \leq \delta}/(v_d \delta^d)$ and $v_d=\frac{\pi^{d/2}}{\Gamma(d/2+1)}$ is the volume of the $d$-dimensional unit Euclidean ball. Here $\delta=\Eps$ and $\lambda=\MinPts/(nv_d \delta^d)$ can be expressed in terms of the original \texttt{DBSCAN} parameters $\Eps > 0$ and $\MinPts \in \nat$. In fact, as noted in \cite{campello2020density}, there is also another representation of \texttt{DBSCAN}$^*$ as $\psiclust(\hat{f}_{k})$ where $\hat{f}_k(x) = \frac{k}{n v_d}  \delta_k(x)^{-d}$ is the $k$-nearest neighbor density estimator \citep{biau2015lectures} with $\lambda=\frac{k}{nv_d} \delta^{-d}$,  $\delta = \Eps$, and $k=\MinPts$. 
\begin{remark}
From the first formulation $\psiclust(\hat{f}_\delta)$ the parameter $\Eps = \delta$ for \texttt{DBSCAN} simultaneously controls both the regularity of the kernel density estimator $\hat{f}_{\delta}$ used to discover core points $\mathcal{A} = A_{\lambda,\hat{f}_{\delta}}$ and also the connectivity of resulting clusters based on the connectivity of the graph $G_{\delta}(\mathcal{A})$. This is in contrast to \texttt{BALLET} where the parameter $\delta$ only controls the connectivity of the clusters, and may explain why \texttt{BALLET} clustering was seen to be less sensitive to the choice of this parameter in \Cref{fig:ballet-dbscan-compare-param-sensitivity}.
\end{remark}

\subsection{Time complexity of evaluating surrogate function}
The time complexity of evaluating $\psiclust(f)$ at any fixed $f$ is comparable to that of the \texttt{DBSCAN} algorithm and an additional time complexity $\kappa_n$ of evaluating $f$ at all of the points in $\cX_n$. Suppose first that the $\delta$ neighborhood graph for all the data points $G_\delta(\cX_n)$ can be pre-computed and stored for future use in an adjacency list representation \cite[Chapter 3]{sanjoy2008algorithms}. In order to evaluate $\psiclust(f)$, one can then (i) calculate the set of active nodes $\Activef \subseteq \cX_n$ by evaluating $f$ at all the data points, (ii) extract the subgraph $G_\delta(\Activef)$ of $G_\delta(\cX_n)$ by scanning the precomputed adjacency list, and (iii) compute the connected components of $G_\delta(\Activef)$ by using the standard breadth (or depth) first search algorithm \citep[Chapter 3]{sanjoy2008algorithms}. Thus, given our precomputed adjacency list representation of $G_\delta(\cX_n)$, the time complexity to evaluate $\psiclust(f)$ is $O(\kappa_n + |G_\delta(\cX_n)|)$ where $|G_\delta(\cX_n)|$ is the sum of the number of edges and vertices in $G_\delta(\cX_n)$. The time complexity of pre-computing the graph $G_\delta(\cX_n)$ is at most that of running the \texttt{DBSCAN} algorithm up to constant multiples. Indeed, $G_\delta(\cX_n)$ can be constructed by performing a range query for each point $x_i \in \cX_n$ to discover the set of points $B(x_i,\delta) \cap \cX_n$; however, this sequence of range queries is also an essential part of the \texttt{DBSCAN} algorithm \citep[see][]{schubert2017dbscan} which would thus also require as many steps. %

\section{The BALLET optimization algorithm}
\label{s:ballet-search-algorithm}
For any sub-partition $\bC = \{C_1, \ldots, C_k\}$ of $\{x_1, \ldots, x_t\}$, we use an equivalent allocation vector representation $\bc=(c_1, \ldots, c_t) \in \{0,1, \ldots, k\}^t$ given by $c_i = h$ if the point $x_i$ belongs to the cluster $h$, i.e. $x_i \in C_h$, and $c_i = 0$ if the point $x_i$ is classified as noise under this sub-partition, i.e. $x_i \in \{x_1, \ldots, x_t\} \setminus \cup_{h=1}^k C_h$. 

Given Monte Carlo samples $\{f^{(s)}\}_{s=1}^S$ from the posterior distribution $P_M( \cdot | \cX_n)$, we first compute the clusterings  $\bC^{(s)} = \psiclust(f^{(s)}) \in \ClustSpace[\cX_n]$ and their allocation vectors $\bc^{(s)}=(c^{(s)}_1, \ldots, c^{(s)}_n)$ for each $s \in \{1 \ldots S\}$. Next, these allocation vectors are  used to precompute the probability estimates in \eqref{eq:ia-binder-risk}, namely
\begin{eqnarray}
\hat{\pi}^{(1)}_{i,j} & \doteq & S^{-1} \sum_{s=1}^S \1{(c^{(s)}_i \neq 0, c^{(s)}_j \neq 0, c^{(s)}_i = c^{(s)}_j)},\qquad \hat{\pi}^{(2)}_{i,j} \doteq  S^{-1} \sum_{s=1}^S \1{(c^{(s)}_i \neq 0, c^{(s)}_j \neq 0, c^{(s)}_i \neq c^{(s)}_j)}, \nonumber \\ \hat{\alpha_i} & \doteq &  S^{-1} \sum_{s=1}^S \1{(c^{(s)}_i \neq 0)} \nonumber 
\end{eqnarray}
for each $i \neq j \in \{1, \ldots, n\}$.
With this, the optimization problem in \eqref{eq:ballet-estimator} reduces to minimizing the risk
\begin{equation}
\label{eq:empirical-risk}
\begin{aligned}
R(\bc') = &(n-1) \bigg\{ m_{ai} \sum_{i=1}^n \1{(c'_i = 0)} \hat{\alpha_i} + m_{ia} \sum_{i=1}^n \1{(c'_i \neq 0)} (1-\hat{\alpha_i}) \bigg\} \\
            &\qquad + \sum_{1 \leq i < j \leq n} \1{(c'_i \neq 0, c'_j \neq 0)} \big\{a \hat{\pi}^{(1)}_{i,j} \1{(c'_i \neq c'_j)} + b \hat{\pi}^{(2)}_{i,j} \1{(c'_i =  c'_j)} \big\}
\end{aligned}
\end{equation}
over all allocation vectors $\bc'=(c'_1, \ldots, c'_n)$ corresponding to sub-partitions $\bC' \in \ClustSpace[\cX_n]$. 

Although exact minimization over the combinatorial space $\ClustSpace[\cX_n]$ of sub-partitions is computationally intractable, we can adapt heuristic algorithms for approximate minimization over the related space of partitions of $\cX_n$ \citep[e.g.,][]{fritsch2009improved, rastelli2018optimal}. Particularly, we consider the algorithm of \citet{dahl2022search} that, given a candidate partition of $\cX_n$, provides two important ways to compute a candidate set of partitions that may have a smaller objective value: (i) a series of incremental update steps called the \emph{sweetening phase} that reassigns each data point $x_i$ (chosen in a random order) to a different cluster if doing so will decrease the objective, and (ii) a series of major update steps called the \emph{zealous update phase} that repeatedly destroys a randomly chosen cluster and then incrementally reallocates the data points if doing so decreases the objective. Starting from an initial partition that is either selected at random or is built incrementally to have a small objective value, the algorithm of \citet{dahl2022search} improves the initial partition using \emph{sweetening phase} followed by \emph{zealous update phase}. This entire process is repeated (in parallel) many times, and the partition with the least objective value among all the explored partitions is reported.

The main primitive operation needed to implement the above algorithm is to incrementally find a low-risk partition including a new data point (say $x_{t+1}$ for $t \in \{1, \ldots, n-1\}$) that respects a given low-risk partition $\{C_1, \ldots, C_k\}$ of some existing set of data points, say $\{x_1, \ldots, x_t\}$. Indeed, the following two kinds of such partitions of $\{x_{1}, \ldots, x_{t+1}\}$ are possible: (a) the new point $x_{t+1}$ is added to its own cluster; this is the partition $\{C_1, \ldots, C_k, \{x_{t+1}\}\}$, or (b) the new point is added to one of the existing clusters (say $C_h$); this is the partition $\{C_1, \ldots , C’_h, \ldots, C_k\}$, where $C’_h = C_h \cup \{x_{t+1}\}$. For each of these $k+1$ partitions, \citet{dahl2022search} recommend evaluating the objective value restricted only to the data points under consideration (i.e. sum only over terms $i, j \in \{1, \ldots, t+1\}$ in our empirical risk \eqref{eq:empirical-risk}) and selecting the partition with the smallest risk among the $k+1$ candidates. 

The aforementioned primitive operation is easily extended to the case of sub-partitions of $\cX_n$. Indeed, suppose $\bC = \{C_1, \ldots, C_k\}$ is a sub-partition of $\{x_1, \ldots, x_{t}\}$. The sub-partition $\bC'$ of $\{x_1, \ldots, x_{t+1}\}$ respects $\bC$ in the following three possible ways: a) the point $x_{t + 1}$ is assigned to the noise cluster; this is just the sub-partition $\bC' = \{C_1, \ldots, C_k\}$ in our notation, b) the point $x_{t + 1}$ is assigned to its own cluster; this is the sub-partition $\bC' = \{C_1, \ldots, C_k, \{x_{t+1}\}\}$, and (c) the point $x_{t+1}$ is assigned to an existing cluster (say $C_h$); this is the sub-partition $\bC' = \{C_1, …, C’_h, …, C_k\}$ where $C’_h = C_h \cup \{x_{t+1}\}$. We then evaluate our risk \eqref{eq:empirical-risk} restricted to the indices $i, j \in \{1, \ldots, t+1\}$ using the allocation vector $\bc' = (c_1, \ldots, c_{t+1})$ corresponding to $\bC'$, and select the sub-partition with the smallest risk among the $k+2$ candidates. This primitive operation allows us to implement the initialization, sweetening, and zealous update phases in the \cite{dahl2022search} algorithm to minimize our risk \eqref{eq:empirical-risk} over allocation vectors that correspond to all sub-partitions of $\cX_n$.  Notably, in the \emph{zealous update phase} the cluster to be destroyed can either be the current noise cluster or one of the current non-noise clusters.

\subsection{Avoiding optimization: \texttt{BALLET} decision theoretic vs plugin estimator?}
\label{ss:plugin-vs-decision-theoretic}

Recall the heuristic \texttt{BALLET} \emph{plugin} estimate $\hat{\bC} = \psi_{\lambda, \delta}(\hat{f})$ that avoids the expensive optimization in \eqref{eq:ballet-estimator} by directly computing the level set clusters of the posterior mean density $\hat{f}(x) \approx \frac{1}{S}\sum_{s = 1}^S f^{(s)}(x)$. In most cases, the plugin clustering estimate will be similar to the decision theoretic \texttt{BALLET} estimator from \eqref{eq:ballet-estimator} when the posterior uncertainty of $f$, and particularly that of the level set $\{f \geq \lambda\}$, is low. We note this in our results from \Cref{s:real-data-analysis} (see~\Cref{tab:synthdata-tabular-results,tab:edsgc-edcci-coverage-full,tab:edsgc-abell-coverage-full}).

However we now illustrate that the two estimators will at times produce different answers because the heuristic plugin estimate does not take into consideration the posterior uncertainty of $f$, which may be substantial. Indeed, by modifying our simple example from \Cref{fig:mix-uniform-normal}, we see differences emerge when the level $\lambda$ is increased to the point that there is non-trivial posterior uncertainty in the induced level set $\{f \geq \lambda\}$ (\Cref{fig:ballet_vs_plugin}).

As a general principle, we recommend the use of Bayes estimators that directly target the quantity of interest, rather than a two-stage plugin approach, where a Bayes estimator is computed for an intermediate quantity. Indeed, there are many examples in the literature in which two-stage plugin approaches are suboptimal.

\begin{figure}
    \centering
    \includegraphics[height=.9\textheight]{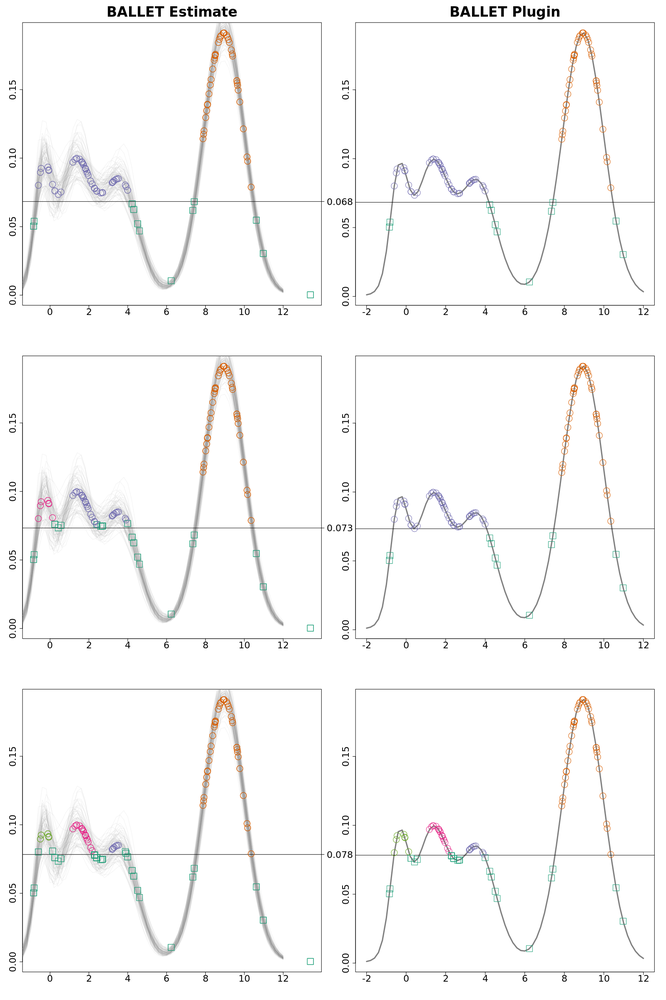}
    \caption{\texttt{BALLET} estimator \eqref{eq:ballet-estimator}  accounts for the posterior uncertainty of $f$ (left) while the plugin estimator (right) does not. The estimates start to differ when level $\lambda$ in \Cref{fig:mix-uniform-normal} is increased so that there is non-trivial posterior uncertainty in the level set $\{f \geq \lambda\}$.}
    \label{fig:ballet_vs_plugin}
\end{figure}

\section{Additional results from analysis of the illustrative challenge datasets}
\label{s:app-toy-challenge}

In this section, we present additional results from the analysis of the illustrative challenge datasets. In \Cref{fig:toy-challenge-data} we visualize the three datasets, and in \Cref{fig:toy-challenge-densities} we show heat maps of the log of the posterior expectation of the data generating density $f$ under three different models: a Dirichlet process mixture of Gaussian distributions (\texttt{DPMM}), an adaptive P\'olya tree model, and a nearest-neighbor Dirichlet mixture model.

In analyzing these datasets, our choice of loss parameters $\lambda$ for \texttt{BALLET} was guided by the discussion in \Cref{s:ballet-tuning-parameter}. In particular, we tuned $\lambda$ to achieve a certain noise level $\nu \in (0,1)$, and given $\nu$ (and thus $\lambda$) the parameter $\delta$ was automatically chosen using the data adaptive procedure in \Cref{ss:surrogate-clustering-function} with our default choice of $k=\lceil \log n \rceil$. Here, $n$ is the sample size of the dataset under consideration.

We describe the clustering results using \texttt{BALLET} for various choices of noise level $\nu$.
In \Cref{fig:toy-challenge-compare,fig:toy-challenge-compare-high} we compare \texttt{BALLET} clustering estimates obtained under our three density models for two different noise levels $\nu \in \{ 5\%, 10\%\}$. 
The \texttt{BALLET} upper and lower bounds for the RNA-seq data corresponding to noise levels $\nu \in \{5\%, 10\%\}$ are shown in \Cref{fig:tsne-ballet-bounds-lvls}.
The persistent clusters (see \Cref{s:persistent-clustering}) across the noise levels $\nu \in \{5\%,10\%,15\%\}$ for the RNA-seq data are shown in \Cref{fig:tsne-persistent-clustering}. We note that the persistent clusters are somewhat qualitatively different across the density models, demonstrating that the choice of prior can have an effect on the nature of clusters that are discovered.

Finally, we also explore an automatic choice of $\nu$ for the various  datasets and density models using the elbow heuristic mentioned in \Cref{s:ballet-tuning-parameter}. The elbow plots describing the selection of $\nu$ are shown in \Cref{fig:toy-challenge-elbow-plots}, while the corresponding clusters are shown in \Cref{fig:toy-challenge-compare-elbow}.

\begin{figure}
    \centering
    \includegraphics[width=0.9\textwidth]{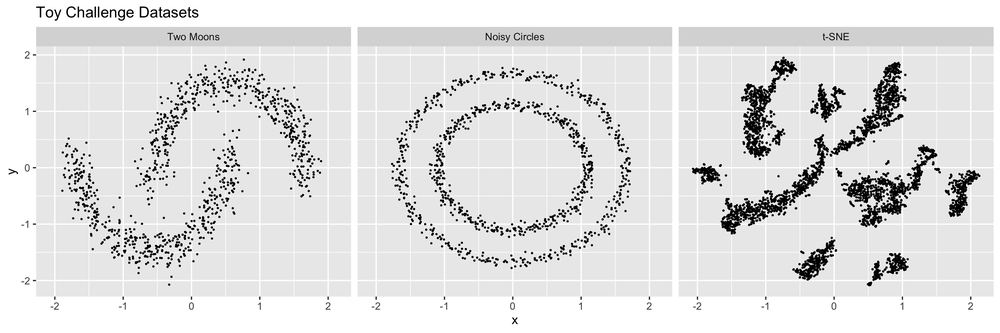}
    \caption{Plots of the three illustrative challenge datasets. From left to right: two moons simulated data, noisy circles simulated data, and a t-SNE embedding of a RNA-seq dataset.}
    \label{fig:toy-challenge-data}
\end{figure}

\begin{figure}
    \centering
    \includegraphics[width=0.9\textwidth]{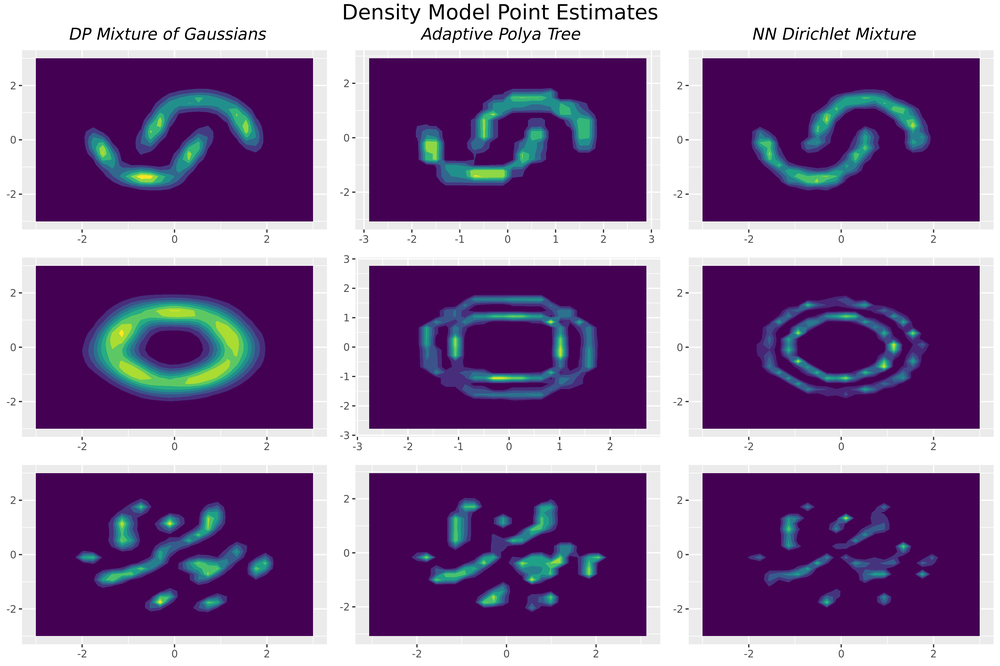}
    \caption{Plots of posterior point estimates of the data-generating densities for each of three illustrative challenge datasets under three different models for the unknown density.}
    \label{fig:toy-challenge-densities}
\end{figure}

\begin{figure}
    \centering
    \includegraphics[width=0.9\textwidth]{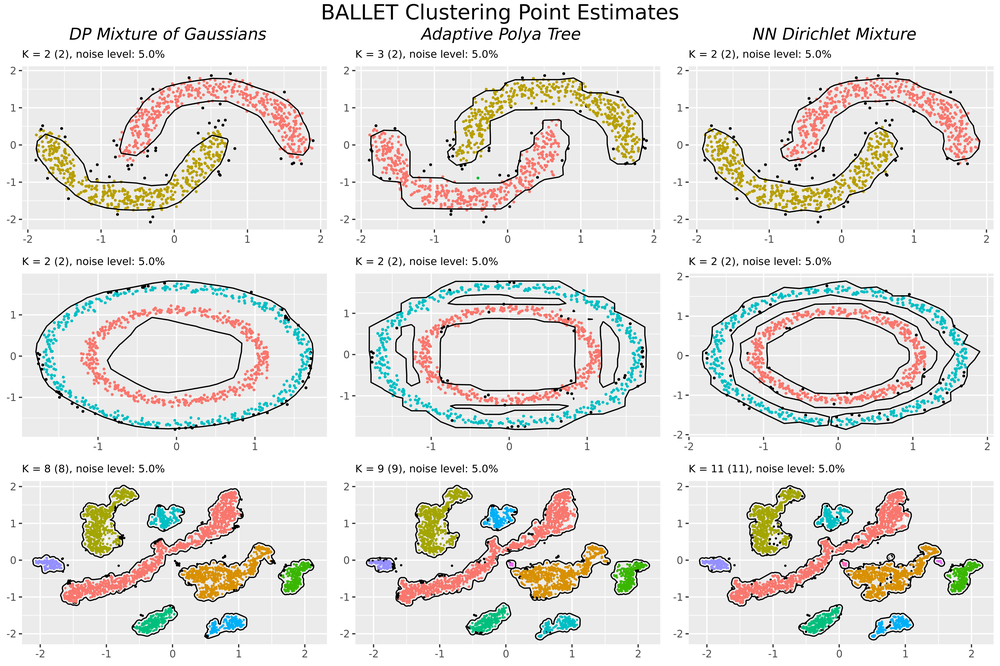}
    \caption{Comparison of \texttt{BALLET} clustering point estimates obtained under the three different density models shown in Figure \ref{fig:toy-challenge-densities} with $\nu=5\%$ noise points. The cardinality of the sub partition is displayed in the title of each plot, as $K = X$, and it is followed, in parentheses by the count of clusters with more than 1 observation.}
    \label{fig:toy-challenge-compare}
\end{figure}

\begin{figure}
    \centering
    \includegraphics[width=0.9\textwidth]{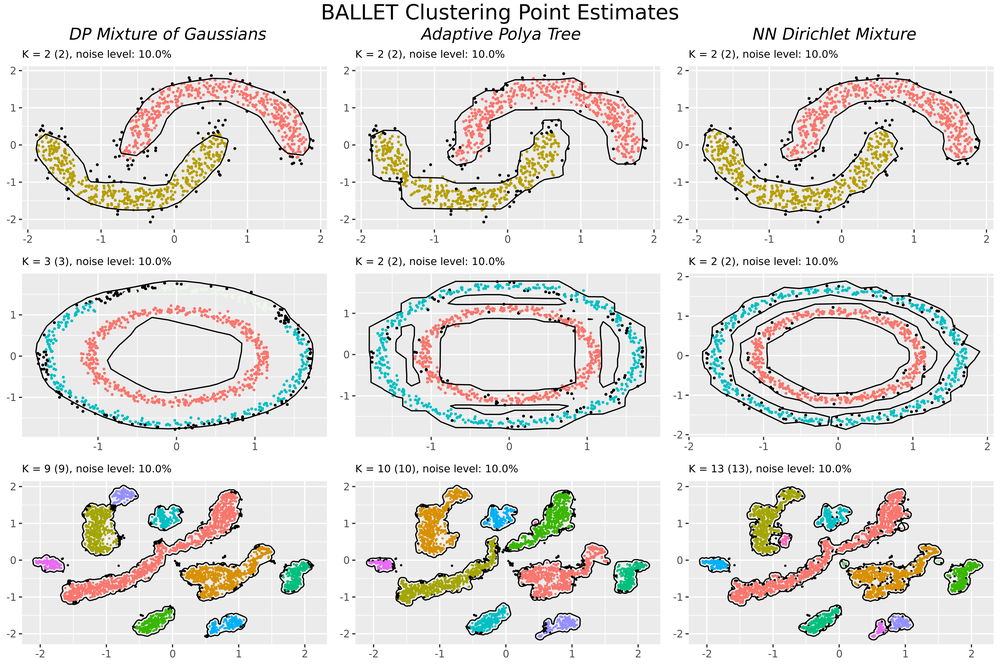}
    \caption{Comparison of \texttt{BALLET} clustering point estimates obtained under the three different density models shown in Figure \ref{fig:toy-challenge-densities} with $\nu=10\%$ noise points. Compared to \Cref{fig:toy-challenge-compare}, some clusters in second and third rows are seen to split into further clusters based on our choice of the density model. While this may be desirable in the RNA-seq dataset in the last row, increasing the density level does not seem desirable for the Noisy Circles dataset in the second row.}
    \label{fig:toy-challenge-compare-high}
\end{figure}

\begin{figure}
	\centering
	\includegraphics[width=.9\textwidth]{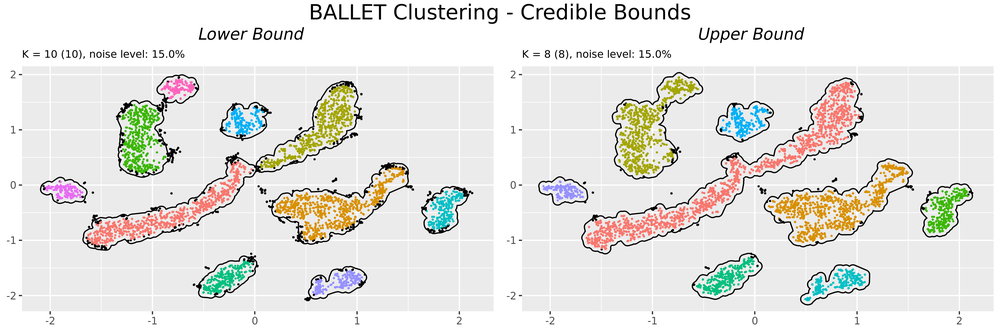}
	\caption{Upper and lower bounds for the 95\% credible ball centered at our \texttt{BALLET} clustering estimate for the RNA-seq data, fit with the \texttt{DPMM} model for $f$. The cardinality of the partition is displayed in the title of each plot, as $K = X$, and it is followed, in parentheses by the count of clusters with more than 1 observation, and the percentage ($\nu=15\%$) of noise points  based on our chosen  level $\lambda$. \Cref{fig:tsne-ballet-bounds-lvls} in \Cref{s:app-toy-challenge} shows additional results for different choices of $\lambda$.}
	\label{fig:tsne-ballet-bounds}
\end{figure}

\begin{figure}
	\centering
	\includegraphics[width=.9\textwidth]{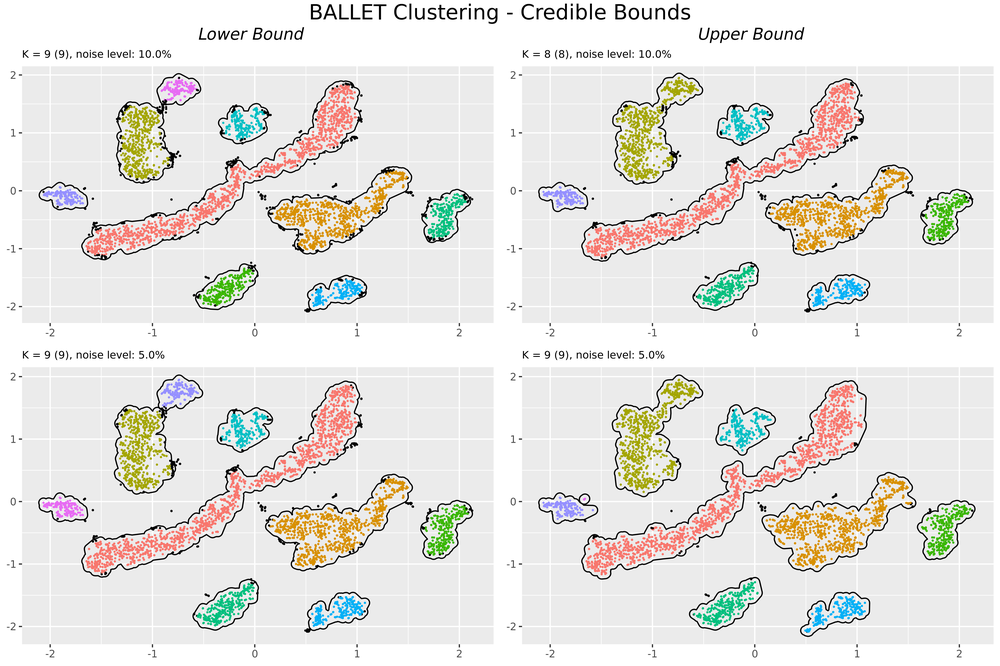}
	\caption{The \texttt{BALLET} upper and lower bounds in  \Cref{fig:tsne-ballet-bounds} for different choices of the level  $\lambda$, as specified in the subplot titles.}
	\label{fig:tsne-ballet-bounds-lvls}
\end{figure}

\begin{figure}
	\centering
	\includegraphics[width=.9\textwidth]{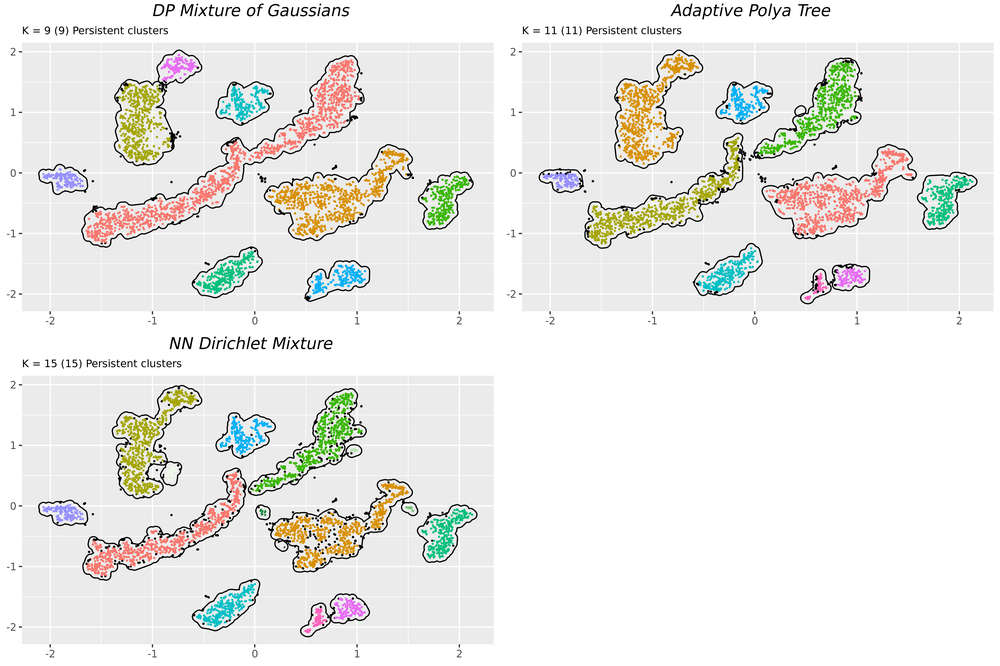}
	\caption{ The persistent clusters (see \Cref{s:persistent-clustering}) across the three density models for the RNA-seq data after applying \texttt{BALLET} with noise levels $\nu \in \{5\%, 10\%, 15\%\}$.}
	\label{fig:tsne-persistent-clustering}
\end{figure}

\section{The mixture of histograms model for densities}
\label{s:random-histogram-model}
This section describes the \emph{mixture of histograms} model that we use to estimate the data generating density in \Cref{s:real-data-analysis}. This model can quickly be fit to a large number of data points since the fitting is primarily based on counting the number of observed data points that  fall into various bins. Further, in contrast to a standard histogram model, the density function from a mixture of histograms tends to be more regular in the sense of having smaller jumps. 

Let us introduce the notation to describe our model.  Suppose $x_i$ for $i = 1, \dots, n$ are independent draws from an unknown distribution with density $f$ supported on a compact set $\cX \subseteq \R^2$. We assume that $f$ can be represented as a finite mixture $f(x; \vec{\pi}, \vec{\mathcal{B}}, \vec{\rho}) = \sum_{k = 1}^K \pi_k H_k(x; \mathcal{B}_k, \vec{\rho}_k)$ of $K \in \nat$ histogram densities, where $\vec{\pi} = (\pi_1, \ldots, \pi_K)$ is a vector of non-negative weights whose coordinates sum to one. For a given $k \in [K]$, the histogram density $H_k(x; \mathcal{B}_k, \vec{\rho}_k) = \sum_{m = 1}^{M} \1{(x \in B_{km})} \rho_{km}$ is a step-function based on a partition $\mathcal{B}_k = \{B_{k1}, \dots, B_{kM}\}$ of size $M$ of $\cX$ and a set of associated density values $\vec{\rho}_k = (\rho_{km})_{m=1}^M$. For simplicity, we fix $|\mathcal{B}_k| = M$ for all $k = 1, \ldots, K$. 

It is convenient to view this model in terms of an equivalent augmented-data representation, associating a latent variable $Z_i$ with each observation $x_i$, so that $f(x_i; Z_i, \vec{\mathcal{B}}, \vec{\rho}) = \sum_{k = 1}^K \1{(Z_i = k)} H_k(x_i; \mathcal{B}_k, \vec{\rho}_k)$ and $\Pr(Z_i = k | \vec{\pi}) = \pi_k$ for each $k \in \{1, \ldots, K\}$. We denote the complete set of observations as $\mathcal{D} = \{x_1, \dots, x_N\}$ and the latent histogram allocation variables as $\mathcal{Z} = \{Z_1, \dots, Z_N\}$. 

For simplicity, we also assume that $\cX = [a, b] \times [c,d]$ and $\mathcal{B}_k$ is a grid (or product) based partition of $\cX$. More precisely, we assume that there is a partition $\mathcal{U}_k = \{U_{k1}, \dots, U_{kM'}\}$ of $[a,b]$ and $\mathcal{V}_k = \{V_{k1}, \dots, V_{kM'}\}$ of $[c,d]$ so that $\mathcal{B}_k = \{U \times V | U \in \mathcal{U}_k, V \in \mathcal{V}_k \}$ and $M = M'^2$. We further assume that partitions $\mathcal{U}_k$, $\mathcal{V}_k$ are constructed based on grid points   $\vec{u}_k = \{u_{k0}, \dots, u_{kM'}\}$, $\vec{v}_k = \{v_{k0}, \dots, v_{kM'}\}$ such that $U_{k1} = [u_{k0}, u_{k1}]$, $V_{k1} = [v_{k0}, v_{k1}]$ and $U_{km} = (u_{k,m-1}, u_{k,m}]$ and $V_{km} = (v_{k,m-1}, v_{k,m}]$ for $2 < m \leq M'$.

\subsection{Prior distribution on parameters}

We now describe our prior distribution for the parameters of the mixture of histograms model. %
Focusing first on the partition $\mathcal{B}_k$, denote $u_{km} = a + (b - a)\sum_{j = 1}^m u'_{kj}$ and $v_{km} = c + (d - c)\sum_{j = 1}^m v'_{kj}$ so that $\vec{u}'_k = (u'_{k1}, \ldots, u'_{kM'})$ and $\vec{v}'_k = (v'_{k1}, \ldots, v'_{kM'})$ lie on the probability simplex.  We specify our prior on  $\mathcal{U}_k$ and $\mathcal{V}_k$ (and thus  $\mathcal{B}_k$) by assuming that   $\vec{u}'_k \sim \text{Dirichlet}(\alpha_b 1_{M'})$ and $\vec{v}'_k \sim \text{Dirichlet}(\alpha_b 1_{M'})$ are independent. The parameters $M'$ and $\alpha_b$ can be thought of as controlling the bin resolution and regularity for the histograms, respectively. In our sky survey analysis we set $M' = 50$ ($M = 2500$) and $\alpha_b = 5$.%

 After specifying our prior for $\mathcal{B}_k$, we  complete our prior specification for the histogram $H_k$ by describing our prior for $\vec{\rho}_k$ given $\mathcal{B}_k$. Since $H_k$ is a density that integrates to one, $\vec{\rho}_k$ should satisfy the constraint $\sum_{m=1}^M \rho_{km} A_{km} = 1$ where $A_{km}$ denotes the Lebesgue measure of bin $B_{km}$. Thus, rather than directly placing a prior on $\vec{\rho}_k$, we place a Dirichlet prior on the parameter $\vec{p}_k = (p_{k1}, \ldots, p_{kM})$, where $p_{km} = A_{km}\rho_{km}$ denotes the probability mass assigned to bin $B_{km}$ by the histogram $H_k$. Thus we suppose $\vec{p}_k | \mathcal{B}_k \sim \text{Dirichlet}(\alpha_d \frac{A_{k1}}{A}, \dots, \alpha_d\frac{A_{kM}}{A})$, choosing $\alpha_d = 1$ as a default. 

Finally, we complete our prior specification on the mixture of histograms model for the unknown density $f$ by choosing to treat all parameters $\{\{\mathcal{B}_1, \vec{\rho}_1\}, \dots, \{\mathcal{B}_K, \vec{\rho}_K\}\}$ of the $K$ histograms as \textit{a priori} independent and fixing the weights $\vec{\pi} = \{\frac{1}{K}, \dots, \frac{1}{K}\}$. In our sky survey analysis we set $K = 50$.

\subsection{Fast posterior sampling by clipping dependence}

We are interested in quickly sampling from the posterior distribution of the density $f \,| \, \mathcal{D}$ when the number of observations $n$ is large. Typically, one would draw samples from the joint posterior $\{\{\mathcal{B}_1, \vec{\rho}_1\}, \dots, \{\mathcal{B}_K, \vec{\rho}_K\}\}, \mathcal{Z} \,|\, \mathcal{D}$, and then, marginalizing over the uncertainty in $\mathcal{Z}$, use the samples of the histogram parameters to construct a posterior on $f$. A sampling algorithm designed to converge to this high-dimensional joint posterior object would be extremely computationally intensive, especially given our large sample size, and would likely require an unacceptably large number of samples to converge. Hence, we simplify inferences via a modular Bayes approach similar to that in
\cite{liu2009modularization}.

Specifically, to update $\vec{\mathcal{B}} = \{\mathcal{B}_1, \ldots, \mathcal{B}_K\}$, we sample from its prior distribution rather than its conditional distribution given the data and other parameters, effectively clipping the dependence of the bin parameters on the other components of the model as described in \cite{liu2009modularization}. Furthermore, we draw only one sample $\vec{\mathcal{B}}^* = \{\mathcal{B}^*_1, \dots, \mathcal{B}^*_K\}$ from the prior distribution on $\vec{\mathcal{B}}$, and reuse this same collection $\vec{\mathcal{B}}^*$ of histogram bins for each round of new samples for the other parameters. %

In addition, rather than iterate between sampling $\vec{p}_k$ from its full conditional, 
\begin{equation*}
    \vec{p}_k | \mathcal{D}, \mathcal{Z}, \mathcal{B}^*_k \sim \text{Dirichlet}(\sum_{i = 1}^N \1{(x_i \in B_{k1})}\1{(Z_i = k)} + \alpha_d \frac{A_{k1}}{A}, \dots, \sum_{i = 1}^N \1{(x_i \in B_{kM})}\1{(Z_i = k)} + \alpha_d\frac{A_{kM}}{A}),
\end{equation*}
and alternately sampling $\mathcal{Z}$ from its full conditional, we marginalize the log density of $\vec{p}_k | \mathcal{D}, \mathcal{Z}, \mathcal{B}^*_k$ with respect to the prior distribution on $\mathcal{Z}$ yielding the distribution   
\begin{equation}
    \vec{p}_k | \mathcal{D}, \mathcal{B}^*_k \sim \text{Dirichlet}(\frac{N_{k1}}{K} + \alpha_d \frac{A_{k1}}{A}, \dots, \frac{N_{kM}}{K}+ \alpha_d\frac{A_{kM}}{A}), 
\end{equation}
which we use in place of the posterior distribution of $\vec{p}_k$ given  $\mathcal{B}^*_k$ and $\mathcal{D}$. Here $N_{km} = \sum_{i = 1}^N \1{(x_i \in B^*_{km})}$ denotes the number of observations that fall into the bin $B^*_{km} \in \mathcal{B}^*_k$.

The resulting algorithm is a fast way to generate independent samples from an approximate modular posterior for 
$f(\mathcal{D})$. This sampler runs almost instantaneously on a personal laptop computer even for sample sizes of $n \approx 40,000$, which would be prohibitive for traditional Markov chain Monte Carlo algorithms for  
density estimation models. Moreover, the samples appear to appropriately reflect our uncertainty in the underlying data-generating density in our experiments.

\section{Additional results from the analysis of the synthetic sky survey data}
\label{s:app-synthetic-sky-survey}

Including a diversity of sizes among the synthetic galaxy clusters led to datasets that more closely resembled the observed data, and it also made the true clusters more challenging to recover with both clustering methods. Hence, we simulated the weights of the active components from a symmetric Dirichlet distribution with a small concentration parameter. The relative weights of the ``galaxy clusters" for one of the 100 synthetic datasets we analyzed are visualized in Figure \ref{fig:synthdata-component-weights}. The specific synthetic data set associated with these weights is shown in Figure \ref{fig:synthdata-data}.

\begin{figure}
    \centering
    \includegraphics[width=0.9\textwidth]{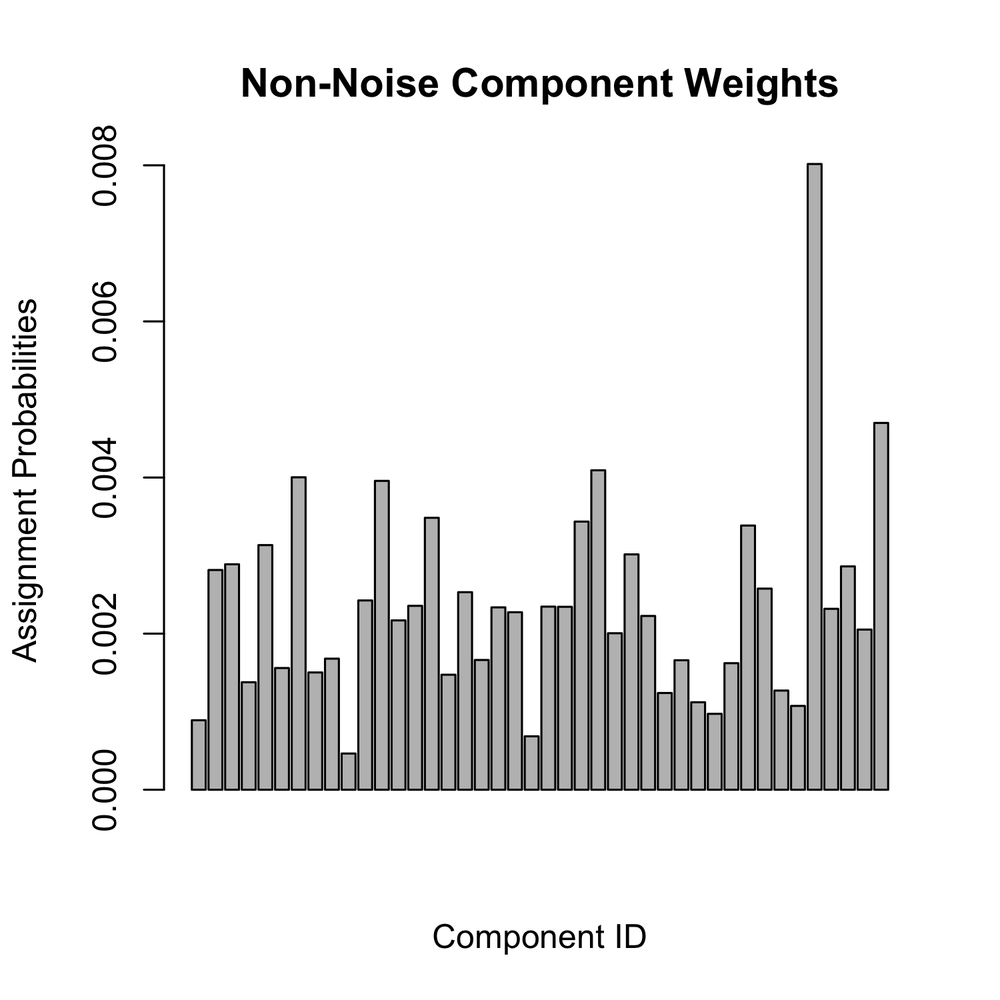}
    \caption{Relative sizes (mixtures weights) of the non-noise components in one of our synthetic sky survey datasets.}
    \label{fig:synthdata-component-weights}
\end{figure}

\begin{figure}
    \centering
    \includegraphics[width=0.9\textwidth]{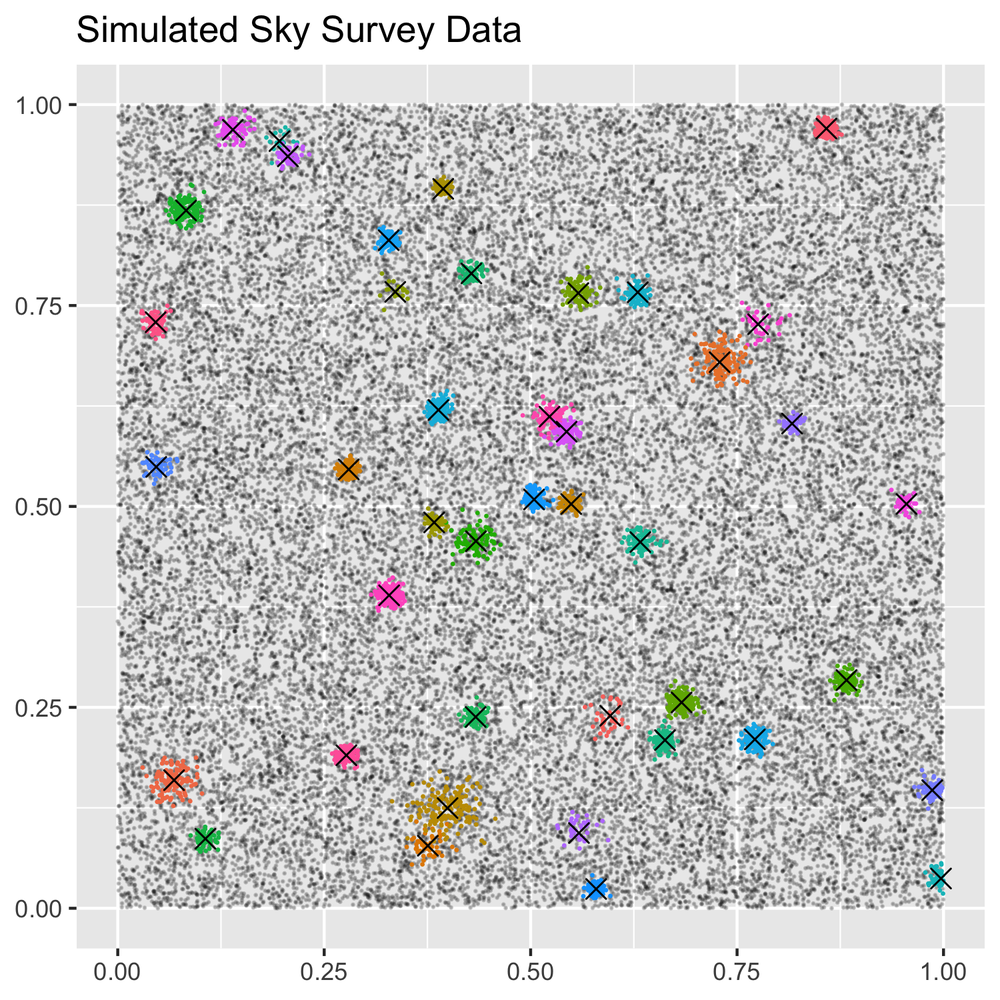}
    \caption{One of our synthetic sky survey datasets. Observations drawn from one of the high-density components are given bright colors, and each of their centers is marked with an $\times$. Observations drawn from the uniform background are colored grey and made translucent.}
    \label{fig:synthdata-data}
\end{figure}

\begin{figure}
    \centering
    \includegraphics[width=0.9\textwidth]{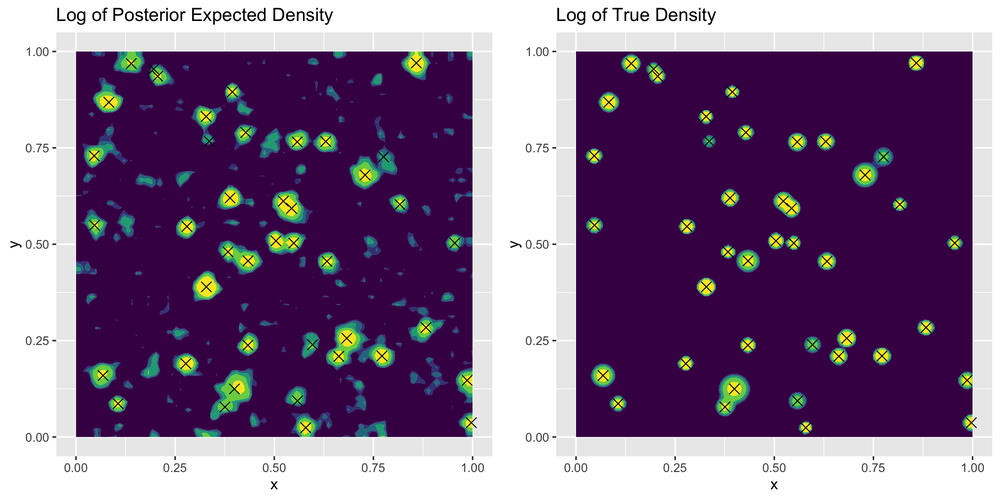}
    \caption{Comparison of $\log(\hat{f})$ and $\log(f)$, where $\hat{f}$ is the posterior expectation of $f$ under our mixture of random histograms model fitted to the data in Figure \ref{fig:synthdata-data}.}
    \label{fig:synthdata-density}
\end{figure}

Figure \ref{fig:ballet-dbscan-compare-param-sensitivity} shows how the performance of \texttt{DBSCAN} is highly sensitive to the choice of tuning parameter. It is interesting to note that the optimal parameters in this application are far from the values suggested by the heuristics proposed in \cite{schubert2017dbscan}, suggesting that in general they will be highly context dependent. We show the performance of optimally tuned \texttt{DBSCAN} in \Cref{fig:synthdata-dbscan-best}, noting that this tuning procedure required knowledge of the ground truth. The bounds of the 95\% credible ball of the \texttt{BALLET} point estimate for the synthetic data are shown in Figure \ref{fig:synthdata-ballet-bounds}. The associated \texttt{BALLET} point estimate is shown in Figure \ref{fig:synthdata-pe} of the main document. The complete results of the sensitivity and specificity of the various point estimates and bounds considered, averaged over the 100 synthetic datasets, are presented in Table \ref{tab:synthdata-tabular-results}.

\begin{figure}%
    \centering
    \subfloat[\centering DBSCAN performance vs $\MinPts$ tuning parameter]{{\includegraphics[width=0.45\textwidth]{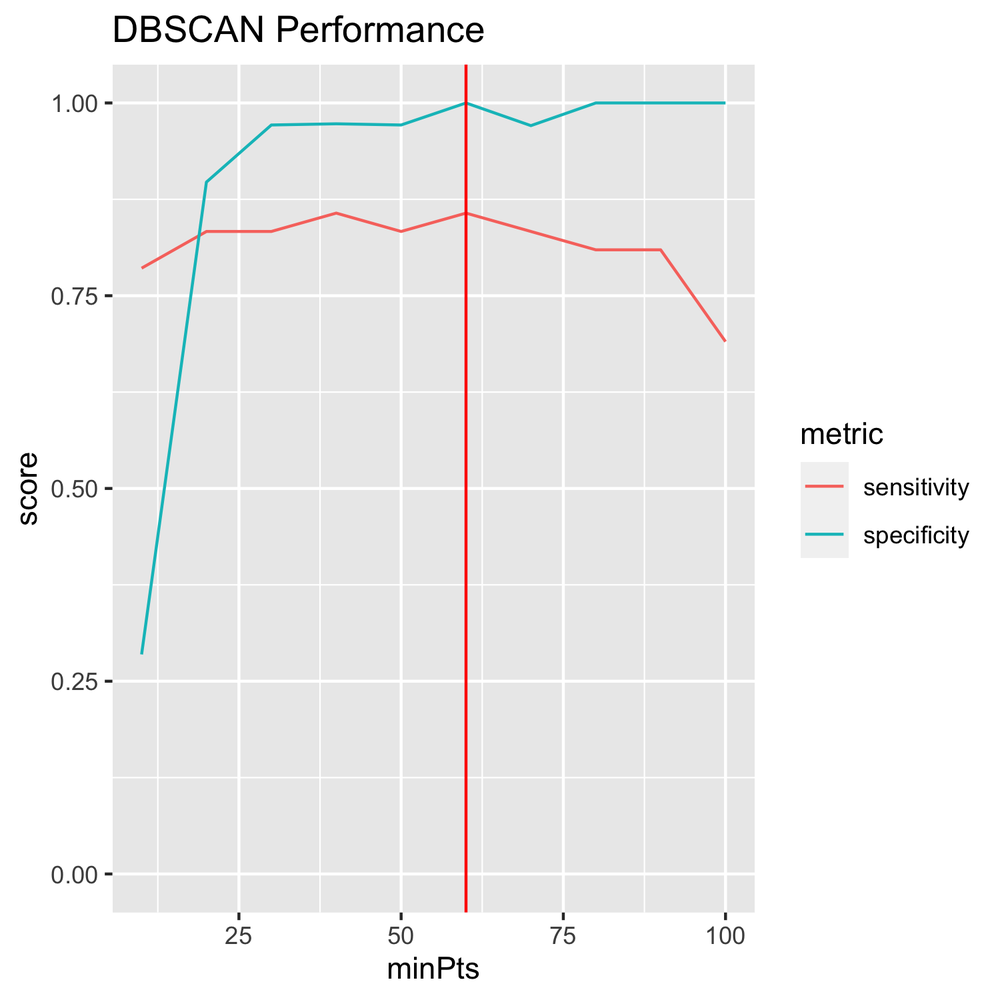} }}%
    \qquad
    \subfloat[\centering BALLET performance vs $\MinPts$ tuning parameter]{{\includegraphics[width=0.45\textwidth]{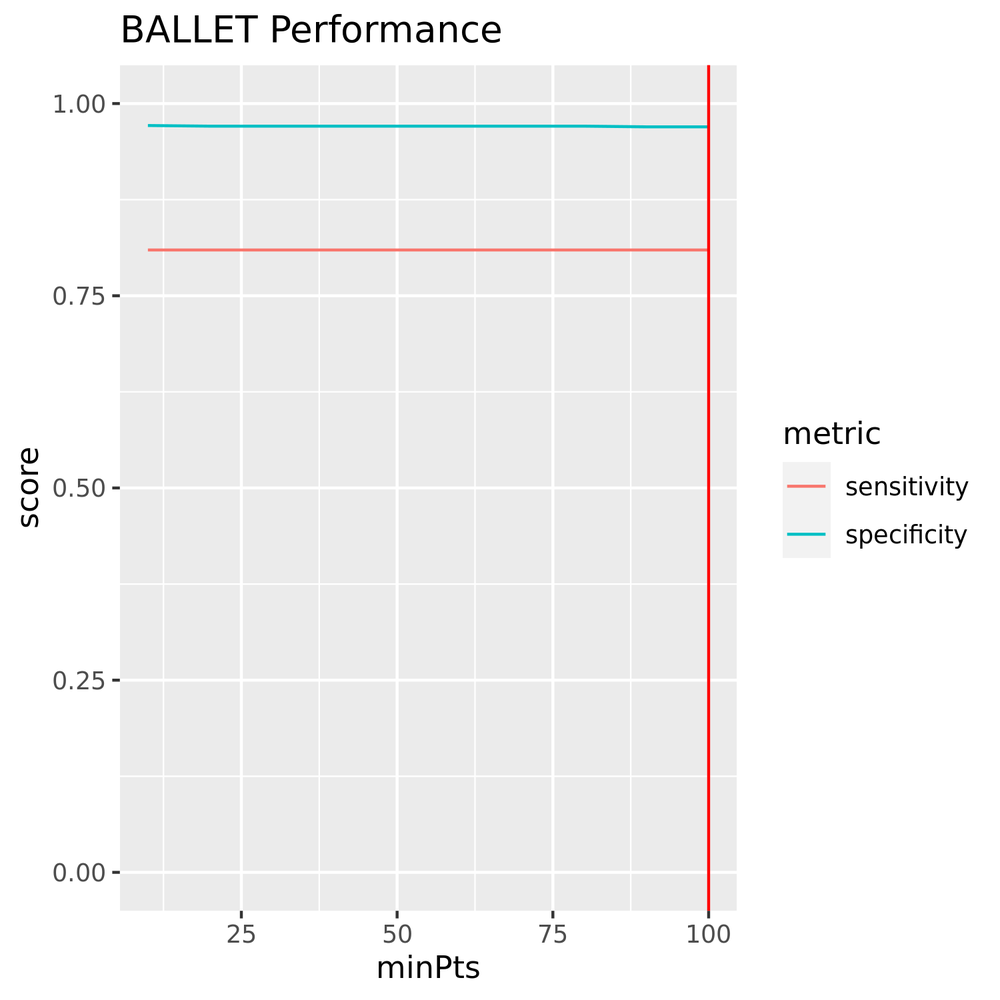} }}%
    \caption{The performance of \texttt{BALLET} and \texttt{DBSCAN} clusters as the tuning parameter $k$ (or equivalently $\MinPts$) varies. Vertical lines call attention to the value of \texttt{$k$} that exhibits the ``best" performance, as determined by the sum of the sensitivity and specificity.}
    \label{fig:ballet-dbscan-compare-param-sensitivity}
\end{figure}

\begin{figure}
    \centering
    \includegraphics[width=0.9\textwidth]{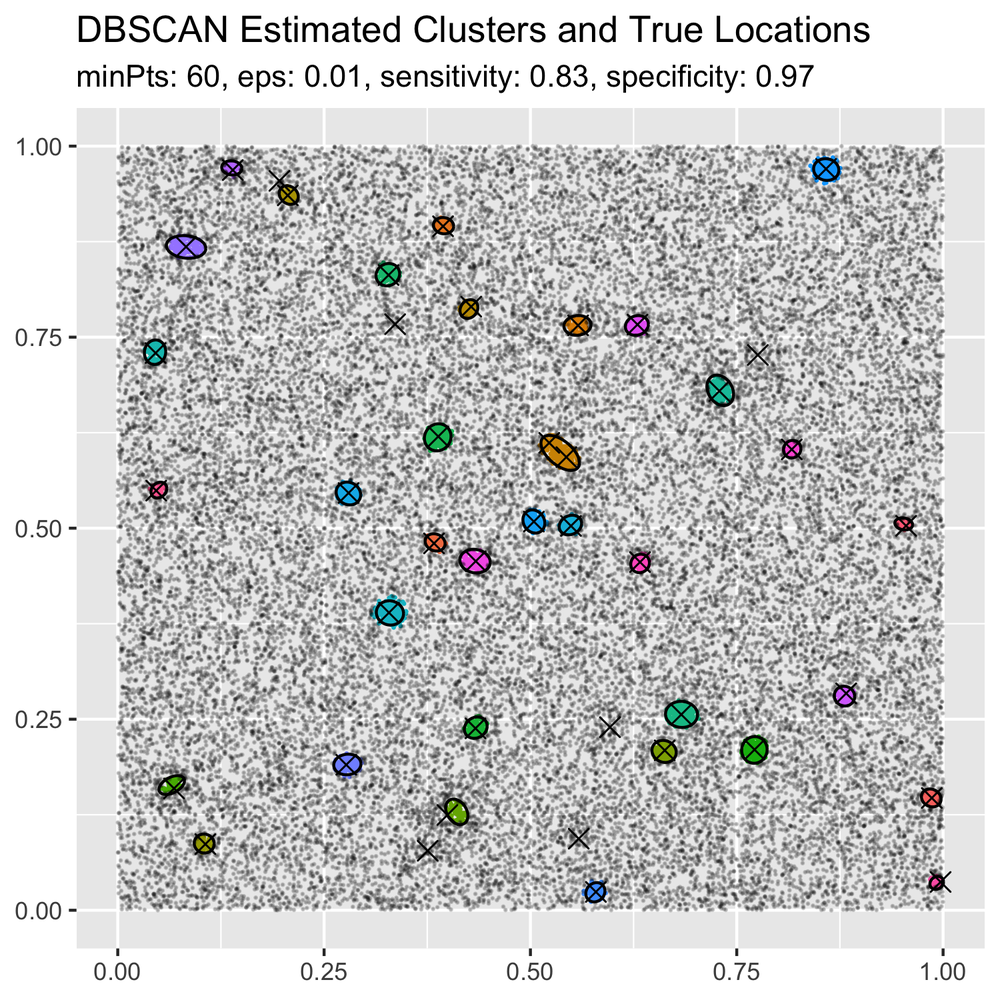}
    \caption{The result of fitting \texttt{DBSCAN} to the particular synthetic sky survey data using the optimal value of $\MinPts$ based on our simulation study.}
    \label{fig:synthdata-dbscan-best}
\end{figure}

\begin{table}[!ht]
\centering
\begin{tabular}{|l||l|l||l|l|l||l|}
\hline
& \texttt{DBSCAN} & \texttt{DBSCAN}$^1$ & \thead{\texttt{BALLET}\\ Lower} & \thead{\texttt{BALLET}\\ Est.} & \thead{\texttt{BALLET}\\ Upper} &  \thead{\texttt{BALLET}\\ Plugin}\\
\hline
Sensitivity & 0.86 & 0.79 & 0.62 & 0.78 & 0.89 & 0.78 \\
Specificity & 0.49 & 0.99 & 0.99 & 0.99 & 0.96 & 0.99  \\
Exact Match & 0.45 & 0.88 & 0.90 & 0.87 & 0.83 & 0.88 \\
\hline
\end{tabular}

\caption{Averaged results from applying \texttt{BALLET} and \texttt{DBSCAN} to 100 replicates of the synthetic sky survey data. For \texttt{BALLET}, we also provide the performance of upper and lower bounds for a 95\% credible ball centered at the point estimate. For \texttt{DBSCAN}, we provide averaged sensitivity and specificity for both our default choice of its tuning parameter and for its optimized parameter choice indicated as \texttt{DBSCAN}$^1$ (see \Cref{fig:ballet-dbscan-compare-param-sensitivity}). }
\label{tab:synthdata-tabular-results}
\end{table}

\begin{figure}
    \centering
    \includegraphics[width=0.9\textwidth]{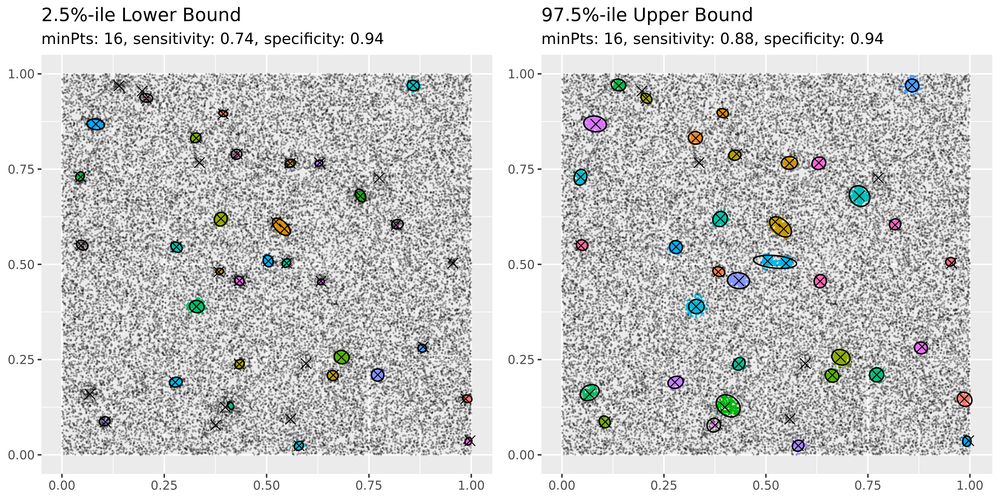}
    \caption{Upper and lower bounds for the 95\% credible ball centered at our BALLET clustering estimate for the particular synthetic dataset shown in in Figure \ref{fig:synthdata-data}.}
    \label{fig:synthdata-ballet-bounds}
\end{figure}

\section{Additional results from analysis of the sky survey data}
\label{s:app-edsgc-sky-survey}
In this section we provide additional results from the analysis of the Edinburgh-Durham Southern Galaxy Catalogue data which appeared in Section \ref{s:real-data-analysis} of the main text. In particular, we visualize the log of the posterior expectation of the data generating density in Figure \ref{fig:edsgc-density-estimate}, \texttt{DBSCAN} and \texttt{BALLET} fits based on our default value of $\MinPts = k_0 = \lceil \log_2(n) \rceil$ in \Cref{fig:edsgc-dbscan-heuristic,fig:edsgc-ballet-pe}, and an alternative \texttt{DBSCAN} fit using the optimal tuning parameters from the simulation study in Figure \ref{fig:edsgc-dbscan-tuned}. We present tabular results collecting the rate of coverage of the EDCCI and Abell catalogs, by the various point estimates and bounds we have considered, in Tables \ref{tab:edsgc-edcci-coverage-full} and \ref{tab:edsgc-abell-coverage-full}, respectively.
\begin{figure}
    \centering
    \includegraphics[width=0.9\textwidth]{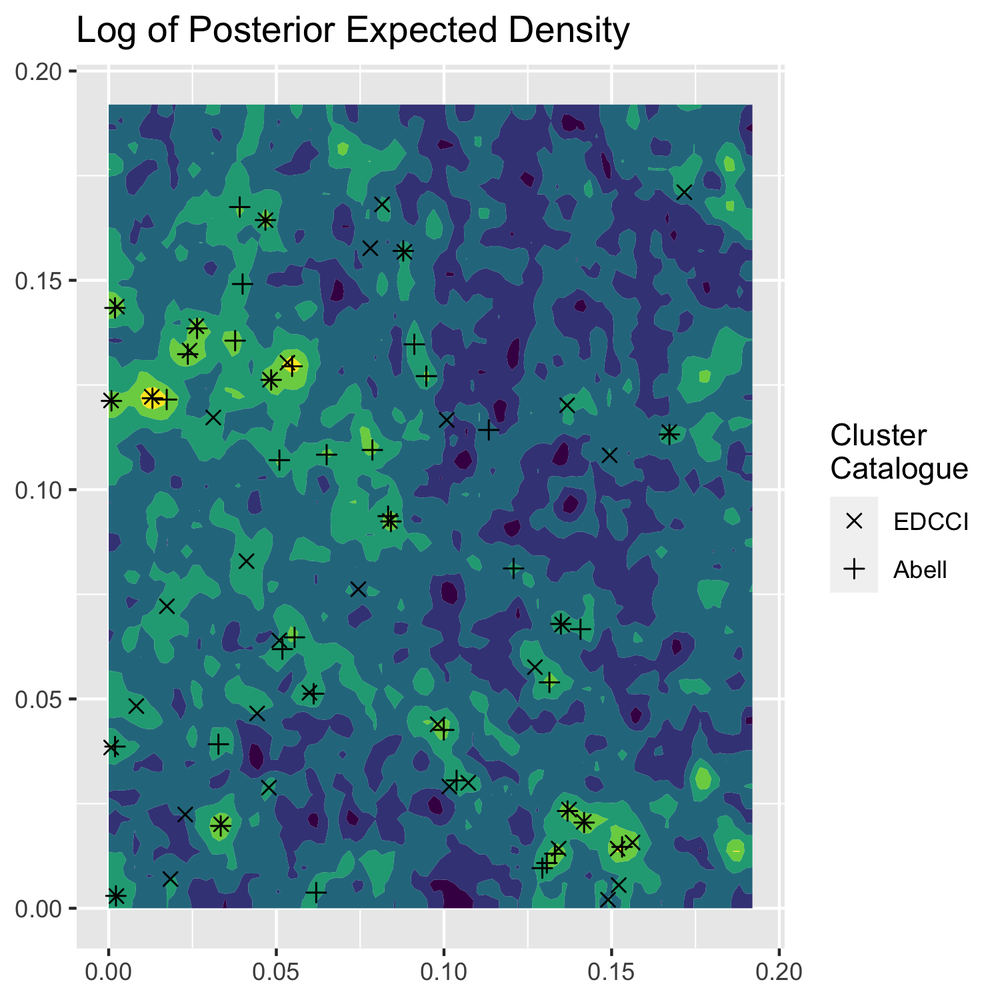}
    \caption{Log of the posterior expectation of the density for the Edinburgh-Durham Southern Galaxy Catalogue data under our mixture of random histograms model. For reference, we have superimposed galaxy clusters reported in the EDCCI and Abell cluster catalogs.}
    \label{fig:edsgc-density-estimate}
\end{figure}

\begin{figure}
	\centering
	\includegraphics[width=0.8\textwidth]{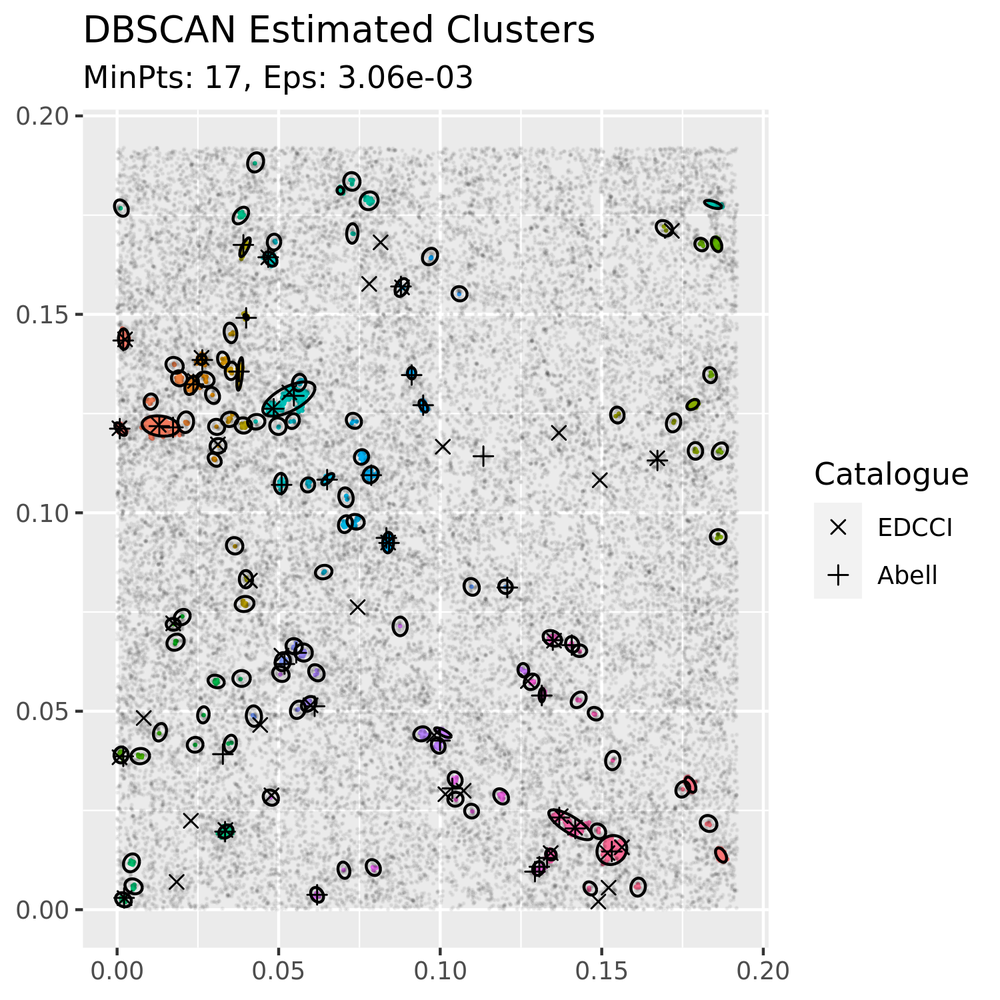}
	\caption{Result of applying \texttt{DBSCAN} to the Edinburgh-Durham Southern Galaxy Catalogue data using our default value of $\texttt{MinPts}$. Cluster centers from the two previously proposed cluster catalogs are plotted with black `+'s (Abell Catalog) and `X's (EDCCI).}
	\label{fig:edsgc-dbscan-heuristic}
\end{figure}

\begin{figure}
	\centering
	\includegraphics[width=0.8\textwidth]{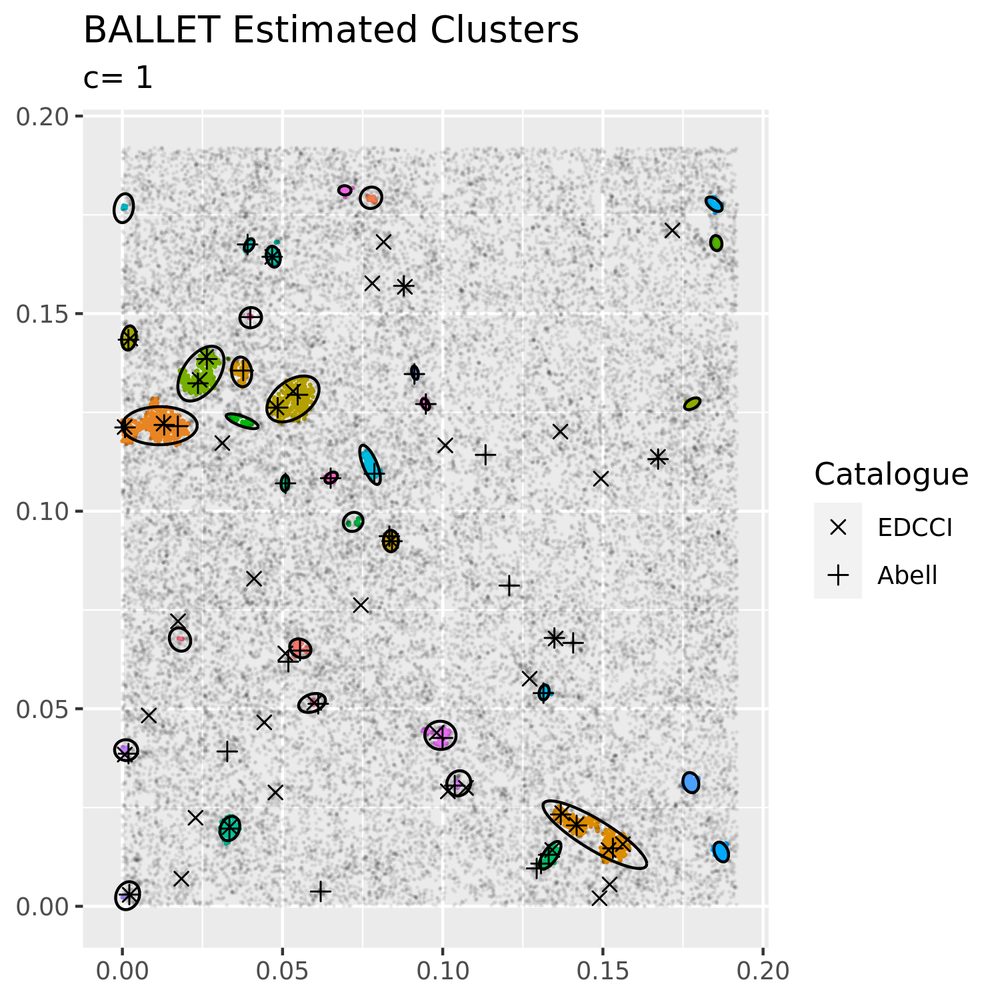}
	\caption{Results of applying \texttt{BALLET} Clustering to the Edinburgh-Durham Southern Galaxy Catalogue data, with 95\% credible bounds presented in \Cref{fig:edsgc-ballet-bounds}.}
	\label{fig:edsgc-ballet-pe}
\end{figure}

\begin{figure}
    \centering
    \includegraphics[width=0.9\textwidth]{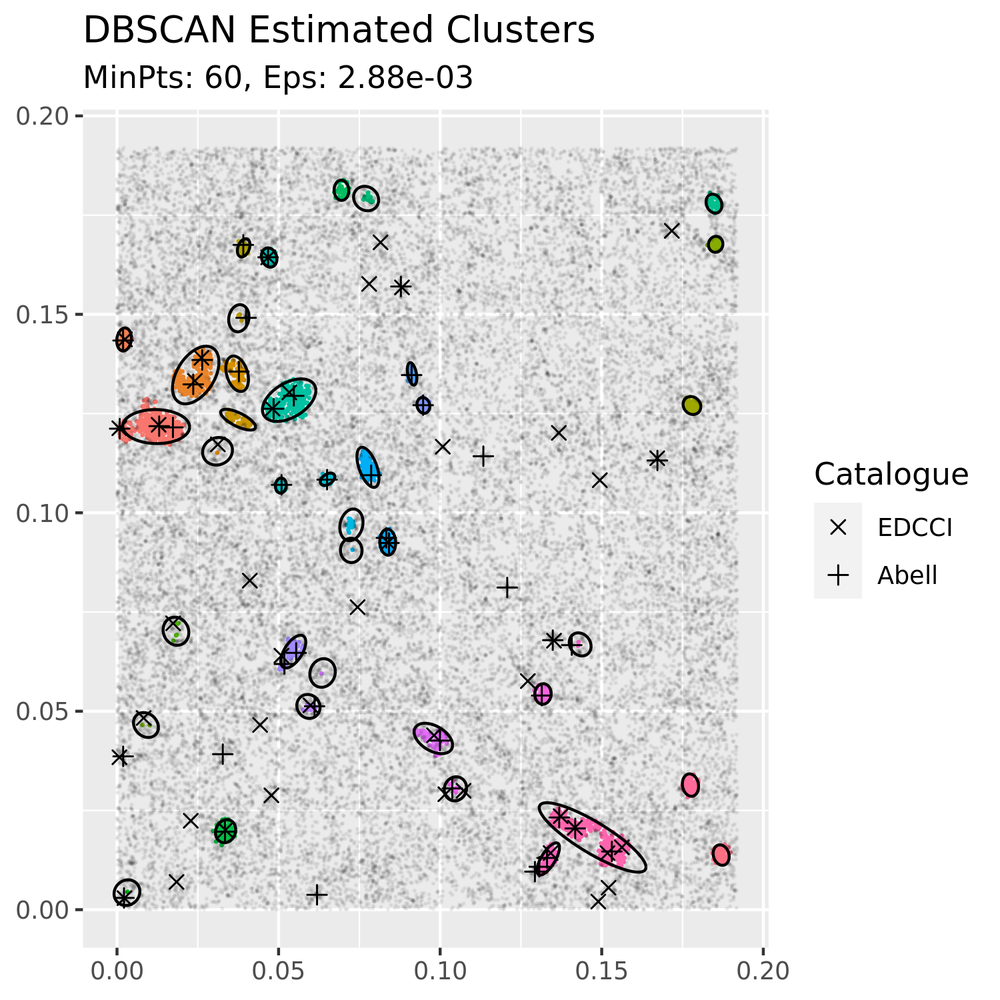}
    \caption{Result of applying \texttt{DBSCAN} to the Edinburgh-Durham Southern Galaxy Catalogue data using the tuning parameter that had optimal performance in our simulation study (\Cref{fig:ballet-dbscan-compare-param-sensitivity}).}
    \label{fig:edsgc-dbscan-tuned}
\end{figure}

\begin{table}[!ht]
\centering
\begin{tabular}{|l||l|l||l|l|l||l|}
\hline
& \texttt{DBSCAN} & \texttt{DBSCAN}$^1$ & \thead{\texttt{BALLET}\\ Lower} & \thead{\texttt{BALLET}\\ Est.} & \thead{\texttt{BALLET}\\ Upper} &\thead{\texttt{BALLET}\\ Plugin}\\
\hline
Sensitivity & 0.40& 0.37& 0.21 & 0.40 & 0.56 & 0.40 \\
Specificity & 0.15& 0.43& 0.73 & 0.40 & 0.34 & 0.42 \\
Exact Match & 0.13& 0.35& 0.67 & 0.26 & 0.29 & 0.28 \\
\hline
\end{tabular}
\caption{\texttt{DBSCAN} and \texttt{BALLET} clustering coverage of the suspected galaxy clusters listed in the Abell catalog. The column labeled \texttt{DBSCAN} reports the performance of the method with the default value of $\MinPts=16$, while \texttt{DBSCAN}$^1$ shows the performance of the method with the optimal value of $\MinPts=60$ chosen based on our simulation study.}
\label{tab:edsgc-abell-coverage-full}
\end{table}

\section{Theory details from \Cref{s:theory}}
\label{s:theory-appendix}

\subsection{Proof of \Cref{thm:consistency}}
\label{ss:proof-of-consistency}

The proof is a simple application of the metric properties of $D$. In particular,  
	\begin{align*}
		D\{\cbayes, \psi(f_0)\} &\leq E_{f \sim P_M(\cdot|\cX_n)} D\{\psitilde(f), \psi(f_0)\} + E_{f \sim P_M(\cdot|\cX_n)} D\{\psitilde(f), \cbayes\}\\
		&\leq 2 E_{f \sim P_M(\cdot|\cX_n)} D\{\psitilde(f), \psi(f_0)\},
	\end{align*}
	where the first line follows by taking expectation with respect to the posterior distribution $P_M(\cdot|\cX_n)$ after using the triangle inequality and symmetry for the metric $D$, while the second line follows by noting that the second term in the right hand side of the first line is no greater than the first term, since $\cbayes$ is given by \eqref{eq:density-based-estimator}. 
	Noting further that $D$ is bounded above by one, we obtain
	\begin{align*}
		E_{f \sim P_M(\cdot|\cX_n)} D\{\psitilde(f), \psi(f_0)\} 
		&\leq P_M\left(f : \rho(f,f_0) > K_n \epsilon_n |\cX_n\right) + \sup_{f: \rho(f,f_0)\leq K_n \epsilon_n} D\{\psitilde(f), \psi(f_0)\} \\
		&= \tau_{1}(\cX_n) + \tau_{2}(\cX_n),
	\end{align*}
	where  $\tau_{1}$ is defined in \Cref{ass:sup-norm-concentration}  and $\tau_{2}$ and constant $K_n$ are as defined in \Cref{ass:continuity-at-f0}. Since $K_n \to \infty$, these assumptions show that $\tau_{1}(\cX_n), \tau_{2}(\cX_n) \pconv 0$ as $n \to \infty$.

\subsection{Proof of \Cref{thm:IABender-is-a-metric}}
\label{ss:metric-properties}

It order to simplify the presentation of our proof we first introduce some notation. 
We note that any sub-partition $\bC = \{C_1, \ldots, C_k\} \in \ClustSpace[\cX_n]$ defines a binary ``co-clustering'' relation $\bC_R : \cX_n \times \cX_n \to \{0,1\}$ on pairs of data points, namely 
$$
\bC_R(x,y) \doteq \1{(x \notin A, y \notin A)} + \sum_{h=1}^k \1{(x \in C_h, y \in C_h)}
$$
where $A = \cup_{h=1}^k C_h$ is the set of active points in $\bC$. In other words, $\bC_R(x,y) = 1$ if $x, y \in \cX_n$ are both noise points or if they belong to a common cluster in $\bC$, and $\bC_R(x,y)=0$ otherwise. Given $\bC$, we can also obtain an indicator function of active points $\bC_A: \cX_n \to \{0,1\}$ such that $\bC_A(x)=1$ if and only if $x \in A$. In fact, knowing the binary functions $\bC_R$ and $\bC_A$ is sufficient to uniquely recover the sub-partition $\bC \in \ClustSpace[\cX_n]$. Indeed, this follows because $\bC_R$ is an equivalence relation on $\cX_n$, and the sub-partition $\bC$ can be recovered by dropping the inactive subset $\bC_A^{-1}(0)$ from the equivalence partition of $\cX_n$ induced by $\bC_R$.

We also introduce the following subscript-free notation for summation of a symmetric function $F: \cX_n \times \cX_n \to \R$ over pairs of distinct data points that lie in $S \subseteq \cX$:
$$
\sum_{x \neq y \in \cX_n \cap S} F(x,y) \doteq \sum_{\substack{1 \leq i < j \leq n \\ x_i, x_j \in S}} F(x_i, x_j) = \frac{1}{2} \sum_{\substack{i, j \in [n] \\ x_i, x_j \in S}} F(x_i, x_j) \1{(i \neq j)}.
$$

\begin{proof}[of \Cref{thm:IABender-is-a-metric}] Similar to analyses of Binder's loss, the first step in our proof is to note that $\Lia$ can be written as a sum of pairwise losses $\phi_{x, y}$ over pairs $x, y \in \cX_n$. In particular, fix any $\bC, \bC' \in \ClustSpace[\cX_n]$, and let $A = \bC_A^{-1}(1)$, $A' = \bC_A^{'-1}(1)$ and $I = \bC_A^{-1}(0)$, $I' = \bC_A^{'-1}(0)$ denote the active and inactive sets of $\bC$ and $\bC'$ , respectively.

	Taking $a = b$ and $m = m_{ia} = m_{ai}$ in \eqref{eq:ia-binder-loss}, we note
	\begin{align}
		\Lia(\bC, \bC') &= m (n-1) (|A \cap I'| + |I \cap A'|) + a \sum_{\substack{1 \leq i < j \leq n \\ x_i, x_j \in A \cap A'}} \1{\{\bC_R(x_i,x_j) \neq \bC'_R(x_i,x_j)\}}  \nonumber\\
		&= \sum_{x \neq y \in \cX_n} \phi_{x,y}(\bC,\bC') \label{eq:sum-representation}
	\end{align}
	where 
	$$
	\phi_{x, y}(\bC,\bC') = m \1{\{\bC_A(x) \neq \bC_A'(x)\}} + m \1{\{\bC_A(y) \neq \bC_A'(y)\}} + a \1{\{\bC_R(x,y) \neq \bC'_R(x,y)\}}\1{\{\bC_A(x)=\bC_A'(x) = \bC_A(y)=\bC_A'(y)\}}.
	$$
	In order to obtain \eqref{eq:sum-representation}, we have used the fact that the last term in $\phi_{x, y}(\bC, \bC')$ is zero when either one of $x$ or $y$ is outside the set $A \cap A'$, and the fact that the summation $\Sigma_{x \neq y \in \cX_n}$ over the first two terms in $\phi_{x, y}(\bC, \bC')$ is equal to $m (n-1) (|A \cap I'| + |I \cap A'|)$.
	
	Now we shall use \eqref{eq:sum-representation} to show that $D = \binom{n}{2}^{-1}\Lia$ is a metric that is bounded above by one when $a, m \leq 1$. Note that at most one out of the three indicator variables in $\phi_{x, y}$ can be non-zero for any instance, and hence $\phi_{x, y}$ is bounded above by one (in fact by $\max(a,m) \leq 1$) for each of the $\binom{n}{2}$ summation variables $x \neq y \in \cX_n$. This shows that $D$ is also bounded above by one. Further, the symmetry of $D$ in its arguments follows from the symmetry of $\phi_{x, y}$ in its arguments for every $x \neq y \in \cX_n$. 
	
	Next suppose $D(\bC, \bC') = 0$. Since the functions $\phi_{x, y}$ are non-negative, this shows that $\phi_{x, y}(\bC, \bC') = 0$ for each $x \neq y \in \cX_n$. Since $2m \geq a > 0$, the functions $\bC_A$ and $\bC'_A$ are equal (or equivalently that $A=A'$), and further that $\bC_R(x,y) = \bC'_R(x,y)$ either when $x, y \in A = A'$ or $x, y \in I = I'$. The latter condition is sufficient to show that the relations $\bC_R$ and $\bC'_R$ are equal since $\bC_R(x,y) = 0 = \bC'_R(x,y)$ when $x \in A, y \in I$ or $x \in I, y \in A$. Since the binary functions  $\bC_A$ and $\bC_R$ determine the sub-partition $\bC$, we have $\bC = \bC'$. 
	
	Finally, to demonstrate that $D$ satisfies the triangle inequality, it suffices to show that for each $x \neq y \in \cX_n$, we have the triangle inequality $\phi_{x, y}(\bC, \bC'') \leq \phi_{x, y}(\bC, \bC') + \phi_{x, y}(\bC', \bC'')$ for any sub-partitions $\bC, \bC', \bC'' \in \ClustSpace[\cX_n]$. Indeed when either $\bC_A(x) \neq \bC''_A(x)$ or $\bC_A(y) \neq \bC''_A(y)$, the triangle inequality for $\phi_{x, y}$ follows from the inequality:
	\begin{align*}
		\1{\{\bC_A(z) \neq \bC_A''(z)\}} &\leq \1{\{\bC_A(z) \neq \bC_A'(z)\}} + \1{\{\bC'_A(z) \neq \bC_A''(z)\}} \qquad z \in \{x, y\}.
	\end{align*}
	Otherwise, let us assume that the previous condition does not hold. Let us further suppose that $\phi_{x, y}(\bC, \bC'') > 0$ or else there is nothing to show. This means that we are under the case $\phi_{x, y}(\bC, \bC'') = a$,  $\bC_A(x) =  \bC''_A(x) = \bC_A(y) =  \bC''_A(y)$, and $\bC_R(x,y) \neq \bC''_R(x,y)$. If $\bC'_A(x) \neq \bC_A(x) = \bC''_A(x)$ (or analogously $\bC'_A(y) \neq \bC_A(y) = \bC''_A(y)$) then the triangle inequality is satisfied as $\phi_{x, y}(\bC, \bC') + \phi_{x, y}(\bC', \bC'') \geq m \1{\{\bC_A(x) \neq \bC_A'(x)\}} + m \1{\{\bC'_A(x) \neq \bC_A''(x)\}} = 2m \geq a = \phi_{x, y}(\bC, \bC'')$. Otherwise, the only remaining case is that $\bC_A(x) = \bC'_A(x) = \bC''_A(x) = \bC_A(y) = \bC'_A(y) = \bC''_A(x)$. Then the triangle inequality is satisfied since
	\begin{align*}
		\phi_{x, y}(\bC, \bC'') = a \1{\{\bC_R(x,y) \neq \bC''_R(x,y)\}} &\leq a \1{\{\bC_R(x,y) \neq \bC'_R(x,y)\}} + a \1{\{\bC'_R(x,y) \neq \bC''_R(x,y)\}} \\
		&= \phi_{x, y}(\bC, \bC') + \phi_{x, y}(\bC', \bC'').
	\end{align*}
	Hence, we have verified the triangle inequality for $\phi_{x, y}$, and hence also for $D$. Combined with the non-negativity of $D$, we have shown that $D$ is a metric.
\end{proof}

\subsection{Proof of \Cref{lem:finite-sample-bound-for-assumption2}}
\label{ss:prove-level-set-clustering}

Letting $\cX = \R^d$, we begin with the necessary assumptions on the unknown data density $f_0: \cX \to \R$ and the threshold level $\lambda > 0$. 
Let  $S_\lambda = \{x \in \R^d : f_0(x) \geq \lambda\}$ denote the level set of the unknown data density $f_0$ at threshold $\lambda \in (0,\infty)$. We make the following assumptions.

\begin{assumption} (Continuity with vanishing tails)
\label{ass:uniform-continuity}
The density $f_0: \R^d \to [0,\infty)$ is continuous and satisfies $\lim_{\|x\| \to \infty} f_0(x) = 0$.
\end{assumption}

\begin{lemma} 
	\label{lem:unif-continuity}
	If \Cref{ass:uniform-continuity} holds then $f_0$ is uniformly continuous.
\end{lemma}
\begin{proof}
	Fix any $\epsilon > 0$. Then since $f_0$ has vanishing tails, there is a $K > 0$ such that $\sup_{x \in \R^d \setminus ([-K, K]^d)} f_0(x) \leq \epsilon/2$, and since $f_0$ is continuous on the compact set $H = [-K-1, K+1]^d$, there is a $\delta \in (0,1)$ such that $|f_0(x)-f_0(y)| \leq \epsilon$ whenever $\|x-y\| \leq \delta$ and $x, y \in H$. Finally if $x, y \in \R^d$ are such that $\|x-y\| \leq 1$ and $\{x, y\} \cap (\R^d \setminus H) \neq \emptyset$ then $x, y \in \R^d \setminus [-K,K]^d$. Thus $|f_0(x)-f_0(y)| \leq f_0(x) + f_0(y) \leq \epsilon/2 + \epsilon/2 = \epsilon$. Hence we have shown that there is a $\delta \in (0,1)$ such that $|f_0(x)-f_0(y)| \leq \epsilon$ whenever $\|x-y\| \leq \delta$ and $x, y \in \R^d$. Since $\epsilon > 0$ is arbitrary, $f_0$ is uniformly continuous.
\end{proof}

\begin{assumption} (Fast mass decay around level $\lambda$) There are constants $C, \beps > 0$ such that 
$\int_{\{x \in \R^d : |f_0(x) - \lambda| \leq \epsilon \}} f_0(x) dx \leq C \epsilon$ for all $\epsilon \in (0, \beps)$.  
    \label{ass:mass-decay}
\end{assumption}

\Cref{ass:mass-decay} is adapted from \cite{rinaldo2010generalized}, and intuitively prevents the density from being too flat around the level $\lambda$. In particular, if $f_0$ satisfies $\|\nabla f_0(x)\| > 0$ for Lebesgue-almost-every $x$, then Lemma 4 in \cite{rinaldo2010generalized} shows that \Cref{ass:mass-decay} will hold for Lebesgue-almost-every $\lambda \in (0, \|f_0\|_\infty)$. Additionally, if $f_0$ is smooth and has a compact support, the authors show that the set of $\lambda \in (0, \|f_0\|_\infty)$ for which \Cref{ass:mass-decay} does not hold is finite.

\begin{assumption} (Stable connected components at level $\lambda$) For any $\lambda_l < \lambda_h \in [\lambda - \beps, \lambda + \beps]$, and $x, y \in S_{\lambda_h}$: 
    \begin{enumerate}
        \item If $x, y$ are disconnected in  $S_{\lambda_h}$, then $x, y$ are also disconnected in  $S_{\lambda_l}$. 
        \item If $x, y$ are connected in  $S_{\lambda_l}$, then $x, y$ are also connected in $S_{\lambda_h}$.
    \end{enumerate}
    \label{ass:robust-cluster}
\end{assumption}

Informally, \Cref{ass:robust-cluster} states that the connected components of the level-set $S_{\lambda'}$ do not  merge or split as $\lambda'$ varies between $(\lambda-\beps, \lambda + \beps)$. When combined with \Cref{ass:uniform-continuity}, this assumption ensures that the level set clusters vary continuously with respect to the level $\lambda$. Various versions of such assumptions have previously appeared in the literature like Assumption C2 in  \cite{rinaldo2010generalized} and Definition 2.1 in \cite{sriperumbudur2012consistency}.

We now prove some intermediate theory on level set estimation that will be useful in the proof of \Cref{lem:finite-sample-bound-for-assumption2}. 
Given data points $\cX_n = \{x_1, \ldots, x_n\}$ suppose we have a density estimator $f$ that approximates $f_0$. 
For a suitably small choice of $\delta > 0$, we estimate the level set $S_\lambda$ by the $\delta$ diameter tube around the active data points, namely:
$$
T_\delta(A_{f, \lambda}) = \bigcup_{\substack{ x \in A_{f, \lambda}}}  B(x,\delta/2),
$$ 
where $A_{f, \lambda} = \{x \in \cX_n : f(x) \geq \lambda \}$ is the set of active data points and $B(x,\delta/2)$ is the open ball of radius $\delta/2$ around $x$. %
To emphasize that $T_\delta(A_{f, \lambda})$ is an  estimator for $S_\lambda$, we denote it as $\hat{S}_{\delta, \lambda}(f) \doteq T_\delta(A_{f, \lambda})$ in the sequel.  

The following lemma shows that the level set estimator $\hat{S}_{\delta, \lambda}(f)$ approximates the level sets of the original density $S_\lambda$ as long as the quantities $\|f_0 - f\|_\infty$ and $\delta > 0$ are suitably small. This result extends Lemma~3.2 in \cite{sriperumbudur2012consistency} to the case when $f$ is an arbitrary approximation to $f_0$. Our proof hinges on using \Cref{cor:non-empty} below rather than a specific kernel density estimator as in \cite{sriperumbudur2012consistency}. 

\begin{lemma}
\label{lem:level-set-containment}
Suppose $\cX = \R^d$ and $f_0: \cX \to [0, \infty)$ is uniformly continuous. Then 
\begin{equation}
	\label{eq:cont-coeff}
	\heps{\eta} \doteq \max\{h \geq 0 : \sup_{x, y \in \cX,\|x-y\| \leq h} |f(x) - f(y)| \leq \eta\}
\end{equation}
is a positive number for each $\eta > 0$. %
Given observations $x_1, \ldots, x_n$ drawn independently from $f_0$ with $n \geq 16$, with probability at least $1-1/n$ we have
$$
S_{(\lambda + \|f_0 - f\|_\infty + \eta)} \subseteq \hat{S}_{\delta, \lambda}(f) \subseteq S_{(\lambda - \|f_0 - f\|_\infty - \eta)},
$$
uniformly over all functions $f: \R^d \to \R$, and constants $\eta, \lambda > 0$ such that $\delta \in [\rnld, 2 \heps{\eta}]$, where $\rnld\doteq\rnldval$ and $v_d$ is the volume of the unit Euclidean ball in $\R^d$.
\end{lemma}

Before we prove the above lemma, we will establish \Cref{cor:non-empty} which provides a lower-bound on the parameter $\delta$ to ensure that the  $\delta$-ball centered around any point in the level set $S_\lambda$ will contain at least one observed sample. This is a corollary of the uniform law of large numbers result from \cite{boucheron2005theory}. We use the following version:

\begin{lemma}\cite[Theorem~15]{chaudhuri2010rates} Let $\cG$ be a class of functions from $\cX$ to $\{0,1\}$ with VC dimension $d < \infty$, and let $P$ be a probability distribution on $\cX$. Let $E$ denote the expectation with respect to $P$. Suppose $n$ points are drawn independently from $P$, and let $E_n$ denote expectation with respect to this sample. Then for any $\delta > 0$, 
$$
-\min(\beta_n^2 + \beta_n \sqrt{E g}, \beta_n \sqrt{E_n g}) \leq E g - E_n g \leq \min(\beta_n^2 + \beta_n \sqrt{E_n g}, \beta_n \sqrt{E g})
$$
holds for all $g \in \cG$ with probability at least $1-\delta$, where $\beta_n = \sqrt{(4/n)\{d \ln 2n + \ln(8/\delta)\}}$.
\label{lem:unif-convergence}
\end{lemma}

\begin{corollary}  %
Suppose $\cX_n = \{x_1, \ldots, x_n\}$ are drawn independently from $f_0$ and $n \geq 16$. Then with probability at least $1-1/n$, we have $
\cX_n \cap B \neq \emptyset$ for each  Euclidean ball $B \subseteq \R^d$ such that $\int_{B} f_0(x) dx \geq  \frac{16 d \ln n}{n}$.
\label{cor:non-empty}
\end{corollary}
\begin{proof} Let $\cG = \{\1{B(x, r)} | x \in \R^d \text{ and } r > 0\}$ be the class of indicator functions of all the Euclidean balls, and note that the VC dimension of spheres in $\R^d$ is $d+1$ (e.g.~\cite{wainwright2019}). Lemma \ref{lem:unif-convergence} then states that with probability at least $1-1/n$, 
$$
P(B) - P_n(B) \leq \beta_n \sqrt{P(B)}
$$
for any Euclidean ball $B \subseteq \R^d$, where $P_n(B) = \frac{1}{n} \sum_{i=1}^n \1{(x_i \in B)}$ is the empirical distribution function and $\beta_n = \sqrt{(4/n)\{(d+1) \ln (2n) + \ln (8n)\}}$. In particular, as long as this event holds and $P(B) > \beta_n^2$, one has $P_n(B) > 0$ and hence $\cX_n \cap B \neq \emptyset$. The proof is completed by noting that $\beta_n^2 \leq \frac{16d \ln n}{n}$ whenever $n \geq 16$.
\end{proof}

\begin{proof}[of \Cref{lem:level-set-containment}] With probability at least $1-1/n$ the event in \Cref{cor:non-empty} holds; we will henceforth condition on this event. Next, let $v_d = \frac{\pi^{d/2}}{\Gamma(d/2+1)}$ be the volume of the unit Euclidean sphere in $d$ dimensions and note that  $\lambda v_d (\delta/2)^d \geq \frac{16 d \ln n}{n}$ whenever $\delta \geq \rnld \doteq \rnldval$. This shows that for any $x \in \cX$ 
\begin{equation}
	\label{eq:non-empty-balls}
	\cX_n \cap B(x,\delta/2) \neq \emptyset \qquad \text{ whenever }  \inf_{y \in B(x,\delta/2)} f_0(y) \geq \lambda,
\end{equation}
and further since $\delta/2 \leq \heps{\eta}$ that
\begin{equation}
\label{eq:unif-cont}
\sup_{y \in B(x,\delta/2)} |f_0(y)-f_0(x)| \leq \eta.
\end{equation}

We are now ready to prove our main statement in \Cref{lem:level-set-containment}. We first show the inclusion $\hat{S}_{\delta, \lambda}(f) \subseteq S_{(\lambda - \|f_0 - f\|_\infty - \eta)}$. Indeed, for any $x \in \hat{S}_{\delta, \lambda}(f)$ there is a $y \in \cX_n$ such that $x \in B(y, \delta/2)$ and $f(y) \geq \lambda$. The inequalities
 $$
 f_0(x) \geq f_0(y) - \eta \geq f(y) - \abs{f_0(y) - f(y)} - \eta \geq \lambda  - \abs{f_0(y) - f(y)} - \eta
 $$
 then show $x \in S_{(\lambda - \|f_0 - f\|_\infty - \eta)}$. Since $x \in \hat{S}_{\delta, \lambda}(f)$ was arbitrary our inclusion follows.
 
Next, we show the inclusion $S_{(\lambda + \|f_0 - f\|_\infty + \eta)} \subseteq \hat{S}_{\delta, \lambda}(f)$.  Pick an $x \in S_{(\lambda + \|f_0 - f\|_\infty + \eta)}$ and note by \eqref{eq:unif-cont} that $\inf_{y \in B(x,\delta/2)} f_0(y) \geq f_0(x) - \eta \geq \lambda + \|f_0 - f\|_\infty$. Thus \eqref{eq:non-empty-balls} shows the existence of  some $z \in B(x,\delta/2) \cap \cX_n$. Further $f(z) \geq f_0(z) - \abs{f_0(z)-f(z)} \geq f_0(x)  - \eta - \|f-f_0\|_\infty \geq \lambda$ since $f_0(z) \geq f_0(x) - \eta$ and $x \in S_{(\lambda + \|f_0 - f\|_\infty + \eta)}$. Thus we have shown that $x \in \hat{S}_{\delta, \lambda}(f)$. Since $x \in S_{(\lambda + \|f_0 - f\|_\infty + \eta)}$ was arbitrary our inclusion follows.
\end{proof}

We now discuss consequences of \Cref{lem:level-set-containment} for level set clustering of data $\cX_n$. As discussed in \Cref{ss:surrogate-clustering-function}, we use the surrogate clustering $\psiclustn(f)$ of data $\cX_n$ defined in \eqref{eq:levelset-clust-function}, which computes the graph-theoretic connected components \citep{sanjoy2008algorithms} of the $\delta$-neighborhood graph $G_{\delta}(A_{f,\lambda})$ having vertices $A_{f,\lambda} = \left\{ x \in \cX_n \mid f(x) \geq \lambda \right\}$ and edges $E = \{(x,y) \in A_{f,\lambda} \times A_{f,\lambda} \mid \|x-y\| < \delta \}$. The following known result (e.g. Lemma 1 in \cite{wang2019dbscan}) connects the surrogate clustering  $\psiclustn(f)$ to the  level-set estimator  $\hat{S}_{\delta, \lambda}(f)$ defined in the last section. We provide an independent proof here for completeness.

\begin{lemma} The surrogate clustering $\psiclustn(f) \in \ClustSpace[\cX_n]$ coincides with the partition of $A_{f,\lambda} 
	= \left\{ x \in \cX_n \mid f(x) \geq \lambda \right\}$ induced by the \emph{topological connected components} of the level set estimator $\hat{S}_{\delta, \lambda}(f)$.
	\label{lem:cciscc}
\end{lemma}
\begin{proof} For any two distinct choice $x, y \in A_{f,\lambda}$ we will show that $x$ and $y$ lie in the same connected component of the graph $G_{\delta}(A_{f,\lambda})$ if and only if they are path connected in the set $\hat{S}_{\delta, \lambda}(f)$. 
	
	Indeed, suppose that $x,y$ are in the same connected component of  $G_{\delta}(A_{f,\lambda})$. Then for some $2 \leq m \leq n$ there are points $\{x_i\}_{i=1}^m \subseteq   A_{f,\lambda}$ with $x_1 = x$, $x_m = y$ and $\|x_i - x_{i+1}\| <  \delta$ for $i = 1, \ldots, m-1$. These conditions ensure that the interval $[x_i, x_{i+1}] \doteq \{t x_i + (1-t) x_{i+1} : t \in [0,1]\}$ is entirely contained within $\hat{S}_{\delta, \lambda}(f)$. Thus there is a continuous path from $x$ to $y$ that entirely lies within $\hat{S}_{\delta, \lambda}(f)$, which ensures that $x, y$ are path connected in $\hat{S}_{\delta, \lambda}(f)$. 
	
	Conversely, suppose that $x, y \in A_{f,\lambda}$ are path connected in $\hat{S}_{\delta, \lambda}(f)$. Thus there is a continuous path $\varphi: [0,1] \to \hat{S}_{\delta, \lambda}(f)$ such that $\varphi(0)=x$ and $\varphi(1)=y$.  Based on $\varphi$, we can define two mappings $T: A_{f,\lambda} \to [0,1]$ and $F: [0,1] \to A_{f,\lambda}$ given by
    $$
    T(z) = \sup\{t \in [0,1] : \varphi(t) \in B(z,\delta/2)\} \quad \text{ and } \quad F(t) \in \argmin_{z \in A_{f,\lambda}} \|z - \varphi(t)\|.
    $$
    We must have $\varphi(t) \in B(F(t), \delta/2)$ for each $t \in [0,1]$ since the image of the path $\varphi$ lies entirely in $\hat{S}_{\delta, \lambda}(f)$. Further, for each $z \in A_{f,\lambda}$ such that $T(z) \in [0,1)$, it must be the case that $\|\varphi(T(z)) -z \| = \delta/2$ due to the continuity of $\varphi$.
    
    Starting with $t_0=0$ and $x_0 = F(t_0) = x$, recursively define $t_{i}=T(x_{i-1}) \in [0,1]$ and $x_i=F(t_i) \in A_{f,\lambda}$ for each $i \geq 1$. %
    By the definition of $T$, we note that $t_i = T(x_{i-1}) \geq t_{i-1}$ since $\varphi(t_{i-1}) \in B(x_{i-1}, \delta/2)$ holds given that $x_{i-1} = F(t_{i-1})$ for each $i \geq 1$. In fact, $\|x_i - x_{i-1}\| < \delta$ since $\|\varphi(t_i) - x_{i-1}\| \leq \frac{\delta}{2}$ follows by using the continuity of $\varphi$ and $t_{i}=T(x_{i-1})$, while $\|\varphi(t_i) - x_i\| < \delta/2$ follows since $x_i=F(t_i)$. Thus we can show that $x_0, x_1, \ldots, $ is an infinite path in $G_{\delta}(A_{f,\lambda})$ starting from $x_0=x \in A_{f,\lambda}$.
    
    Next, we claim that the path $x_0, x_1, \ldots, x_m$ in $G_{\delta}(A_{f,\lambda})$ will terminate at $x_m = F(t_m) = y$, where $m$ is smallest integer such that $t_{m} = 1$. Thus the proof will be complete once we show that such an $m \in \nat$ will exist. Whenever $t_{i-1} < 1$, we can observe that $t_{i-1} \neq t_i$ since $\varphi(t_i) \notin B(x_{i-1}, \delta/2)$ but $\varphi(t_{i-1}) \in B(x_{i-1}, \delta/2)$. Further as long as $t_{i-1} < 1$, we must also have $x_i \notin \{x_{0}, \ldots, x_{i-1}\}$ because $\varphi(t_i) \in B(x_i, \delta/2)$ but $\varphi(t_i) \notin \cup_{j=0}^{i-1} B(x_j, \delta/2)$ since $t_i > \max(t_0, \ldots, t_{i-1})$. Hence we have shown that for each $i \geq 1$, the points $x_0, \ldots, x_i \in A_{f,\lambda}$ will be distinct as long as $t_{i-1} < 1$. Since $A_{f,\lambda}$ is a finite set, there must be $m \in \nat$ such that $t_{m}=1$ and $x_{m} = F(t_{m}) = \varphi(t_m)=y$. Thus $x, y$ are connected by a path in $G_{\delta}(A_{f,\lambda})$.

\end{proof}  

When \Cref{lem:level-set-containment} holds and \Cref{ass:robust-cluster} is satisfied, the topological connected components of $\hat{S}_{\delta, \lambda}(f)$ will be close to those of the level set $S_\lambda$ if $\|f-f_0\|_\infty$ and $\delta$ are suitably small. To formally define this relationship we start with the following definition. 

\begin{definition}
	\label{def:co-clustring}
 Consider the binary \emph{co-clustering} relations $T, \hT: \cX \times \cX \to \{0,1\}$ defined as follows. For any $x, y \in \cX$, we define $T(x,y) = 1$  if $x$ and $y$ either both fall outside the level set $S_\lambda$ or if they lie in the same topological connected component of $S_\lambda$, otherwise we let $T(x,y) = 0$. The estimated quantity $\hT(x,y)$ is defined similarly as above, but with $S_\lambda$ replaced by $\hat{S}_{\delta, \lambda}(f)$. 
\end{definition}

\begin{lemma} 
Suppose that \Cref{ass:robust-cluster} is satisfied and the conclusion of  \Cref{lem:level-set-containment} holds with $\epsilon \doteq \|f-f_0\|_\infty + \eta \leq \beps$. Then 
whenever $T(x,y) \neq \hT(x,y)$ for some $x, y \in \cX$, it must follow that $\{x,y\} \cap S_{(\lambda - \epsilon)} \setminus S_{(\lambda + \epsilon)} \neq \emptyset$.
\label{cor:error-only-in-ciritcal-band}
\end{lemma}
\begin{proof}
	Fix any pair $x, y \in \cX$. It suffices to show that $T(x,y) = \hT(x,y)$ whenever $\{x, y\} \cap S_{(\lambda - \epsilon)} \setminus S_{(\lambda + \epsilon)} = \emptyset$. We will consider the following cases:
    \begin{enumerate}
    \item Case $x, y \in S_{(\lambda + \epsilon)}$. \Cref{ass:robust-cluster}  states that the  topological connectivity between $x, y$ as points in $S_{(\lambda')}$ remains unchanged as long as $\lambda' \in [\lambda - \beps, \lambda + \beps]$. Further  \Cref{lem:level-set-containment} shows that 
    \begin{equation}
     S_{(\lambda + \epsilon)} \subseteq \hat{S}_{\delta, \lambda}(f) \subseteq S_{(\lambda - \epsilon)}.
    \label{eq:inclusion}     
    \end{equation}
    Thus if $T(x,y) = 1$, points $x, y$ will be connected in $S_{(\lambda + \epsilon)}$ and hence also in $\hat{S}_{\delta, \lambda}(f)$, and thus we must have  $\hT(x,y) = 1$. Conversely, if $T(x,y) = 0$, then $x, y$ are disconnected in $S_{(\lambda - \epsilon)}$ and hence also in $\hat{S}_{\delta, \lambda}(f)$, giving $\hT(x,y) = 0$. 

    \item Case $x, y \notin S_{(\lambda - \epsilon)}$. Then $T(x,y) = 1$ since $x, y \notin S_\lambda$. But by \cref{eq:inclusion}, $x, y \notin \hat{S}_{\delta, \lambda}(f)$ and thus $\hT(x,y) = 1$.

    \item Case $x \in S_{(\lambda + \epsilon)}$ and $y \notin S_{(\lambda - \epsilon)}$ (or vice-versa). Then $T(x,y) = 0$ since $x \in S_\lambda$ but $y \notin S_\lambda$. \Cref{eq:inclusion} shows that $x \in \hat{S}_{\delta, \lambda}(f)$ and $y \notin \hat{S}_{\delta, \lambda}(f)$, and thus $\hT(x,y) = 0$.
    \end{enumerate}

	In any case, we have shown that $T(x,y) = \hT(x,y)$ if the condition $\{x,y\} \cap S_{(\lambda - \epsilon)} \setminus S_{(\lambda + \epsilon)} \neq \emptyset$ does not hold.
\end{proof}

If \Cref{ass:mass-decay} holds in addition to the result in \Cref{cor:error-only-in-ciritcal-band}, then one immediately notes that for samples $X, Y$ drawn independently at random from $f_0$ we have 
\begin{align*}
	P_{f_0}\{T(X,Y) \neq \hT(X,Y)\} &\leq P_{f_0}[\{X,Y\} \cap S_{(\lambda - \epsilon)} \setminus S_{(\lambda + \epsilon)} \neq \emptyset]\\
							&\leq 2 P_{f_0}\{X \in S_{(\lambda - \epsilon)} \setminus S_{(\lambda + \epsilon)}\} = 2 \int_{\{x: |f_0(x) - \lambda| \leq \epsilon\}} f_0(x) dx \leq 2C\epsilon.
\end{align*}
where $P_{f_0}$ denotes the probability under independent draws $X,Y$ from $f_0$.
This suggests that if $\|f-f_0\|_\infty$ and $\delta > 0$ are suitably small, so that $\epsilon$ can be chosen to be small, then  for any fixed pairs of indices $1 \leq i <  j \leq n$, the data points $x_i,x_j$ will, with  probability at least $1 - C\epsilon$, be identically co-clustered by the surrogate function $\psiclustn$ and the level-set function $\psi_{\lambda}$, that is, points $x_i,x_j$ will either be in the same cluster in both $\psiclustn$ and  $\psi_{\lambda}$, or they will be in different clusters of both $\psiclustn$ and $\psi_{\lambda}$. The following theorem builds on this intuition to bound $D\{\psiclustn(f), \psi_{\lambda}(f_0)\}$ where $D = \binom{n}{2}^{-1}\Lia$ is the loss from \Cref{thm:IABender-is-a-metric}.

\begin{theorem} Let $f_0$ and $\lambda > 0 $ satisfy \Cref{ass:uniform-continuity,ass:mass-decay,ass:robust-cluster}, and let $\cX_n = \{x_1, \ldots, x_n\}$ be independent draws from $f_0$. Then, whenever $n \geq 16$, with probability at least $1-\frac{n+1}{n^2}$ 
\begin{equation}
	\label{eq:sup-statement}
	 \sup_{\delta \in [\rnld, 2\heps{\epsilon}]} \sup_{f:\|f-f_0\|_\infty \leq \epsilon} D\{\psiclust[\delta](f), \psi_{\lambda}(f_0)\} \leq 8 \bigg(C \epsilon + \sqrt{\frac{\ln n}{n}}\bigg) \qquad \text{for every } \epsilon \in (0, \beps/2),
\end{equation}
where $\psiclust$ is the surrogate clustering defined in \cref{eq:levelset-clust-function}, $\psi_{\lambda}$ is the true level set clustering  defined in \Cref{ss:decision-theoretic-framework}, $D = \binom{n}{2}^{-1} \Lia$ is the loss from \Cref{thm:IABender-is-a-metric}, $\eta \mapsto \heps{\eta}$ is defined in \eqref{eq:cont-coeff}, $\rnld \doteq \rnldval$, and $v_d$ is the volume of the unit Euclidean ball in $d$ dimensions.
\label{thm:level-set-clustering-consistency}
\end{theorem}
\begin{proof} By \Cref{lem:unif-continuity}, the assumptions of \Cref{lem:level-set-containment} are satisfied. Thus, if we take $\eta = \epsilon \in (0, \beps/2)$ in \Cref{lem:level-set-containment}, we see that the condition
\begin{equation}
S_{(\lambda + 2\epsilon)} \subseteq \hat{S}_{\delta, \lambda}(f) \subseteq S_{(\lambda -  2\epsilon)}
\label{eq:cond-S1}
\end{equation}
holds uniformly over all $f: \cX \to \R$ with $\|f-f_0\|_\infty \leq \epsilon$ and   $\delta \in [\rnld, 2\heps{\epsilon}]$ with probability at least $1-1/n$. %
Henceforth, let us suppose that this event holds. Recall the true and estimated co-clustering relations $T$ and $\hT$ from \Cref{def:co-clustring}. By \Cref{cor:error-only-in-ciritcal-band}, for any $f, \delta$ such that $\|f-f_0\|_\infty \leq \epsilon$ and  $\delta \in [\rnld, 2\heps{\epsilon}]$,  we see that if $T(x,y) \neq \hT(x,y)$ for some $x, y \in \cX$, then one of $x$ or $y$ must lie in the region $\Delta(\epsilon) \doteq S_{(\lambda - 2\epsilon)} \setminus S_{(\lambda + 2\epsilon)} \subseteq \cX$.

Next we note that only a small fraction of observed data points $\cX_n$ lie in the region $\Delta(\epsilon) \subseteq \cX$. We use Hoeffding's inequality to establish this, noting that the event
$$
\hat{P}\{\Delta(\epsilon)\} - P_{f_0}\{\Delta(\epsilon)\} \leq \sqrt{\frac{\ln n}{n}}
$$
holds with probability at least $1-1/n^2$, where $\hat{P}(A) = \frac{1}{n}\sum_{i=1}^n \1{(x_i \in A)}$ denotes the empirical measure of any $A \subseteq \cX$, and $P_{f_0}\{\Delta(\epsilon)\} = \int_{\Delta(\epsilon)} f_0(x) dx$ denotes its population measure under the density $f_0$. Under \Cref{ass:mass-decay} we have $P_{f_0}\{\Delta(\epsilon)\} = \int_{\{x: |f_0(x)-\lambda| \leq 2\epsilon\}} f_0(x) dx \leq 2C\epsilon$ and thus:
\begin{equation}
	\label{eq:cond-S2}
\hat{P}\{\Delta(\epsilon)\} \leq 2C \epsilon + \sqrt{\frac{\ln n}{n}}.
\end{equation}

By the union bound, the events \eqref{eq:cond-S1} and \eqref{eq:cond-S2} will simultaneously hold with probability at least $1-\frac{n+1}{n^2}$. We henceforth assume that these  events hold. We are now ready to establish \eqref{eq:sup-statement}. Fix any $\epsilon \in (0, \beps/2)$, $\delta \in [\rnld, 2\heps{\epsilon}]$, and $f$ with $\|f-f_0\|_\infty \leq \epsilon$, and, for brevity, let $\bChf, \bC_0 \in \ClustSpace[\cX_n]$  denote $\psiclust[\delta](f) $ and $\psi_{\lambda}(f_0)$ respectively. Starting from the representation \eqref{eq:sum-representation} in the proof of \Cref{thm:IABender-is-a-metric}, we note that:
\begin{align*}
	&D(\bChf, \bCz) = \frac{1}{n(n-1)} \sum_{i \in [n]} \sum_{j \in [n] \setminus \{i\}} \phi_{x_i, x_j}(\bChf, \bCz) \\
	&= \frac{1}{n(n-1)}  \sum_{i \in [n]} \sum_{j \in [n] \setminus \{i\}} \bigg[ m \1{\{\bChf[f,A](x_i) \neq \bC_{0,A}(x_i)\}} +  m \1{\{\bChf[f,A](x_j) \neq \bC_{0,A}(x_j)\}}\\
	&\qquad\qquad +  a \1{\{\bChf[f,R](x_i,x_j) \neq \bC_{0,R}(x_i,x_j)\}}\1{\{\bChf[f,A](x_i)=\bC_{0,A}(x_i) = \bChf[f,A](x_j)=\bC_{0,A}(x_j)\}} \bigg] \\
	&= \frac{2m}{n} \sum_{i \in [n]} \1{\{\bChf[f,A](x_i) \neq \bC_{0,A}(x_i)\}}  + \frac{a}{n(n-1)} \sum_{i \in [n]} \sum_{j \in [n] \setminus \{i\}} \1{\{\bChf[f,R](x_i,x_j) \neq \bC_{0,R}(x_i,x_j)\}} \1{(x_i, x_j \in A_{f,\lambda} \cap A_{f_0,\lambda})} \\
	&= \frac{2m}{n} \sum_{i \in [n]} \1{(x_i \in A_{f,\lambda} \triangle A_{f_0,\lambda})} + \frac{a}{n(n-1)} \sum_{i \in [n]} \sum_{j \in [n] \setminus \{i\}} \1{\{\hT(x_i,x_j) \neq T(x_i,x_j)\}} \1{(x_i, x_j \in A_{f,\lambda} \cap A_{f_0,\lambda})}.
\end{align*}
Indeed, for the third equality, we have used that the last summand in the second equation (i.e.~the term in the third line) is non-zero only when $x_i, x_j \in A_{f,\lambda} \cap A_{f_0,\lambda}$, where $A_{f,\lambda} = \{x \in \cX_n : f(x) \geq \lambda \}$ and $A_{f_0,\lambda} = S_\lambda \cap \cX_n$ are the active sets of $\bChf$ and $\bCz$, respectively. For the subsequent equality, $\triangle$ symbolizes the symmetric difference between sets. Here we note  by definition that the co-clustering relation $\bC_{0,R}$ is the relation $T$ restricted to $\cX_n$.
Further, restricting to the points in $A_{f,\lambda}$, \Cref{lem:cciscc} shows that the co-clustering relation $\bChf[f,R]$ defined via $\psiclust[\delta](f)$ is equal to the co-clustering relation $\hT$ defined via the connected components of $\hat{S}_{\delta, \lambda}(f)$, i.e.~$\bChf[f,R](x,y) = \hT(x,y)$ for any $x, y \in A_{f,\lambda}$.

In order to complete the proof, we note the inequality $\1{\{T(x,y) \neq \hT(x,y)\}} \leq \1{\{x \in \Delta(\epsilon)\}} + \1{\{y \in \Delta(\epsilon)\}}$ and inclusion $A_{f,\lambda} \triangle A_{f_0,\lambda} \subseteq \Delta(\epsilon) \cap \cX_n$. While the inequality follows from the argument noted at the beginning of this proof, the inclusion follows since  $\1{\{f_0(x) \geq \lambda\}} = \1{\{f(x) \geq \lambda\}}$ whenever  $x \in \cX \setminus  \Delta(\epsilon)$ and  $\|f-f_0\|_\infty \leq 2\epsilon$. We thus obtain the bound: 
\begin{align*}
	D(\bChf, \bCz) \leq 2(m + a) \hat{P}\{\Delta(\epsilon)\} \leq 8 \bigg(C \epsilon + \sqrt{\frac{\ln n}{n}}\bigg).
\end{align*}
Since $\epsilon \in (0, \beps/2)$, $\delta \in [\rnld, \heps{\epsilon}]$, and $f$ with $\|f-f_0\|_\infty \leq \epsilon$ were arbitrary, we have shown that  \eqref{eq:sup-statement} holds.
\end{proof}

The proof of \Cref{lem:finite-sample-bound-for-assumption2} now follows as a special case of the above theorem. Indeed, suppose $f_0$ is an $\alpha$-H\"older continuous function so that $|f_0(x)-f_0(y)| \leq C_\alpha |x-y|^\alpha$ for some constant $C_\alpha > 0$. Then from \eqref{eq:cont-coeff} we find that $\heps{\eta} \geq (\eta/C_\alpha)^{1/\alpha}$ for any $\eta > 0$. Thus we can take $\epsilon = \max(\gamma, C_\alpha (2\delta)^\alpha)$ in \Cref{thm:level-set-clustering-consistency} to obtain \Cref{lem:finite-sample-bound-for-assumption2} with  $\bar{\gamma} = \bar{\epsilon}/2$, $\bar{\delta} = \frac{1}{2}(\bar{\epsilon}/2C_\alpha)^{1/\alpha}$ and $C_0 = 8(1+C)(1+C_\alpha)2^{\alpha}$.%

\subsection{Proof of \Cref{lem:sufficiency-of-data-adaptive-estimator}}
\label{ss:analysis-of-data-adaptive-estimator}

Here we show that our data-adaptive choice of $\delta = \hat{\delta}$ from \eqref{eq:data-adaptive-delta} based on the $k$-nearest neighbor distance  $\delta_k(x) \doteq \inf\{r  > 0 : |B(x,r) \cap \cX_n| \geq k \}$ will satisfy conditions of \Cref{lem:finite-sample-bound-for-assumption2}.
 
The argument of our proof starts with the following corollary of \Cref{lem:unif-convergence} used in \cite{chaudhuri2010rates} and later works like \cite{dasgupta2014optimal,jiang2017density} to study properties of $k$-nearest neighbor density estimates.

\begin{lemma}[Lemma 2 in \cite{dasgupta2014optimal}] 
	Suppose $P$ is a probability measure on $\R^d$ and $\hat{P}(A) = n^{-1}\sum_{i=1}^n \I{X_i \in A}$ is the empirical distribution based on $n$ i.i.d. samples $X_1, \ldots, X_n$ from $P$. Pick $0 < t < 1$ and let $C_{t,n}\doteq 16 \log(2/t) \sqrt{d \log n}$. If $k \geq d\log n$ then with probability at least $1-t$, for every ball $B \subseteq \R^d$ we have:
	\begin{align*}
		P(B) \geq C_{t,n} \frac{\sqrt{d \log n}}{n} \implies \hat{P}(B) > 0\\
		P(B) \geq k/n + C_{t,n} \frac{\sqrt{k}}{n} \implies \hat{P}(B) \geq \frac{k}{n} \text{ and } \\
		P(B) \leq k/n - C_{t,n} \frac{\sqrt{k}}{n} \implies \hat{P}(B) < \frac{k}{n}.
	\end{align*}
	\label{lem:uniform-k-nbds}
\end{lemma}

This leads to the following corollary for the behavior of our $k$ nearest neighbor distance based on data $\cX_n = \{x_1, \ldots, x_n\}$ drawn independently from the assumed distribution $P_0(A) \doteq \int_A f_0(x) dx$.
\begin{corollary} Suppose $k \geq (32)^2 d\log n$. Then with probability at least $1-2e^{-\frac{1}{32}\sqrt{\frac{k}{d \ln n}}}$ uniformly over $x \in \R^d$ and $r > 0$ we have:
	\begin{align*}
		\delta_k(x) \leq r \qquad&\text{ if } P_0(B(x, r))  \geq \frac{3k}{2n} \qquad\text{ and }\\
		\delta_k(x) \geq r \qquad&\text{ if } P_0(B(x, r)) \leq \frac{k}{2n}.
	\end{align*}
	\label{cor:knn-distance}
\end{corollary}
\begin{proof} We will take $t = 2e^{-\frac{1}{32}\sqrt{\frac{k}{d \ln n}}}$ in \Cref{lem:uniform-k-nbds} noting that $C_{t,n} = \frac{\sqrt{k}}{2}$. Thus \Cref{lem:uniform-k-nbds} shows that with probability $1-2e^{-\frac{1}{32}\sqrt{\frac{k}{d \ln n}}}$:  
\begin{align*}
	P_0(B) \geq \frac{3k}{2n} &\implies \hat{P}(B) \geq \frac{k}{n} \qquad\text{ and } \\
	P_0(B) \leq \frac{k}{2n} &\implies \hat{P}(B) < \frac{k}{n}.
\end{align*}
for each $x \in \cX$ and $r > 0$ with $B=B(x,r)$ and $\hat{P}(B) = \frac{|B \cap \cX_n|}{n}$. The proof is completed by noting that $\delta_k(x) = \inf \{r | \hat{P}(B(x,r)) \geq k/n \}$. Hence when $\hat{P}(B(x,r)) \geq \frac{k}{n}$ we must  have $\delta_k(x) \leq r$ and when $\hat{P}(B(x,r)) < \frac{k}{n}$ we must have $\delta_k(x) \geq r$. 
\end{proof}

Now we are ready to prove \Cref{lem:sufficiency-of-data-adaptive-estimator}. 

\begin{proof}[of \Cref{lem:sufficiency-of-data-adaptive-estimator}]
	
	By \Cref{ass:uniform-continuity} and \Cref{lem:unif-continuity}, $f_0$ is uniformly continuous and bounded. Thus there are constants $\bar{r} > 0$ and $M > 0$ such that  
	$$
	\sup_{x \in \cX} f_0(x) \leq M < \infty
	$$
	and 
	\begin{equation*}
		\sup_{\substack{x,y \in \cX \\ \|x-y\| \leq \bar{r}}} |f_0(x) - f_0(y)| \leq \lambda/4.
		\label{eq:ucont-bound}	
	\end{equation*}

	We will assume that $k \in [L \log n, n/L]$ for a suitably large constant $L > 0$ that is independent of $n$, which can be determined by examining the details of this proof. For example, we will assume that $L$ is large enough so that the event in \Cref{cor:knn-distance} holds with high probability.
	
	First let us show that $\hat{\delta}$ from \eqref{eq:data-adaptive-delta} will be less than $\bar{\delta}$. This will follow if for any $x_i \in \Activef[\hat{f}]$ we can show that $\delta_k(x_i) \leq r_0 = \min(\bar{r},\bar{\delta}/2)$. Indeed, since $\|\hat{f} - f_0\|_\infty \leq \lambda/2$, we must have $f_0(x_i) \geq \hat{f}(x_i) - \|f_0-\hat{f}\|_\infty \geq \lambda - \lambda/2 = \lambda/2$. Further, since $r_0 \leq \bar{r}$, we must have $\inf_{y \in B(x_i, r_0)} f_0(x) \geq \lambda/4$. This shows that
	$$
	P_0(B(x_i, r_0)) \geq \frac{\lambda}{4} v_d (r_0)^d \geq \frac{3}{2L} \geq \frac{3k}{2n}
	$$
	as long as $L \geq \frac{6}{\lambda v_d (r_0)^d}$. By \Cref{cor:knn-distance}, we must have $\delta_k(x_i) \leq r_0$ as required. 
	
	Next, let us show that $\hat{\delta} \geq \rnld$. Since $\hat{\delta} \geq \inf_{x_i \in \Activef[\hat{f}]} \delta_k(x_i)$  we will in-fact show that $\delta_k (x) \geq \rnld$ for any $x \in \cX$. Indeed, this will follow from \Cref{cor:knn-distance} once we can show that $P_0(B(x,\rnld)) \leq \frac{k}{2n}$. From the definition of $\rnld=\rnldval$ and the maximum value $M$ for $f_0$, we can note that
	$$
	P_0(B(x,\rnld)) \leq M v_d (\rnld)^d = 2^{d+4} \frac{M d \ln n}{\lambda n} \leq \frac{L \log n}{2n} \leq \frac{k}{2n}
	$$
	as long as $L \geq \frac{2^{d+5} M d}{\lambda}$.
	
\end{proof}

\section{Selecting the level $\lambda$}
\label{s:ballet-tuning-parameter}

The level set threshold $\lambda > 0$ is an important parameter for our analysis, and its choice needs to align well with the nature of clustering that we seek. In order to improve interpretation and  comparison of  level set clusters across different density models and clustering methods, following \cite{cuevas2001further,scrucca2016identifying}, we choose the fraction of noise points $\nu \in (0,1)$ rather than the actual density level $\lambda > 0$. Indeed, there is a one-to-one association between the two parameters when our true data generating distribution has a continuous density. Our experiments here demonstrate at least the following three possible ways to practically choose the level $\lambda$, depending on the goals of our clustering analysis.

\begin{enumerate}
    \item \emph{A known value of the level $\lambda$.} In our sky-survey analysis (\Cref{s:real-data-analysis}), the clustering of interest corresponded to an approximately known value of $\lambda$ motivated by scientific considerations. While our analysis in \Cref{s:real-data-analysis} directly used this threshold $\lambda$, we note in \Cref{s:persistent-clustering} that exploring the persistence of clusters across nearby choices of $\lambda$ may improve clustering accuracy. Indeed, even if the target level $\lambda$ is known exactly, the need for checking persistence of clusters  across nearby levels has also appeared in theoretical studies of level set clustering \citep{steinwart2011adaptive, sriperumbudur2012consistency, jiang2017density}.
    
    \item \emph{Finding the level $\lambda$  to separate a noisy background.} Often, our clusters of interest will be  connected components of regions with significantly large data density values, separated by noisy regions of comparatively much lower density values. For example, this is the case for our toy data example from \Cref{fig:mix-uniform-normal} and our illustrative data examples in \Cref{s:toy-data-analysis} if we are interested in the connected components of the obvious regions of non-negligible data density. (Note: depending on the density model used for the RNA-seq example, there is perhaps still some ambiguity about whether some observations bordering the major regions should be called noisy or not.) For these datasets, motivated by \texttt{DBSCAN} \citep{ester1996density, schubert2017dbscan}, we have found the following elbow heuristic useful: we sort the values  of the logarithm of the density estimates $\{\log \hat{f}(x_i)\}_{i=1}^n$  at the observations, and use the `kneedle' algorithm \citep{satopaa2011kneedle} to find a so-called elbow (or knee) in the plot of the logarithm of the density estimates  versus their ranks (see \Cref{fig:toy-challenge-elbow-plots,fig:illustration-elbow-selection}).  The intuition here is that a noisy observation $x_i$  will have a much smaller value of $\log \hat{f}(x_i)$ compared to a non-noisy observation $x_i$, and since the fraction of noisy observations is assumed to be small, this will reflect as an elbow in our plot. \Cref{fig:toy-challenge-compare-elbow} shows the \texttt{BALLET} clusters for the illustrative challenge datasets, based on the level selected using this elbow heuristic.
    
    \item \emph{Finding nuanced clusters by varying the density $\lambda$.} A careful choice of the level $\lambda$ can reveal more nuanced clustering structure in the data, whereby what seemed like a single cluster at a lower value of $\lambda$ can split into more than one cluster when a higher value of $\lambda$ is used. Indeed, this has motivated estimation of an entire hierarchical clustering tree as $\lambda$ varies \cite[see][and references therein]{wang2019dbscan,campello2020density,steinwart2023adaptive}, but additional strategies are then needed to obtain a flat clustering from the hierarchical clustering  tree \citep{campello2013framework,scrucca2016identifying}. Here, particularly for the RNA-seq dataset, we visualize the \texttt{BALLET} clusters for a range of values of $\nu \in \{5\%, 10\%, 15\%\}$ (see \Cref{fig:toy-challenge-dpmm,fig:toy-challenge-compare,fig:toy-challenge-compare-high}). Some of the clusters when $\nu = 5\%$ are seen to split further when we choose $\nu = 10\%$. In \Cref{fig:tsne-persistent-clustering}, we show the persistent clusters (\Cref{s:persistent-clustering}) for this dataset obtained by post-processing the results corresponding to the noise levels $\nu \in \{5\%,10\%,15\%\}$. We  note that a related notion of post-processing of the output of level set clustering methods across different levels has been explored in \cite{steinwart2011adaptive,sriperumbudur2012consistency,steinwart2015fullyAdaptive} to consistently estimate the smallest level where the true density has more than one connected component, but their aim is different from what we need here.
\end{enumerate}

\begin{figure}
    \centering
    \includegraphics[width=0.85\textwidth]{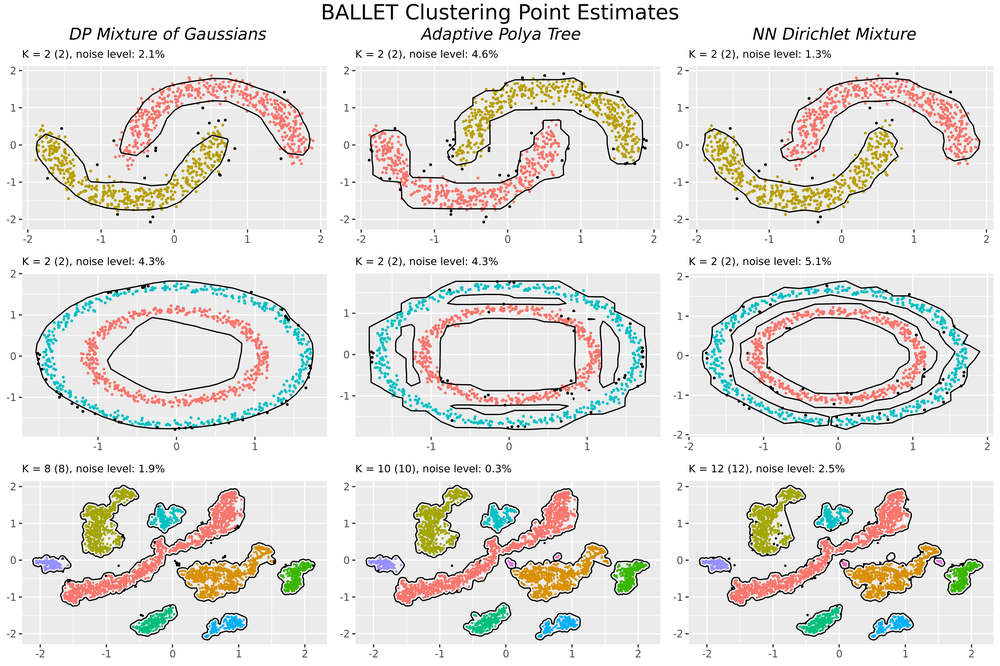}
    \caption{\texttt{BALLET} clustering point estimates obtained under the three different density models shown in \Cref{fig:toy-challenge-densities} with the level chosen using the elbow heuristic (see \Cref{fig:toy-challenge-elbow-plots}).}
    \label{fig:toy-challenge-compare-elbow}
\end{figure}

\begin{figure}
    \centering
    \includegraphics[width=0.85\textwidth]{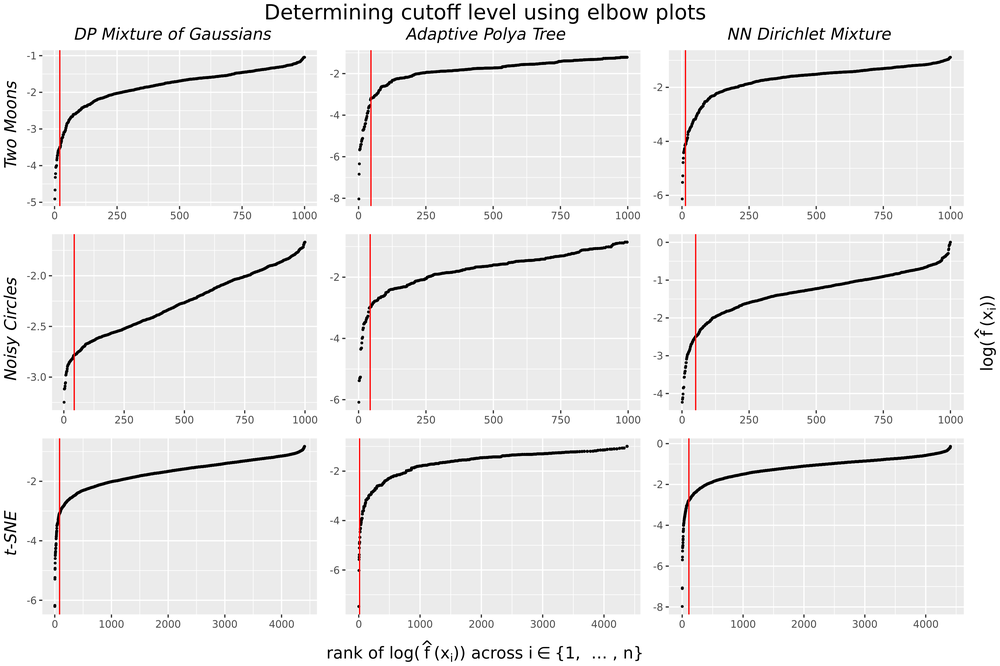}
    \caption{Elbow heuristic to choose the level for the illustrative challenge datasets across density models (\Cref{fig:toy-challenge-data})  based on sorting the log of posterior median density $\hat{f}$ at observations $\{x_i\}_{i=1}^n$ for each dataset and model pair. The elbow value (red vertical line) was automatically determined using the `kneedle' algorithm of \cite{satopaa2011kneedle}. \Cref{fig:toy-challenge-compare-elbow} shows the corresponding \texttt{BALLET} clusters.}
    \label{fig:toy-challenge-elbow-plots}
\end{figure}

\begin{figure}
    \centering
    \includegraphics[width=0.5\linewidth]{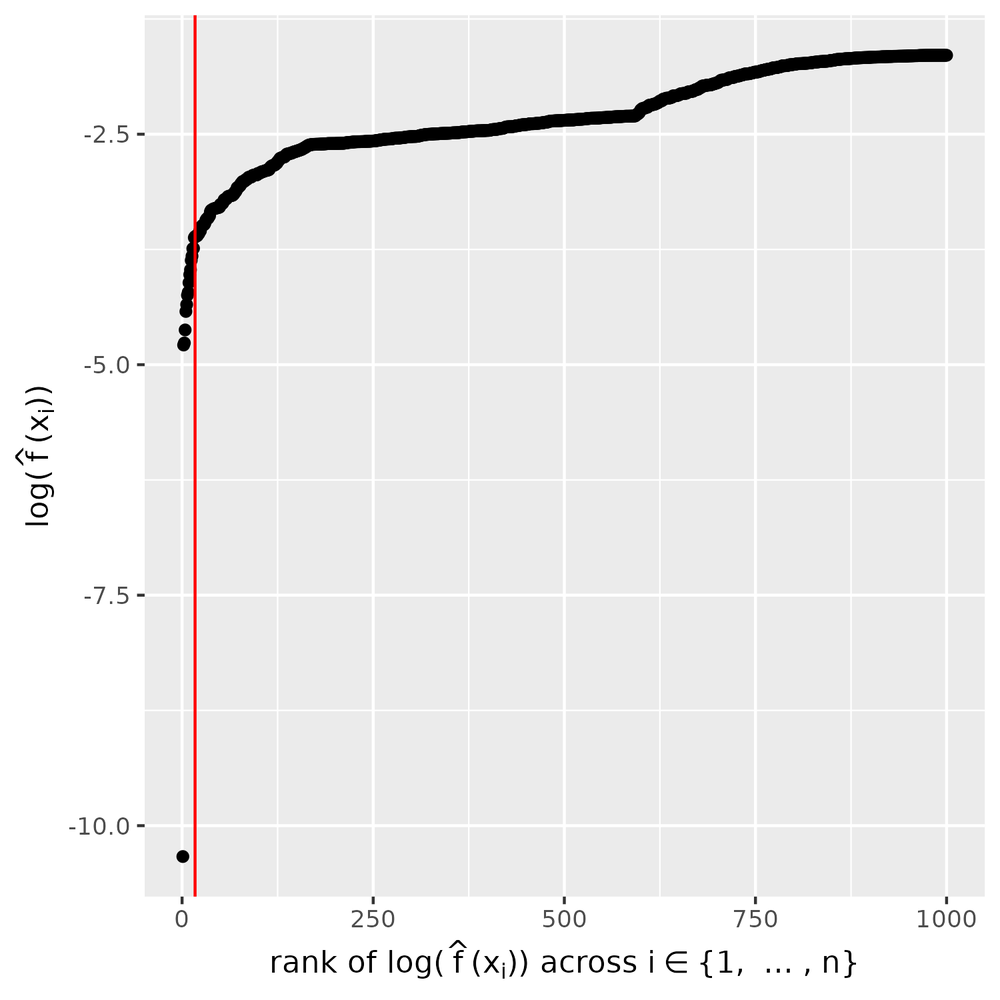}
    \caption{Elbow plot illustrating our selection of the level in \Cref{fig:mix-uniform-normal} based on sorting the log of posterior median density $\hat{f}$ evaluated at the observations $\{x_i\}_{i=1}^n$. The elbow value (red vertical line) was automatically determined using the `kneedle' algorithm of \cite{satopaa2011kneedle}.}
    \label{fig:illustration-elbow-selection}
\end{figure}

\section{Persistent Clustering}
\label{s:persistent-clustering}

\subsection{Motivation: robustness to the choice of level $\lambda$}

A key problem with level set clustering is that we may not exactly know the level \citep{campello2020density} or, worse yet, that our results can be sensitive to the exact  level that we choose for our analysis. Here we describe how to summarize clustering results  across multiple values of the level by visualizing a cluster tree \citep{clustree}, and reduce our  sensitivity to any single choice of the level by identifying clusters that remain active or ``persistent'' across all the levels in the tree.

As described in \Cref{s:discussion}, we expect the level set clusters of our Edinburgh-Durham Southern Galaxy Catalogue data to be sensitive to the exact value of the level $\lambda = (1+c) \bar{f}$,  determined by the scientific constant $c$. Since $c$ is believed to be around one \citep{jang2006nonparametric}, our preliminary analysis of this data in  \Cref{s:real-data-analysis} proceeded with the assumption that $\lambda = 2\bar{f}$,  or equivalently that $c=1$. Here we summarize our results from computing the \texttt{BALLET} clusters at various density levels  corresponding to the values $c \in \{.8, .9, \ldots, 1.2\}$.

\subsection{Visualizing the cluster tree}

\begin{figure}
    \centering
    \includegraphics[width=\textwidth]{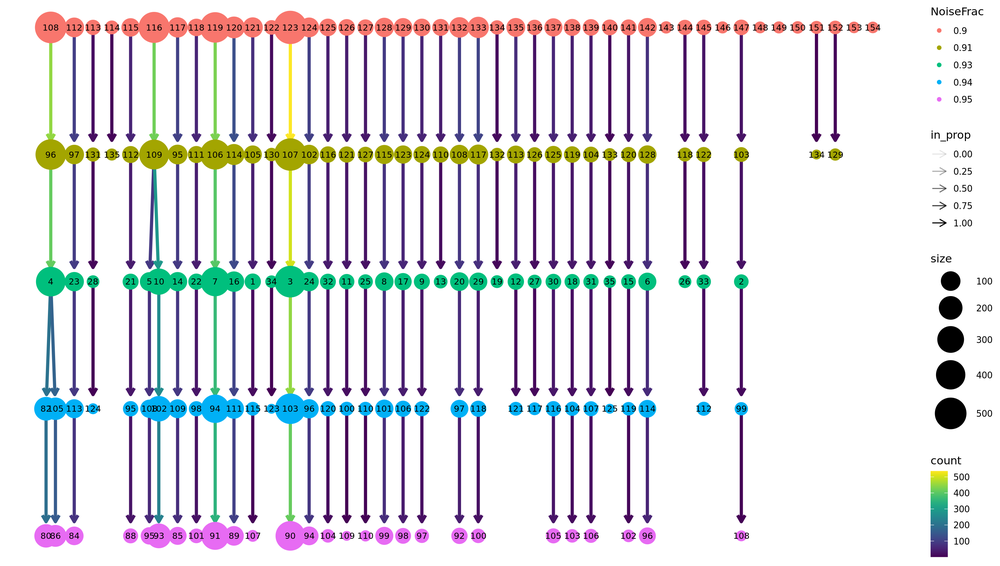}
    \caption{The \texttt{BALLET} cluster tree \citep{clustree} for the Edinburgh-Durham Southern Galaxy Catalogue data across multiple density levels corresponding to $c \in \{.8, .9, \ldots, 1.2\}$. The nodes in each row are the \texttt{BALLET} clusters for the  fixed level $\lambda=(1+c)\bar{f}$, where $c$ increases as we go down the tree. An edge between nodes in two successive levels indicates an overlap between the two corresponding clusters.  While most clusters at the top level  ($c=0.8$) have a unique child in the tree at each lower level (as $c$ increases), some clusters at the top level split into multiple children or did not have any descendent in the bottom levels. For each cluster at the bottom level, the \emph{persistent clustering} algorithm finds its topmost ascendant in the tree below any (potential) split. }
    \label{fig:edsgc-cluster-tree}
\end{figure}

It is well known \citep{hartigan1975clustering,campello2020density,menardi2016modal} that the level set clusters across different levels of the same density are nested in a way that can be organized into a tree. In particular, given two clusters from two different levels of the same density, it is the case that either both the clusters are disjoint, or one of the clusters is contained inside the other. 

We empirically found that our \texttt{BALLET} estimates across various levels could similarly be organized into a tree. We visualized this tree in \Cref{fig:edsgc-cluster-tree} by modifying code for the  \texttt{clustree} package  in R \citep{clustree}.  We see that while \texttt{BALLET} found 44 clusters at the lower level ($c=.8$), it only found 27 clusters at the higher level ($c=1.2$), indicating that more than a third of the lower level clusters disappear as the choice of the level is slightly increased.  Further, in this process, two of the lower level clusters are also seen to split into two clusters each.

\subsection{Persistent Clustering}

Given the sensitivity of level set clusters to the choice of level, we now describe a simple algorithm that processes the cluster tree to extract clusters that are active (persistent) across all the levels in the tree. Some clusters can split into multiple sub-clusters as we increase our level in the cluster tree (i.e. go down the tree). In such cases we will only focus on the cluster's descendants at the time of the last split.

Suppose a cluster tree like \Cref{fig:edsgc-cluster-tree} is given. Starting from each cluster at the bottom row of the tree,  the \emph{Persistent Clustering} algorithm involves walking up the tree until we (i) either hit the top row of the tree, or (ii) hit a node whose parent has more than one child. The collection of clusters corresponding to the final nodes obtained from these runs will be called \emph{persistent clusters}. 

\texttt{BALLET} persistent clusters for the Edinburgh-Durham Southern Galaxy Catalogue data are shown in \Cref{fig:edsgc-ballet-persistent}.
\Cref{tab:edsgc-catalogues-vs-persistent} compares the performance of \texttt{BALLET} persistent clusters to those at the fixed level ($c=1$). We find that persistent clustering improves specificity on both the Abell and EDCCI catalogs without loss in sensitivity. 

\begin{figure}
    \centering
    \includegraphics[width=1\linewidth]{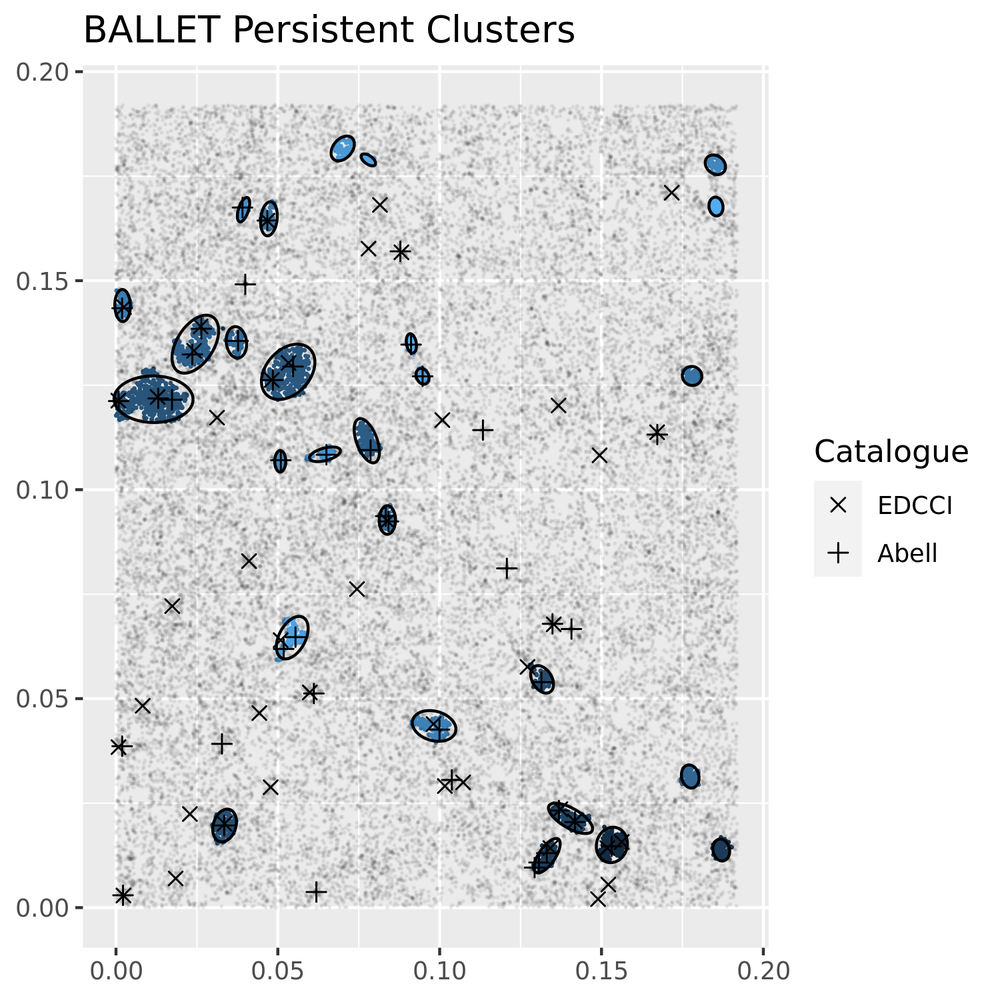}
    \caption{The \texttt{BALLET} persistent clustering estimate for the Edinburgh-Durham Southern Galaxy Catalogue data across levels $c \in \{.8, \ldots, 1.2\}$.}
    \label{fig:edsgc-ballet-persistent}
\end{figure}

\begin{table}
\centering
\begin{tabular}{|c|c|c|}
\hline
 & BALLET (persistent) & BALLET ($c=1$) \\
\hline
Sensitivity (EDCCI)  & 0.69 & 0.67 \\
Specificity (EDCCI)  & 0.74 & 0.69 \\
Exact Match (EDCCI)  & 0.48 & 0.51\\
\hline
Sensitivity (Abell)  & 0.40 & 0.40 \\
Specificity (Abell)  & 0.44 & 0.40 \\
Exact Match (Abell)  & 0.26 & 0.26\\
\hline
\end{tabular}
    \caption{Comparing results from \texttt{BALLET} persistent clusters across $c \in \{.8, \ldots, 1.2\}$ to the \texttt{BALLET} point estimate at $c=1$.  Persistent clustering improves the specificity for both the catalogues without losing sensitivity.}
    \label{tab:edsgc-catalogues-vs-persistent}.
\end{table}

While we have motivated the idea of persistent clustering by the practical concern of robustness, the idea of obtaining a single clustering by cutting the cluster tree at locally adaptive levels has been explored before in the algorithmic level set clustering literature \citep{campello2020density,campello2015hierarchical}. Such methods are useful when we want to recover density-based clusters that can only be separated by considering differing values of the levels (\Cref{fig:level-set-pathological}).

\begin{figure}
	\centering
	\includegraphics[width=0.6\textwidth]{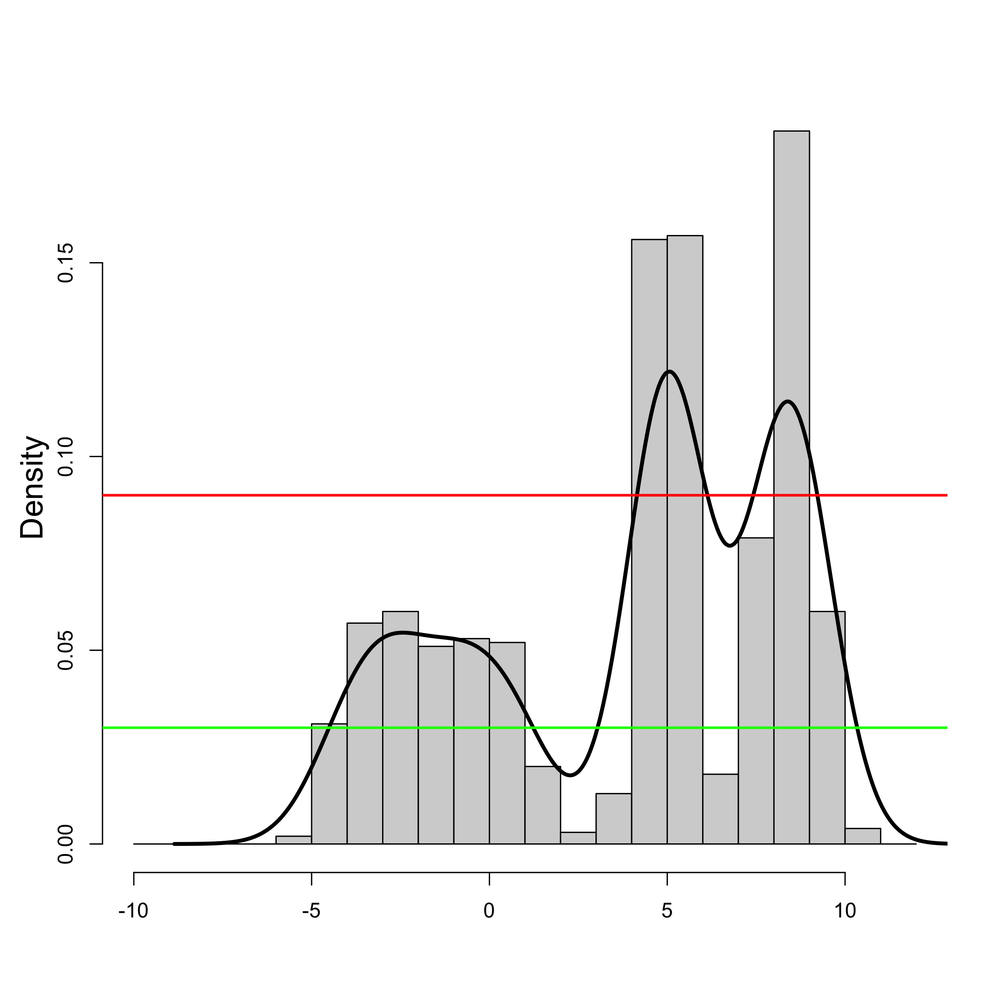}
	\caption{An example of a situation in which we might want to cluster data according to locally adaptive levels. }
	\label{fig:level-set-pathological}
\end{figure}

\section{Other clustering methods}
\label{s:other-clustering-methods}

There are a wide variety of clustering algorithms (e.g.~\cite{wani2024comprehensive,xu2015comprehensive}) because no single notion of clustering is  useful across all applications \citep{von2012clustering,hennig2015true}. Here our focus has been on Bayesian statistical approaches to clustering \citep{wade2023bayesian} that account for sampling variability within the data and have the ability to use application-dependent prior information. In principle, our  density-based clustering framework allows for the combination of statistical inference with any flexible clustering notion required by the application (provided the clustering $\psi(f)$ can be computed given the population density $f$).

As an example of our framework, our \texttt{BALLET} methodology shows the ability to find arbitrary shaped clusters in comparison to Gaussian mixture models, which have been predominantly  used for Bayesian clustering \citep{wade2023bayesian}. While additional algorithmic approaches like spectral and hierarchical clustering \citep{wani2024comprehensive} also have the ability to find arbitrary shaped clusters, their clustering can be sensitive to the presence of even a few noisy observations. This is seen in \Cref{fig:two-circles-with-outliers} with the addition of six new observations to a sample of $n=600$ observations from one of the datasets considered in \Cref{s:app-toy-challenge}.

\begin{figure}
    \centering
    \includegraphics[width=.9\linewidth]{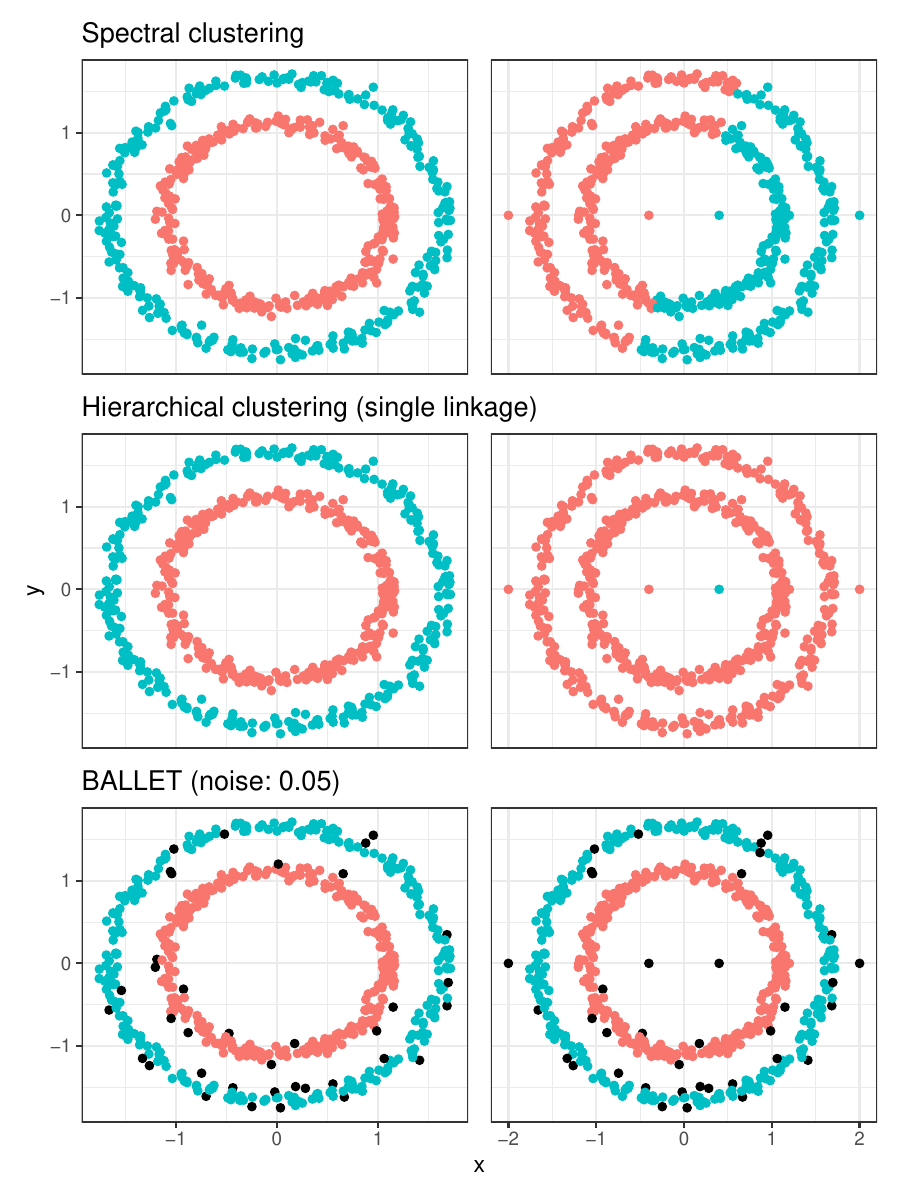}
    \caption{Spectral and hierarchical clustering results change when we add six equally spaced observations on the $y=0$ line to $n=600$ observations sampled from one of the datasets in \Cref{s:app-toy-challenge} (left: original clustering, right: clustering with six observations added). \texttt{BALLET} clustering based on $\nu=5\%$ noise points is majorly unaffected here as most of these additional points are declared to be noise.}
    \label{fig:two-circles-with-outliers}
\end{figure}

\clearpage
\putbib[bibliography]
\end{bibunit}

\end{document}